\documentclass[11pt,letterpaper]{article}

\usepackage{palatino}
\usepackage{mathpazo}
\usepackage{stmaryrd}
\usepackage{amsmath,amsthm,bm}
\usepackage{amssymb,latexsym,graphicx}
\usepackage{mathtools}
\usepackage{url}

\DeclarePairedDelimiter{\abs}{\lvert}{\rvert}
\DeclarePairedDelimiter{\norm}{\lVert}{\rVert}

\DeclarePairedDelimiter{\ceil}{\lceil}{\rceil}
\newcommand{\triplevert}{|\mkern-2mu|\mkern-2mu|}

\newcommand{\ket}[1]{|#1\rangle}
\newcommand{\bra}[1]{\langle#1|}

\newcommand{\Tr}{\mbox{\rm Tr}}
\newcommand{\Trho}{\mbox{\rm Tr}_\rho}
\def\calH{{\cal H}}

\DeclareMathOperator{\poly}{poly}
\newcommand{\beq}{\begin{equation}}
\newcommand{\eeq}{\end{equation}}

\newcommand{\Es}[1]{\textsc{E}_{#1}}

\newcommand{\classfont}{\mathrm}

\newcommand{\F}{\mathbb{F}}
\newcommand{\NP}{\classfont{NP}}

\newcommand{\PSPACE}{\classfont{PSPACE}}
\newcommand{\EXP}{\classfont{EXP}}
\newcommand{\NEXP}{\classfont{NEXP}}

\newcommand{\IP}{\classfont{IP}}
\newcommand{\MIP}{\classfont{MIP}}
\newcommand{\MIPstar}{\classfont{MIP}^*}
\newcommand{\MIPer}{\classfont{MIP}^{\classfont{er}}}
\newcommand{\MIPns}{\classfont{MIP}^{\classfont{ns}}}
\newcommand{\QIP}{\classfont{QIP}}

\newcommand{\QMS}{\classfont{QMIP}^*}

\newcommand{\parityMIP}{\oplus\classfont{MIP}}
\newcommand{\parityMIPstar}{\parityMIP^{\ast}}

\newcommand{\PCP}{\textsc{PCP}}

\newcommand{\cons}{\textsc{con}}
\newcommand{\inc}{\textsc{inc}}

\newcommand{\ML}{\mathrm{ML}}

\newcommand{\calA}{\mathcal{A}}
\newcommand{\calB}{\mathcal{B}}
\newcommand{\calC}{\mathcal{C}}

\newcommand{\calK}{\mathcal{K}}
\newcommand{\calP}{\mathcal{P}}

\newcommand{\Fp}{\mathbb{F}}
\newcommand{\Ftwo}{\mathbb{F}_2}
\newcommand{\N}{\mathbb{N}}

\newcommand{\C}{\mathbb{C}}
\newcommand{\vx}{{\vct{x}}}
\newcommand{\vy}{{\vct{y}}}

\DeclareMathOperator{\Id}{Id}
\newcommand{\eps}{\varepsilon}

\newcommand{\vct}[1]{\bm{#1}}

\newtheoremstyle{problem}
  {}{}{\setlength{\parskip}{4pt}\setlength{\parindent}{0pt}}{}{\bfseries}{.}{\newline}
  {\thmname{#1}\thmnumber{~#2}\thmnote{: #3}}
\newtheorem{theorem}{Theorem}

\newtheorem{proposition}[theorem]{Proposition}
\newtheorem{lemma}[theorem]{Lemma}

\newtheorem{claim}[theorem]{Claim}
\newtheorem{fact}[theorem]{Fact}
\newtheorem{corollary}[theorem]{Corollary}
\newtheorem{definition}[theorem]{Definition}
\theoremstyle{problem}
\newtheorem{problem}{Problem}

\begin{document}

\title{A multi-prover interactive proof for NEXP \\ sound against entangled provers}
\author{Tsuyoshi Ito\footnote{NEC Laboratories America, Inc., USA\@.
Email: tsuyoshi@nec-labs.com. Supported in part by ARO/NSA grant W911NF-09-1-0569.
Also supported by grants from NSERC, CIFAR, QuantumWorks, MITACS, CFI, and ORF received while this author was a postdoctoral fellow at the Institute for Quantum Computing and David R.~Cheriton School of Computer Science, University of Waterloo, Canada.}
\and Thomas Vidick\footnote{Computer
Science and Artificial Intelligence Laboratory, Massachusetts
Institute of Technology, USA\@. Email: vidick@csail.mit.edu. Supported by the National Science Foundation under Grant No. 0844626.
Part of this work was done while this author was a graduate student in the Computer Science department at UC Berkeley, as well as during visits to the Perimeter Institute in Waterloo, Canada and NEC Labs America.
}}
\date{}
\maketitle

\begin{abstract}
We prove a strong limitation on the ability of entangled provers to collude in a multiplayer game. Our main result is the first nontrivial lower bound on the class~$\MIPstar$ of languages having multi-prover interactive proofs with entangled provers; namely~$\MIPstar$ contains~$\NEXP$, the class of languages decidable in non-deterministic exponential time. While Babai, Fortnow, and Lund (\emph{Computational Complexity}\ 1991) proved the celebrated equality~$\MIP = \NEXP$ in the absence of entanglement, ever since the introduction of the class~$\MIPstar$ it was open whether shared entanglement between the provers could weaken or strengthen the computational power of multi-prover interactive proofs. Our result shows that it does not weaken their computational power: $\MIP\subseteq\MIPstar$. 

At the heart of our result is a proof that Babai, Fortnow, and Lund's \emph{multilinearity test} is sound even in the presence of entanglement between the provers, and our analysis of this test could be of independent interest. As a byproduct we show that the correlations produced by any entangled strategy which succeeds in the multilinearity test with high probability can always be closely approximated using shared randomness alone.
\end{abstract}

\section{Introduction}


Multiprover interactive proof systems~\cite{BenGolKilWig88STOC} are at the heart of much of the recent history of complexity theory, and the celebrated characterization $\MIP = \NEXP$~\cite{BabForLun91CC} is one of the cornerstones on which the PCP theorem~\cite{AroSaf98JACM,AroLunMotSudSze98JACM} was built. While the key assumption on the multiple provers in an interactive proof system is that they are not allowed to communicate, traditionally this has been taken to mean that their only distributed resource was shared randomness. In a quantum universe, however, it is natural to relax this assumption and allow the provers to share \emph{entanglement}. While still not allowing them to communicate, this increases their ability to collude against the verifier by exploiting the nonlocal correlations allowed by entanglement. The corresponding complexity class $\MIPstar$ was introduced in~\cite{CHTW04}, raising a fundamental question: \emph{what is the computational complexity of entangled provers?}

Even before their modern re-formulation in the language of multiplayer games, starting with the work of Bell in the 1960s~\cite{Bell:64a} the strength of the nonlocal correlations that could be obtained from performing local measurements on entangled particles has been intensely investigated through the use of \emph{Bell inequalities} (upper bounds on the strength of classical correlations) and \emph{Tsirelson inequalities} (upper bounds on the strength of quantum correlations). Games, or proof systems, generalize this setup by introducing an additional layer of \emph{interaction}: in this new context, we think of the experimenter (the verifier) as interacting with the physical devices (the provers) through the specific choice of settings (questions) that he makes, and the outcomes (answers) that he observes. The arbitrary state and measurements that are actually made inside the devices are reflected in the provers' freedom in choosing their strategy. The fundamental observation that quantum mechanics violates certain Bell inequalities translates into the fact that there exists interactive proof systems in which entangled provers can have a strictly higher success probability than could any classical, non-entangled provers. 

A dramatic demonstration of this possibility is given by the Magic Square game~\cite{MerminMS,PeresMS}, a simple one-round game for which the maximum success probability of classical provers is $8/9$, but there exists a \emph{perfect} winning strategy for entangled provers. Cleve, H\o yer, Toner, and Watrous~\cite{CHTW04} were the first to draw complexity-theoretic consequences from such \emph{non-local} properties of entanglement. They study the class $\parityMIP$ of languages having two-prover interactive proofs in which there is a single round of interaction, each of the provers is restricted to answering a single bit, and the verifier only bases his accept/reject decision on the parity of the two bits that he received. While it follows from work of H\aa stad~\cite{Hastad01} that this class equals $\NEXP$ (and is thus as powerful as the whole of $\MIP$) for an appropriate setting of completeness and soundness parameters, Cleve et al.\ show that the corresponding entangled-prover class $\parityMIPstar$ \emph{collapses} to $\EXP$ for any choice of completeness and soundness parameters that are separated by an inverse polynomial gap.\footnote{This was later improved~\cite{Weh06STACS} to the inclusion of $\parityMIPstar$ in the class of two-message single-prover interactive proofs $\QIP(2)\subseteq \PSPACE$~\cite{JUW09}.} 

\medskip

Despite intense efforts, for a long time little more was known, and prior to our work the best lower bound on $\MIPstar$ resulted from the trivial observation that multiple entangled provers are at least as powerful as a single prover, hence $\IP = \PSPACE \subseteq \MIPstar$, where the first equality is due to~\cite{LunForKarNis92JACM,Sha92}.\footnote{It was recently shown that quantum messages are no more powerful than classical messages in \emph{single-prover} interactive proof systems~\cite{JJUW11}: $\QIP=\PSPACE$. That result, however, has no direct relationship with our work: in our setting the messages remain classical; rather the ``quantumness'' manifests itself in the presence of \emph{entanglement} between the provers, which is a notion that only arises when more than one prover is present.} The main difficulty in improving this trivial lower bound is the following: while the PCP theorem gives us a variety of two-prover interactive proof systems for $\NEXP$-complete problems, there is no a priori reason (see e.g.\ the Magic Square game, which has a similar structure to that of basic proof systems for MAX-$3$-XOR, or the aforementioned collapse of $\parityMIPstar$) that they should remain sound in the presence of entanglement.
Indeed, if one considers provers allowed to reproduce any distribution that is \emph{no-signaling},\footnote{A collection of distributions on the provers' answers, one for every tuple of questions, is no-signaling if, for any such distribution, its marginal on any subset of the provers is independent of the questions to the remaining provers.} then it follows from a linear-programming formulation\footnote{This formulation was first observed by Daniel Preda.} of the problem that the corresponding class $\MIPns \subseteq \EXP$ --- in fact, for the case of two-prover single-round proof systems it was even shown in~\cite{Ito10} that $\MIPns(2,1)=\PSPACE$. Entanglement, however, does not allow the provers to reproduce the full set of no-signaling strategies, and these results leave the complexity of the class $\MIPstar$ completely open. 

The fact that entanglement, as a shared resource, is poorly understood is also reflected in the complete absence of reasonable \emph{upper bounds} on the complexity class $\MIPstar$: while the inclusion $\MIP\subseteq\NEXP$ is straightforward, we do not know of any limits on the \emph{dimension} of entanglement that may be useful to the provers in a given interactive proof system, and as a result their maximum success probability is not even known to be computable (see~\cite{SW08,DLTW08,NPA08NJP} for more on this aspect).

Since existing protocols may no longer be sound in the presence of entanglement between the provers, previous work has focused on finding ways to \emph{modify} a given protocol in a way that would make it \emph{entanglement resistant}; that is, honest provers (in the case of a YES-instance) can convince the verifier without shared entanglement while dishonest provers (in the case of a NO-instance) cannot convince the verifier with high probability even with shared entanglement. This was the route taken in~\cite{KKMTV11,yao:tsirelson,ItoKM09}, which introduced techniques to limit the provers' use of their entanglement. They proved non-trivial lower bounds on variants of the class $\MIPstar$, but with error bounds that are weaker than the standard definitions allow for. These relatively weak bounds came as a result of the ``rounding'' technique developed in these works: by adding additional constraints to the protocol, one ensures that optimal entangled strategies are in a sense close to classical, un-entangled strategies. This closeness, however, was shown using a rounding procedure that had a certain ``local'' flavor, inducing a large loss in the quality of the approximation.\footnote{See the ``almost-commuting implies nearly-commuting'' conjecture in~\cite{KKMTV11} for more on this aspect.}

In addition,~\cite{ItoKM09}, based on~\cite{KKMTV11}, showed that $\PSPACE$ has two-prover \emph{one-round} interactive proofs with entangled provers, with perfect completeness and exponentially small soundness error. Prior to our work, this was the best lower bound known on single-round multi-prover interactive proof systems with entanglement. 

\paragraph{Additional related work.} Given the apparent difficulty of proving good lower bounds on the power of multi-prover interactive proof systems with entangled provers, researchers have studied a variety of related models. Maybe the most natural 
extension of $\MIPstar$ consists in giving the verifier more power by allowing him to run in quantum polynomial-time, and exchange quantum messages with the provers. The resulting class is called $\QMS$ (the $Q$ stands for ``quantum verifier'', while the $^*$ stands for ``entangled provers''), and it was formally introduced in~\cite{KobMat03JCSS}, where it was shown that $\QMS$ contains $\MIPstar$ (indeed, the verifier can always force classical communication by systematically measuring the provers' answers in the computational basis).
Recently Reichardt et~al.~\cite{ReiUngVaz-1209.0448} showed that~$\QMS=\MIPstar$ (the possibility of which had been suggested earlier in~\cite{BFK10}).
Ben-Or et al.~\cite{BenHP08} introduced a model in which the verifier is quantum and the provers are allowed communication but no entanglement, and showed that the resulting class contains $\NEXP$. Other works attempt to characterize the power of $\MIPstar$ systems using \emph{tensor norms}~\cite{RapTS07,JungePPVW10}; so far however such norms have either led to computable, but very imprecise, approximations, or have remained (to the best of our knowledge) intractable.   

\subsection{Results}

Let $\MIPstar(k,m,c,s)$ be the class of languages that can be decided by an $m$-round interactive proof system with $k$ (possibly entangled) provers and with completeness $c$ and soundness error $s$.\footnote{We refer to Section~\ref{sec:prelim-mip} for a more complete definition of the class $\MIPstar$.}
Our main result is the following.

\begin{theorem}\label{thm:nexp_main}
All languages in $\NEXP$ have a three-prover poly-round interactive proof system with perfect completeness and exponentially small soundness error against entangled provers.  That is, for every~$q\in\poly$, it holds that
\[
  \NEXP \subseteq \MIPstar(3,\poly,1,2^{-q}).
\]
\end{theorem}

Theorem~\ref{thm:nexp_main} resolves a long-standing open question~\cite{KobMat03JCSS}, showing that entanglement does not weaken the power of multi-prover interactive proof systems: together with the inclusion $\MIP\subseteq \NEXP$, it implies that $\MIP \subseteq \MIPstar$. We note that the proof system in Theorem~\ref{thm:nexp_main} does not require honest provers to use any entanglement in order to achieve perfect completeness in the case of a YES-instance.
In other words, if we denote by~$\MIPer$ the class of languages having entanglement resistant multi-prover interactive proof systems with bounded error, our proof of~Theorem~\ref{thm:nexp_main} shows that~$\NEXP\subseteq\MIPer$. Because~$\MIPer\subseteq\MIP$ by definition, this implies $\MIPer = \NEXP$.

The interactive proof system used in the proof of Theorem~\ref{thm:nexp_main} uses three provers
and a polynomial number of rounds of interaction.
We do not know if the number of provers can be reduced; however if one is willing to increase it by one then the amount of interaction required can be reduced to a single round, i.e.\ one message from the verifier to each prover, and one message from each prover to the verifier. Indeed, our proof system has the additional property of being  \emph{non-adaptive}: the verifier can select his questions for all the rounds before interacting with any of the provers. It is shown in~\cite{Ito-parallel} that a non-adaptive entanglement-resistant protocol may be parallelized to a single round of interaction at the cost of adding an extra prover. Applying this result to Theorem~\ref{thm:nexp_main} gives the following corollary.

\begin{corollary}\label{cor:mip-one-round}
  All languages in $\NEXP$ have a four-prover one-round interactive proof system with perfect completeness and soundness error against entangled provers bounded away from $1$ by an inverse polynomial, that is: 
  \[
    \NEXP \subseteq \MIPstar(4,1,1,1-1/\poly).
  \]
\end{corollary}

Prior results on the complexity of multi-prover interactive proofs with entangled provers have often been stated using the languages of \emph{games}~\cite{CHTW04,KKMTV11,KRT10}. The main difference, in terms of computational complexity, is in the way the input size is measured. In the case of games the input is an explicit description of the game, including a list of all possible questions and valid answers, while in the setting of proof systems the messages may be described implicitly: it is their \emph{length} that is polynomial in the input size. 

Because of this difference in scaling, our results do not immediately imply any NP-hardness result in the setting of multi-player games with entangled players. Nevertheless, by adapting the proof of Theorem~\ref{thm:nexp_main} and using the PCP theorem one can show the following. There is a constant $\kappa>1$ and a procedure that, given as input an arbitrary $3$-SAT formula with $n$ variables and $m = \poly(n)$ clauses, runs in time  $2^{O(\log^\kappa n)}$ and produces an explicit description of a three-player game of size $S=2^{O(\log^\kappa n)}$ (i.e. the number of rounds of interaction and the total number of questions and answers that can be sent and received is at most $S$). The game has the property that, if the $3$-SAT formula was satisfiable, then there is a perfect strategy for the players, which does not require any entanglement. If, however, the $3$-SAT formula was not satisfiable, then there is no strategy for the players, even using entanglement, that succeeds with probability greater than $1/2$. 

If one could show the above with constant $\kappa=1$ then it would follow that finding a constant-factor approximation to the maximum success probability of three entangled players in a game with polynomially many rounds and questions is $\NP$-hard; our result is limited to obtaining some possibly large $\kappa>1$. The main point, however, is that the hardness of approximation is up to \emph{constant} factors. This is in contrast to all previous results which were limited to hardness of approximation up to factors approaching $1$ very quickly as the input size grew (even after arbitrary sequential or even parallel repetition).\footnote{Cleve, Gavinsky, and Jain~\cite{CleGavJai09QIC} obtained a constant-factor hardness result for games with constant answer size, but in which the number of questions sent by the verifier is \emph{exponential}.}

At the heart of the proof of Theorem~\ref{thm:nexp_main} is a soundness analysis of Babai, Fortnow, and Lund's \emph{multilinearity test} in the presence of entanglement between the provers: we show that it is in a sense ``immune'' to the strong non-local correlations that entangled provers may in general afford. We believe that this analysis should be of wider interest, and we explain the test and the main ideas behind its analysis in the presence of entanglement in Section~\ref{sec:multilin-game} below. We first briefly outline the overall structure of our proof system in Section~\ref{sec:proof-outline}. It is very similar to the one introduced by Babai, Fortnow, and Lund~\cite{BabForLun91CC} to prove $\NEXP\subseteq \MIP$; our contribution consists in proving its soundness against entangled provers.

\subsection{Proof outline}\label{sec:proof-outline}

Our interactive proof system verifies membership in a specific $\NEXP$-complete language, \emph{succinct $3$-colorability} (see Problems~1 and~2 in Section~\ref{sec:nexpproblems} for a definition). We give a three-prover, poly-round interactive protocol for it that has perfect completeness and soundness error bounded away from $1$ by an inverse-polynomial in the input size. (Theorem~\ref{thm:nexp_main} is obtained by sequentially repeating this interactive proof system.)
We emphasize that the proof system we use is not new, as it is essentially the same as the one introduced in~\cite{BabForLun91CC}.
We nevertheless outline it because there is a small difference in how the ``oracle'' in~\cite{BabForLun91CC} is simulated by provers, which is the reason our protocol, unlike the one in~\cite{BabForLun91CC}, requires more than two provers.

Simplifying a little bit (we refer the reader to Section~\ref{sec:protocol} for details),
the verifier in our protocol is given as input two integers $n,p$ in unary (think of $p$ as much larger than $n$, but still polynomial),
a description of a finite field $\Fp$ of size~$p$, and a low-degree polynomial $f:(\Fp^n)^2 \times (\Fp)^2 \to \Fp$.
His goal is to verify whether there exists a multilinear function $g:\Fp^n\to\Fp$ such that $f(\vx,\vy,g(\vx),g(\vy))=0$ for all $\vx,\vy \in \{0,1\}^n\subset \Fp^n$. If this is the case then the input is a YES-instance, whereas if for all functions $g$ that are ``close'' to multilinear functions at least one of the constraints $f(\vx,\vy,g(\vx),g(\vy))=0$ is not satisfied then it is a NO-instance. The difficulty, of course, is that there are exponentially many constraints to verify, and \emph{all} must be satisfied for the instance to be a YES-instance.

The protocol is divided into two distinct parts, which only weakly interact with each other. In the first part of the protocol, the verifier performs a polynomial-round \emph{low-degree sum-check test} with a single prover, say the last prover (see Lemma~\ref{lemma:summation-test} for an explicit formulation). This test is based on ideas already introduced by Lund, Fortnow, Karloff, and Nisan~\cite{LunForKarNis92JACM} and can be used to verify that a low-degree function defined over $\Fp^k$ vanishes on all of $\{0,1\}^k$.
We will apply it to the low-degree function $h:(\Fp^n)^2\to \Fp$ defined by $h(\vx,\vy) = f(\vx,\vy,g(\vx),g(\vy))$.
An important point for us is that, in the LFKN protocol, the verifier eventually only needs to evaluate $h$ at a \emph{single} point $(\vx,\vy)\in (\Fp^n)^2$ chosen uniformly at random.
Of course, the verifier only knows $f$, not $g$, and therefore the verifier asks the two remaining provers the values~$g(\vx)$ and~$g(\vy)$.

However, note that here the function $g$ is arbitrary (we are trying to verify its existence), \emph{except that it has to be multilinear}.
The goal of the second part of the protocol is to ensure that it is indeed chosen according to some multilinear function.
Therefore, the verifier will sometimes perform a certain ``multilinearity test'' with the three provers,
which enforces that, however the provers answer their queries, it must be according to a function that is close to a multilinear function.
The two tests will be indistinguishable from the point of view of the provers because the marginal distribution on the question to each prover is uniform over $\Fp^n$ in both cases.

\medskip

Completeness of the protocol is easy to verify, and in the case of a YES-instance honest provers do not need any entanglement to be accepted with probability $1$.  To prove soundness, assuming four entangled provers succeed with probability that is polynomially close to $1$, we wish to conclude that the instance given as input to the verifier must be a YES-instance.

Note that provers successful in the overall protocol must, in particular, succeed with high probability in the multilinearity test. The key step in the analysis consists in showing the following: Any three entangled provers that succeed in the multilinearity test with high probability are ``indistinguishable'' from \emph{classical} provers who use shared randomness to jointly sample a multilinear function $g$, and then answer question $\vx$ with $g(\vx)$. This step is the one that requires the most work, and we explain it in more detail in the next section. (In particular, we will clarify what is meant by ``indistinguishable''.)

Assuming this informal statement holds, it is not too hard to conclude the analysis of the protocol. Indeed, having replaced 
two out of the three provers by classical provers, there is only a single ``quantum'' prover left, the one used to perform the sum-check test in the first part of the protocol.
But entanglement cannot be useful to a single prover, and hence we may also assume that this last prover behaves classically. Since all provers are now classical, we have reduced our analysis to the classical setting and can appeal to the results in~\cite{BabForLun91CC} to conclude. We refer to Section~\ref{sec:protocol} for a more detailed presentation and soundness analysis of the protocol.

\subsection{The multilinearity game}\label{sec:multilin-game}

The key step in the proof of Theorem~\ref{thm:nexp_main} is the analysis of the multilinearity test of~\cite{BabForLun91CC}, which generalizes the celebrated \emph{linearity test} of Blum, Luby, and Rubinfeld~\cite{BLR93} and is essential in constructing a protocol for $\NEXP$ that has messages of polynomial length.\footnote{One can devise a protocol based on the linearity test alone, but it requires the verifier to send messages with exponential length to the provers. Such use of the linearity test was already key in establishing the early result $\NP \subseteq \PCP(\poly,1)$; see e.g.\ Theorem 2.1.10 in~\cite{AroLunMotSudSze98JACM}.} The test can be formulated as a game played between the verifier and three players. The game is parametrized by a finite field $\Fp$ and an integer $n$. In the game, the verifier performs either of the following with probability $1/2$ each:
\begin{itemize}
\item
  \emph{Consistency test.}
  The verifier chooses~$\vct{x}\in\Fp^n$ uniformly at random
  and sends the same question~$\vct{x}$ to all three players.
  He expects each of them to answer with an element of $\Fp$,
  and accepts if and only if all the answers are equal.
\item
  \emph{Linearity test.}
  The verifier chooses~$i\in\{1,\dots,n\}$, $\vct{x}\in\Fp^n$ and~$y_i,z_i\in\Fp$ uniformly at random,
  and sets~$y_j=z_j=x_j$ for every~$j\in\{1,\dots,n\}\setminus\{i\}$.
  He sends~$\vct{x},\vct{y},\vct{z}$ to the three players,
  receives $a,b,c\in\Fp$,
  and accepts if and only if
  \[
    \frac{b-a}{y_i-x_i}=\frac{c-b}{z_i-y_i}=\frac{c-a}{z_i-x_i}.
  \]
\end{itemize}

Babai, Fortnow, and Lund show that, if any three \emph{deterministic} players are accepted by the verifier with probability at least $1-\eps$ in this game, then the functions they each apply to their questions in order to determine their respective answers are close to a single \emph{multilinear} function $g:\Fp^n\to\Fp$ (see Theorem~4.16 in~\cite{BabForLun91CC} for an analysis of a variant of the test over the integers). That is, for all but at most a fraction roughly $O(n^2\eps)$ (provided $|\Fp|$ is large enough) of $\vx\in\Fp^n$, the players' answer to question $\vx$ is precisely $g(\vx)$. 

A major hurdle in proving a similar statement in case the players are allowed to use quantum mechanics already arises in \emph{formulating} the statement to be proven: even in the case of players restricting their use of entanglement as shared randomness, what meaning should one ascribe to their strategies being ``close to multilinear''? Indeed, it could be that the answer of each player to a fixed question, when taken in isolation, is uniformly random: the whole substance of the strategy is in the \emph{correlations} between the answers of different players. 
This difficulty is usually set aside by ``fixing the randomness''. Entanglement, however, cannot be ``fixed'', and this forces us to face even the presumably simpler case of randomized strategies head on. We show that the following is an appropriate formulation of Babai et al.'s multilinearity test in the general setting of entangled (or even just randomized) players (see Theorem~\ref{thm:lintestclose} for a precise statement).

\begin{theorem}[Informal]\label{thm:multilin-inf}
Suppose that three entangled players who share a permutation-invariant state $\ket{\Psi}$ succeed in the multilinearity game with probability $1-\eps$ where each player uses the set of measurements $\{A_\vx^a\}_{a\in \Fp}$ to determine his answer to the verifier's question $\vx\in\Fp^n$.

Then there exists a \emph{single} measurement $\{V^g\}$, independent of any question and with outcomes in the set of all multilinear functions $g:\Fp^n\to \Fp$, such that, in the multilinearity game, each player's action is \emph{indistinguishable} from that of player whom, upon receiving his question~$\vx$, would
\begin{enumerate}
\item Measure his share of $\ket{\Psi}$ with $\{V^g\}$, obtaining a multilinear function~$g$ as an outcome,
\item Answer his question $\vx$ with $g(\vx)$.
\end{enumerate}
Moreover, the multilinear functions used by the three players are identical with high probability.
\end{theorem}

In case the players are classical, but may use shared randomness, the theorem makes the following simple statement: players successful in the multilinearity game are ``indistinguishable'' from players who would first look up their random string, based on that alone select a multilinear function $g$, and finally answer their respective questions $\vx_i$ with $g(\vx_i)$. While such a statement is a direct corollary of Babai, Fortnow, and Lund's analysis, our contribution is to prove it without first ``fixing the randomness'' --- and to show that it also holds for the case of players using entanglement. 

\paragraph{An appropriate notion of distance on entangled-prover strategies.} Crucial to the applicability of Theorem~\ref{thm:multilin-inf} is the precise notion of ``indistinguishability'' used. Indeed, while there is no hope of making statements on the players' measurements or their shared entangled state themselves (since the verifier has no direct access to them throughout the protocol), one still needs to use a notion that is strong enough to be meaningful even when the multilinearity game is executed as a building block in the larger protocol explained in the previous section. 

The measure we use is based on the notion of \emph{consistency} between two measurements, and it may be useful to introduce it here in a simplified setting (precise definitions are given in Section~\ref{sec:prelim-notation}). Let $\{A^i\}_{i\in I}$ and $\{B^i\}_{i\in I}$ be two quantum measurements of the same dimension and indexed by the same set of outcomes: $A^i,B^i \geq 0$ for all $i\in I$, and $\sum_i A^i = \sum_i B^i = \Id$. Let $\ket{\Psi}$ be a bipartite state that is invariant under permutation of its two subsystems, and $\rho$ its reduced state on either. We say that $A$ and $B$ are $\eps$-\emph{consistent} if the following holds:
\beq\label{eq:cons-intro-0}
 \cons(A,B) \,:=\, \sum_{i} \,\bra{\Psi} A^i \otimes B^i \ket{\Psi} \,\geq\, 1-\eps.
\eeq
This definition has an operational interpretation: the two measurements $A$ and $B$, when performed on the two subsystems of $\ket{\Psi}$, give the same outcome except with probability $\eps$ . The key fact about consistent measurements is the following. Suppose that $A$ and $A$, $B$ and $B$, and $A$ and $B$ are all $\eps$-consistent. Then $A$ and $B$ are \emph{indistinguishable} in the sense that 
\beq\label{eq:cons-intro}
\sum_i \,\big\| \sqrt{A^i} \rho \sqrt{A^i} - \sqrt{B^i} \rho \sqrt{B^i} \big\|_1 \,=\,O(\sqrt{\eps}).
\eeq
This last expression corresponds to a more familiar notion of closeness of two measurements: they are close if the post-measurement states resulting from applying either are close in trace distance. The fact that~\eqref{eq:cons-intro-0} essentially implies~\eqref{eq:cons-intro} relies on Winter's ``gentle measurement'' lemma~\cite[Lemma~9]{Winter99IEEEIT} (see also Aaronson's ``almost as good as new'' lemma~\cite[Lemma~2.2]{Aar05}), a key tool in our analysis.

In this paper we will consider two measurements to be close whenever they are consistent, having the assurance that this notion of closeness implies the more traditional one expressed by~\eqref{eq:cons-intro}. In particular, it is not hard to verify that~\eqref{eq:cons-intro} implies that either measurement may be ``replaced'' by the other even in a wider context; see the proof of Claim~\ref{claim:pi} in Section~\ref{sec:protocol} for more details on how this can be done. The advantage of using this measure is that constraints on the consistency of measurements arise naturally from the analysis of the multilinearity game, and it is a notion that is very convenient to work with. 

\paragraph{Analysis of the multilinearity game: rounding entangled strategies.} 
Theorem~\ref{thm:multilin-inf} states that success in the multilinearity game forces even entangled players to make a trivial use of their entanglement: since the measurement $\{V^g\}$ is independent of their respective questions, they might as well perform it before the game starts, in which case they are not using their entanglement at all. Hence the theorem implies that entangled players are no more powerful than classical players in that game. A key insight of our work, however, is to avoid any attempt to prove such a statement \emph{directly}. Instead, our proof technique consists in progressively manipulating the players' strategies themselves, without \emph{explicitly} trying to relate them to a classical strategy. 

Our goal is to show how the measurement $\{V^g\}$ can be extracted from the initial set of measurements $\{A_\vx^a\}$ which depend on~$\vx\in\Fp^n$.\footnote{While we do give an explicit, inductive algorithmic procedure showing how $\{V^g\}$ can be constructed, this is not necessary: the point is only in proving its \emph{existence}.} 
More precisely, we show how, starting from the original measurements $\{A_{\vct{x}}^a\big\}$, one may remove the dependence of $\{A_{\vct{x}}^a\}$ on $\vct{x}\in\Fp^n$ one coordinate at a time --- eventually reaching the measurement $\{V^g\}$. Towards this we construct a sequence of measurements $\{B_{x_{k+1},\ldots,x_n}^g\big\}_g$, for $k=1,\ldots,n$, with outcomes $g$ in the set of multilinear functions $\Fp^{k}\to \Fp$. Each of these measurements has the following key property: the respective strategies corresponding to (i) measuring according to $\{A_{\vct{x}}^a\}$ and answering $a$ or (ii) measuring according to $\{B_{x_{k+1},\ldots,x_n}^g\}$ and answering $g(x_1,\ldots,x_{k})$ are \emph{consistent}, in the sense described in Eq.~\eqref{eq:cons-intro-0}: two distinct players using either strategy will obtain the same answer with high probability (provided they started with the same question). 

This sequence of measurements is defined by induction, and we only explain the one-dimensional case here. Our construction is intuitive: $\{B^g\}$ corresponds to measuring using $\{A_{x_1}^a\}$ twice, \emph{in succession}, using two randomly chosen values of $x_1$, and returning the unique linear function $g$ which interpolates between the two outcomes obtained. This can be interpreted as a quantum analogue of the reconstruction procedure already used in the linearity test of Blum, Luby, and Rubinfeld: to recover a linear function it suffices to evaluate it at two random points, and then interpolate. The construction of the measurements $\{B_{x_{k+1},\ldots,x_n}^g\}$ for the one dimensional case is given in Claim~\ref{claim:lines}, and in the general case in Lemma~\ref{lem:lemma2}, which states a quantum analogue of Babai et al.'s ``pasting lemma''~\cite[Lemma~5.11]{BabForLun91CC}. 

An additional hurdle arises as a result of the induction: the quality of the approximation between the original measurements $\{A_\vx^a\}$ and the constructed measurements $\{B_{x_{k+1},\ldots,x_n}^g\}$ blows up \emph{exponentially} with $k$. In order to control this error, one has to perform an additional step of \emph{self-improvement}. This step was a key innovation in the work of Babai, Fortnow, and Lund, and extending it to the setting of entangled strategies requires substantially more work. While for the case of deterministic strategies Babai et al.\ were able to show, using the expansion properties of the hypercube, that any ``reasonably good'' $k$-linear approximation $g$ at any point in the induction was automatically ``extremely good'', in our case we need to actively update the measurements through a self-correction procedure, obtaining the ``improved'' measurements as the optimum of a certain convex optimization problem. The need for such active correction is not a limitation of our approach, but rather reflects a fundamental difference between the quantum and the classical, deterministic settings: while two binary-valued functions either fully agree or fully disagree at any point, two quantum measurements can produce outcomes according to distinct but arbitrarily close distributions (think of one of the measurements as being obtained from the other by a small perturbation, such as an arbitrarily small rotation). It is this kind of ``error'' that needs to be corrected, and we explain our method to do so in more detail in Section~\ref{sec:lemma1}. 

\subsection{Discussion and open questions}\label{sec:discussion}

Improving the parameters in Theorem~\ref{thm:nexp_main} and Corollary~\ref{cor:mip-one-round} is an open problem.  For example, it might be possible to reduce the number of provers to two, and the number of rounds of interaction to one, while still preserving exponentially small soundness error, resulting in the inclusion~$\NEXP\subseteq\MIPstar(2,1,1,2^{-q})$ for every polynomial~$q$. This  would be an analogue of the known containment~$\NEXP\subseteq\MIP(2,1,1,2^{-q})$~\cite{FeiLov92STOC}.  Our overall protocol for~$\NEXP$ requires three provers, and four provers if we would like to parallelize it by using~\cite{Ito-parallel}. We leave the problem of reducing the number of provers for future work. It may also be possible to improve the soundness guarantees in Corollary~\ref{cor:mip-one-round} by using the parallel repetition techniques from~\cite{KV11parallel}, but we have not explored this possibility.

In comparison to the PCP theorem, there are important parameters which are not explicit in Theorem~\ref{thm:nexp_main} and Corollary~\ref{cor:mip-one-round}: the amount of randomness used by the verifier and the total answer length.  In our constructions, both of them are just bounded by a polynomial in the input length for~$\NEXP$, and they are poly-logarithmic for the scaled-down version corresponding to verification of languages in~$\NP$.  If these numbers are respectively reduced to a logarithm and a constant for~$\NP$ with a constant soundness, the result will be an analogue of the PCP theorem in presence of entanglement. Obtaining such a result may require extending our analysis of the multilinearity test to the more powerful \emph{low-degree} tests that were key to establishing the ``scaled-down'' version of the PCP theorem. 

Honest provers in our protocol do not need entanglement in order to achieve completeness $1$ in the case of a YES-instance. It remains open whether entanglement can have any positive use in this context: 
is~$\MIPstar$ strictly larger than~$\MIP=\NEXP$?

\paragraph{Organization of the paper.} After giving some necessary preliminaries, Section~\ref{sec:protocol} describes the protocol used to prove Theorem~\ref{thm:nexp_main}, and shows how the theorem follows from a claim about the multilinearity game in the presence of entangled provers. Section~\ref{sec:multilin} introduces a more technical claim about the analysis of the multilinearity game, which is suitable to a proof by induction on the number~$n$ of variables in the verifier's questions in the game. The actual analysis is given in Section~\ref{sec:mainind}. 

\paragraph{Acknowledgments.}
Tsuyoshi Ito thanks John Watrous for helpful discussions.
Thomas Vidick thanks Umesh Vazirani for many inspiring discussions throughout the time that this work was being carried out, and in particular for first suggesting to adapt Babai et al.'s multilinearity test to the entangled-prover setting.
The authors also thank Scott Aaronson, Dmitry Gavinsky, Oded Regev, and an anonymous referee for helpful suggestions.

\section{Preliminaries}\label{sec:prelim}

In the remainder of the paper we assume that the reader is familiar
with computational complexity theory~\cite{Goldreich08,AroBar09},
as well as with basic notions
in quantum information~\cite{NieChu01,KitSheVya02}
such as density matrices, POVM measurements, quantum channels, and the trace distance.
For more on quantum computational complexity we refer the reader to
a recent survey by Watrous~\cite{Watrous09-survey}.

\subsection{Notation}\label{sec:prelim-notation}

For a field~$\Fp$, a linear function $g\colon\Fp\to\Fp$ is a function such that there exists $a,b\in\Fp$, $g(x)=ax+b$.
A multilinear function $g\colon\Fp^k\to\Fp$ is a function that is linear in each of its coordinates.
$\ML(\Fp^k,\Fp)$ will denote the set of all multilinear functions from $\Fp^k$ to~$\Fp$.
We will denote tuples using bold symbols such as~$\vx$ and~$\vct{b}$.
Given a tuple~$\vx=(x_1,\ldots,x_n)$ and~$k\in[n]$, we let $\vx_{\leq k} := (x_1,\ldots,x_k)$, $\vx_{> k} := (x_{k+1},\ldots,x_n)$ and~$\vx_{\neg k} := (x_1,\ldots, x_{k-1},x_{k+1},\ldots,x_n)$.

Given a positive matrix $\rho$ and an arbitrary matrix $A$, we let $\Tr_\rho(A):=\Tr(A\rho)$. In case $\rho$ is a matrix on the tensor product of two Hilbert spaces $\mathcal{H}$ and $\mathcal{H}'$, and $A$ is a matrix acting on $\mathcal{H}$, we will sometimes abuse notation and write $\Tr_\rho(A)$ for $\Tr_\rho(A\otimes \Id_{\mathcal{H}'})$. If $\ket{\Psi} \in \mathcal{H}^{\otimes k} \otimes \mathcal{H}'$ is a state that is invariant under permutation of the first $k$ registers, we will often abuse notation further and use the symbol $\rho$ to denote the reduced density of $\ket{\Psi}$  on either of the first $k$ registers, or even any pair of registers among the first $k$, etc. Hence any expression of the form $\Trho(A \otimes B)$ should really be read as 
$$\bra{\Psi} A\otimes B \otimes \Id_\mathcal{H} \otimes \cdots \otimes \Id_{\mathcal{H}} \otimes \Id_{\mathcal{H}'} \ket{\Psi},$$
where the position of $A$ and $B$ among the first $k$ registers is immaterial by permutation-invariance. For any $\rho\geq 0$, we let 
$$\|A\|_\rho^2 \,:= \,\Tr\big(AA^\dagger \rho),$$
and observe that $A\mapsto \|A\|_\rho$ is a semi-norm (it is definite if $\rho$ is invertible). It satisfies the following Cauchy-Schwarz inequality: for any $A,B$,
$$ \Trho\big( AB^\dagger \big) \,\leq\, \|A\|_\rho\,\|B\|_\rho.$$

\paragraph{Measurements.} In this paper, a measurement is a collection of non-negative matrices $\{P^a\}_{a\in A}$ such that $\sum_a P^a = \Id$ (this is usually called a \emph{Positive Operator-Valued Measure}, or POVM). The set $A$ is the set of \emph{outcomes} of the measurement; outcomes will always appear as superscripts. The measurement is said \emph{projective} if $P^a$ is a projector, i.e.\ $(P^a)^2 = P^a$, for every $a$.  A \emph{sub}-measurement is a collection of non-negative matrices $\{P^a\}_{a\in A}$ such that $\sum_a P^a \leq \Id$. For integers $0\leq k \leq n$ we will also consider families of sub-measurements, indexed by $x \in \Fp^{n-k}$ and with outcomes in the set $\ML(\Fp^{k},\Fp)$. Such a family $P=\{ P_{\vx_{\geq k}}^g \}$ will be called a \emph{family of sub-measurements of arity $k$} (the parameter $n$ will often be left implicit). A family of sub-measurements of arity $n$ is thus a single sub-measurement with outcomes in $\ML(\Fp^n,\Fp)$. Given a family of sub-measurements $P=\{ P_{\vx_{\geq k}}^g \}$ of arity $k$, we will often use the notation
$$ P_{\vx_{\geq k}} \,:=\, \sum_g P_{\vx_{\geq k}}^g\qquad\text{and}\qquad P_{\vx_{\geq \ell}}\,:=\, \Es{x_k,\ldots,x_{\ell-1}} P_{\vx_{\geq k}}$$ 
for any $k\leq\ell\leq n$, where the expectation is taken with respect to the uniform distribution on $\Fp^{\ell-k}$. Given two families of sub-measurements $P$ and $Q$ with arities $k\leq \ell$ respectively, we define their \emph{consistency}
$$\cons(P,Q):= \Es{\vx\in\Fp^n} \sum_{f,g:\, g_{|\vx_{k\cdots\ell-1}}=f} \Trho\big( P_{\vx_{\geq k}}^f \otimes Q_{\vx_{\geq \ell}}^g \big), $$
where $g_{|\vx_{k\cdots\ell-1}}$ is the $(n-\ell)$-linear function obtained by restricting $g$'s $(\ell-k)$ first variables to $\vx_{k\cdots\ell-1}$, and their \emph{inconsistency}
$$\inc(P,Q):= \Es{\vx\in\Fp^n} \sum_{f,g:\, g_{|\vx_{k\cdots\ell-1}\neq f}} \Trho\big( P_{\vx_{\geq k}}^f \otimes Q_{\vx_{\geq \ell}}^g \big), $$
where $\rho$ is a density matrix which will always be clear from the context. If $k>\ell$ then we define $\cons(P,Q) := \cons(Q,P)$ and $\inc(P,Q):=\inc(Q,P)$. We will also use shorthands $\cons(P)=\cons(P,P)$ and $\inc(P) = \inc(P,P)$. Note that if $P$ is a complete family of measurements, i.e.\ $\sum_f P_{\vx_{\geq k}}^f =\Id$ for every $\vx_{\geq k}$, then 
$$\cons(P,Q)+\inc(P,Q)\,=\, \Es{\vx} \sum_g \Trho\big(Q_{\vx_{\geq \ell}}^g \big) \,=\,\Trho(Q),$$
 which equals $1$ if $Q$ is also complete.

\subsection{Multi-prover interactive proofs}\label{sec:prelim-mip}

In this section we define the complexity classes that this work is concerned with: multi-prover interactive proof systems ($\MIP$ systems) and multi-prover interactive proof systems \emph{with entanglement} ($\MIPstar$ systems). 

Let~$k(n)$ be an integer, denoting the number of provers, and~$m(n)$ an integer denoting the number of rounds. Both~$k(n)$ and~$m(n)$ are from the set of polynomially bounded, polynomial-time computable functions in the input size $|x|$, denoted by $\poly$.
Further, $c$ and $s$ denote polynomial-time computable functions
 of the input size into $[0,1]$ corresponding to completeness acceptance probability and
soundness error. For notational convenience in what follows  we will omit the arguments of these functions.

\paragraph{Multi-prover interactive proof systems ($\MIP$ systems):}

Let~$k,m,l\in\poly$.
A~$k$-prover interactive proof system consists of a verifier~$V$ and~$k$ provers~$P_1,\dots,P_k$. The verifier is a probabilistic polynomial-time Turing machine, and the provers are computationally unbounded. Each of them has a read-only input tape and a private work tape. Each prover has a communication tape. The verifier has a random tape. The verifier also has~$k$ communication tapes, one for each prover, each of which is~$l$ bits long.

The input tape for every party contains the same input string~$x$. The protocol consists of~$m(\abs{x})$ rounds. In each round, first the verifier runs for a polynomial amount of time, updating the work and communication tapes.  After that, the content of the~$i$th communication tape is sent to the~$i$th prover for each~$i=1,\dots,k(\abs{x})$.  Each prover reads this string, updates the content of his own work tape, and decides a reply to the verifier.  The reply from the $i$th prover is written in the~$i$th communication tape, and this round completes.  After~$m(\abs{x})$ rounds of interaction, the verifier produces a special output bit, designating acceptance or rejection.  The operations by provers are instantaneous and do not have to be even computable; the provers are assumed to be able to ``compute'' any function.

For simplicity, we assume that each message between the verifier and the provers in each round is exactly~$l$ bits long for the purpose of a formal definition, but it is not hard to modify the definition to incorporate the more general case which does not satisfy this assumption.
Formally, a \emph{strategy} for~${P_1,\dots,P_k}$ in a $k$-prover~$m$-round interactive proof system consists of the length~$l'\in\N$ of a work tape, and~$km$ mappings~$f_{ij}\colon\{0,1\}^l\times\{0,1\}^{l'}\to\{0,1\}^l\times\{0,1\}^{l'}$ for~$1\le i\le k$ and~$1\le j\le m$.
Each mapping~$f_{ij}$ specifies the operation which prover~$i$ performs in round~$j$:
$f_{ij}(q,w)=(q',w')$ means that if the message from the verifier in this round is~$q$ and the work tape contains string~$w$ before the operation by the prover,
then the message to the verifier in this round is~$q'$ and the work tape contains string~$w$ after the operation.

\begin{definition}
Let~$k,m\colon\N\to\N$, and let~$c,s\colon\N\to[0,1]$ such that~$c(n)>s(n)$ for all~$n\in\N$.
A language $L$ is in ${\MIP(k, m, c, s)}$ if and only if there exists an $m$-round polynomial-time verifier $V$
for a $k$-prover interactive proof system such that, for every input $x$:
\begin{description}
\item[\textnormal{(Completeness)}] if ${x \in L}$, there exists a strategy for provers~${P_1, \dots,
P_{k}}$ such that the interaction protocol of $V$ with $(P_1,\ldots,P_k)$ results in the verifier accepting with probability at least $c$,
 \item[\textnormal{(Soundness)}] if ${x \not\in L}$, for any strategy for provers~${P'_1, \ldots, P'_{k}}$, the probability that the interaction protocol of $V$ with $(P_1,\ldots,P_k)$ results in the verifier accepting is at most $s$.
\end{description}
\label{Definition: MIP(k,m,c,s)}
\end{definition}

In this formulation, the provers are deterministic, but this is not a limitation because it is well-known that the power of the model does not change if we allow the provers to share a random source.

If some of the parameters $k$, $m$, $c$, and~$s$ are sets of functions instead of single functions,
the class is interpreted to be the union over all choices in the sets.
For example,
\[
  \MIP(4,1,1,1-1/\poly)=\bigcup_{f\in\poly}\MIP(4,1,1,1-1/f).
\]
We denote~$\MIP(\poly,\poly,2/3,1/3)$ simply by~$\MIP$.

\paragraph{Multi-prover interactive proof systems with entanglement ($\MIPstar$ systems):}
First introduced in~\cite{CHTW04}, $\MIPstar$ systems are defined analogously to $\MIP$ systems. The only difference is that now the provers are allowed to be \emph{quantum}, while the verifier (and communication) remains bounded in classical probabilistic polynomial-time. This implies that the provers may share an arbitrary entangled state~$\ket{\Psi}$ among themselves before the protocol starts and that each prover may use his part of the entangled state to determine his reply to the verifier. In each round, the provers individually receive the messages from the verifier in a message register, perform a quantum operation on this register together with their share of the entangled state, measure the message register in the computational basis, and send back the outcome to the verifier.

Formally, an \emph{entangled strategy} for~${P_1,\dots,P_k}$ in a $k$-prover~$m$-round interactive proof system with entanglement consists of the length~$l'\in\N$ of a work tape, $km$ quantum channels~$\Phi_{ij}$ from a quantum register of~$l+l'$ qubits to itself for~$1\le i\le k$ and~$1\le j\le m$, and the initial quantum state~$\ket\Psi$ of the work tape, which is a~$kl'$-qubit state.
Each channel~$\Phi_{ij}$ specifies the operation which prover~$i$ performs in round~$j$:
the first~$l$ qubits in the state correspond to the message from and to the verifier, and the last~$l'$ qubits represent the content of the work tape.
After the prover's operation, the first~$l$ qubits are measured in the computational basis and sent to the verifier.

\begin{definition}
A language $L$ is in ${\MIPstar(k, m, c, s)}$ if and only if there exists an $m$-round polynomial-time verifier $V$
for $k$-prover interactive proof systems such that, for every input $x$:
\begin{description}
\item[\textnormal{(Completeness)}] if ${x \in L}$, there exists an entangled strategy for provers ${P_1, \ldots,
P_{k}}$ such that the interaction protocol of $V$ with $(P_1,\ldots,P_k)$ results in the verifier accepting with probability at least $c$,
 \item[\textnormal{(Soundness)}] if ${x \not\in L}$, for any entangled strategy for provers ${P'_1, \ldots, P'_{k}}$, the probability that the interaction protocol of $V$ with $(P_1,\ldots,P_k)$ results in the verifier accepting is at most $s$.
\end{description}
\label{Definition: MIPstar(k,m,c,s)}
\end{definition}

In certain cases, we can simplify part of the definition of entangled strategies.
Suppose that the verifier interacts with certain prover~$P_i$ only once;
i.e., the verifier is guaranteed to send~$P_i$ the empty string (or a fixed string) in rounds other than round~$j$,
and is guaranteed to ignore the reply from~$P_i$ in rounds other than round~$j$.
In this case, instead of specifying~$m$ quantum channels to describe the behavior of~$P_i$ in the~$m$ rounds,
we may just specify measurements~$A_q=(A_q^r)$ for each message~$q$ from the verifier,
where the outcome of each measurement gives a reply to the verifier.%
\footnote{Any classical post-processing by the prover can be incorporated as part of the description of his measurement.}
Since all the interactive proof systems considered in this paper have the property that the verifier interacts with each prover only once except for one prover,
we use this simplified formulation in many places.

Note that we do not assume any upper bound on the size~$l'$ of the work tape used by each prover (in particular, we do not assume that~$l'\in\poly$; the model with this restriction is considered in~\cite{KobMat03JCSS}).
However, we do assume that they only use a finite-dimensional Hilbert space.
A more general definition is commuting-operator provers, considered by Tsirelson~\cite{Tsirelson80LMP} in the context of Bell inequalities and later in~\cite{SW08,DLTW08,NPA08NJP,yao:tsirelson}.
Although we expect that our results remain valid with minor modifications to the proofs even if dishonest provers are allowed to use arbitrary commuting-operator strategies, we have not explored this possibility.

\paragraph{Symmetry.} We will make an important use of symmetry in the protocols that we introduce. It will be a useful simplifying assumption in two respects: first it lets one assume that the set of measurements used by all provers is the same. Second, and most important, it implies that the provers' shared entangled state is also permutation-invariant. 

\begin{definition} Let $(P_1,\ldots,P_k,\ket{\Psi})$ be a $k$-prover strategy.\footnote{We think of $P_i$ as an arbitrary representation of the set of all quantum channels applied by prover $i$ throughout the protocol.} We say that this strategy is symmetric, or permutation-invariant, if $P_1 = \cdots = P_k$ and $\ket{\Psi}$ is invariant with respect to any permutation of the  subsystems corresponding to each prover.
\end{definition}

The following simple lemma (which already appears in~\cite[Lemma~4]{KKMTV11}) shows that one can always assume  without loss of generality that if a game has a certain symmetry then there is an optimal strategy for the provers which reflects that symmetry.

\begin{lemma}\label{lem:symmetry1}
Suppose an $\MIPstar$ proof system is given such that the protocol treats provers $P_1,\ldots,P_k$ symmetrically (i.e.\ the protocol is invariant under permutation of their questions and corresponding inverse-permutation of their answers). 
Then given
any strategy $P_1,\ldots,P_k$ with entangled state $\ket{\Psi}$ that succeeds with probability $p$, there exists a strategy $P'_1,\ldots, P'_k$ with entangled state $\ket{\Psi'}$ and success probability $p$ such that $P'_{1}=\cdots = P'_{k}$ and $\ket{\Psi'}$ is permutation-invariant. 
\end{lemma}

\begin{proof}
By appropriately padding
with extra qubits, assume that all $k$ registers of $\ket{\Psi}$ have the same dimension. Define strategies
$P'_1,\ldots,P'_k$ as follows: the provers share the entangled state
$\ket{\Psi'}=\sum_{\sigma\in\mathfrak{S}_{k}}
\ket{\sigma(1)}\otimes \cdots \otimes \ket{\sigma(k)}\otimes\ket{\Psi^{\sigma}}$, where the register containing $\ket{\sigma(i)}$ is given to prover $i$ and $\ket{\Psi^{\sigma}}$ is obtained from $\ket{\Psi}$ by permuting its registers according to $\sigma$. For $1\leq i \leq k$ prover $i$ measures the register containing $\ket{\sigma(i)}$ and behaves as in the strategy $P_{\sigma(i)}$. By the assumed symmetry of the protocol this new strategy has the same success probability $p$, and $\ket{\Psi'}$ has the required symmetry properties.
\end{proof}

The following claim states a trivial but useful fact about symmetric one-round strategies.

\begin{claim}\label{claim:prelim_symmetryuse} Let $(P_1,\ldots,P_k,\ket{\Psi})$ be a symmetric one-round strategy, and for every $i\in\{1,\ldots,k\}$, $\{A_i^a\}_a$ a measurement for the $i$-th prover in that strategy. Then for every permutation $\sigma$ on $\{1,\ldots,k\}$, and every $(a_1,\ldots,a_k)$,
$$ \bra{\Psi} A_1^{a_1} \otimes \cdots \otimes A_k^{a_k} \ket{\Psi} \,=\, \bra{\Psi} A_{\sigma(1)}^{a_{\sigma(1)}} \otimes \cdots \otimes A_{\sigma(k)}^{a_{\sigma(k)}} \ket{\Psi}.$$
\end{claim}

\subsection{$\NEXP$-complete problems}\label{sec:nexpproblems}

We will use the following $\NEXP$-complete problem, whose~$\NEXP$-completeness was shown by Papadimitriou and Yannakakis~\cite{PapYan86IC}:

\begin{problem}[Succinct 3-colorability]
  \textit{Instance.}
  An integer~$n\in\N$ in unary
  and a Boolean circuit~$C$ for a function~$\{0,1\}^n\times\{0,1\}^n\to\{0,1\}$
  which represents the adjacency matrix of a graph on~$2^n$ vertices.

  \textit{Question.}
  Is the graph represented by~$C$ 3-colorable?
\end{problem}

Using the standard technique of arithmetization
(see e.g.\ Proposition~3.1 and Lemma~7.1 of Ref.~\cite{BabForLun91CC}), one can show that the following problem is also $\NEXP$-complete.

\begin{problem}[Succinct 3-colorability, arithmetized version] \label{problem:ora-3-sat-arith}
  \textit{Instance.}
  Integers~$r,n\in\N$ in unary
  and an arithmetic expression%
  \footnote{An \emph{arithmetic expression} is a rooted tree whose internal nodes represent either addition or multiplication
    and whose leaves represent either variables or an integer constant.
    The size of an arithmetic expression is the number of nodes
    plus the sum of the number of bits required to represent the integer for each constant node.}
  for a polynomial~$f(\alpha,\vct{z},\vct{b}_1,\vct{b}_2,\allowbreak a_1,a_2)$,
  where~$\vct{z}$ represents~$r$ variables
  and each of~$\vct{b}_1,\vct{b}_2$ represents~$n$ variables.

  \textit{Yes-promise.}
  If~$\F$ is a field with more than two elements and~$\alpha\in\F\setminus\{0,1\}$,
  then there exists a mapping~$A\colon\{0,1\}^n\to\{0,1,\alpha\}$ such that
  for all~$\vct{z}\in\{0,1\}^r$ and all~$\vct{b}_1,\vct{b}_2\in\{0,1\}^n$, it holds that
  \begin{equation}
    f(\alpha,\vct{z},\vct{b}_1,\vct{b}_2,A(\vct{b}_1),A(\vct{b}_2))=0.
    \label{eq:arithmetized-constraint}
  \end{equation}

  \textit{No-promise.}
  If~$\F$ is a field with more than two elements and~$\alpha\in\F\setminus\{0,1\}$,
  then for every mapping~$A\colon\{0,1\}^n\to \F$,
  there exist~$\vct{z}\in\{0,1\}^r$ and~$\vct{b}_1,\vct{b}_2\in\{0,1\}^n$
  such that Eq.~(\ref{eq:arithmetized-constraint}) is not satisfied.
\end{problem}

We note that the degree of the polynomial $f$ represented by the arithmetic expression can be at most the size of the arithmetic expression, and is therefore bounded by the input size. 

\subsection{Summation test}\label{sec:summation-test}

Let~$\Fp$ be a finite field of characteristic two.%
\footnote{The restriction to fields of characteristic two
  arises from the use of Theorem~\ref{theorem:small-bias}
  in Appendix~\ref{appendix:and-test}.}
If~$\abs{\Fp}=2^k$,
an encoding scheme of elements in~$\Fp$ is specified by~$k$
and an irreducible polynomial~$f(t)$ over~$\Ftwo$ of degree~$k$.
In particular, if~$\abs{\Fp}=2^{2\cdot3^e}$,
then it is known that~$f(t)=t^{2\cdot3^e}+t^{3^e}+1$ is irreducible over~$\Ftwo$,
and this specifies an encoding scheme for~$\Fp$
(see Appendix~G.3 of Goldreich~\cite{Goldreich08}).%
\footnote{Alternatively, we can use a deterministic polynomial-time algorithm to find an irreducible polynomial
  of a specified degree over~$\Ftwo$ by Shoup~\cite{Shoup90MCOM}.}

Consider the following promise problem, which has both an explicit and an implicit input.

\begin{problem}[Summation Test Problem]
  \textit{Explicit input.}
  Integers~$m,d\in\N$ in unary, and an encoding scheme of a finite field~$\Fp$
  of characteristic two.

  \textit{Implicit input.}
  A mapping~$h\colon \Fp^m\to \Fp$.

  \textit{Promise.}
  The given encoding scheme is valid,
  and the mapping~$h\colon \Fp^m\to \Fp$ is a polynomial function of degree at most~$d$ in each variable.

  \textit{Question.}
  Is
  \begin{equation}
    \sum_{\vct{x}\in\{0,1\}^m}h(\vct{x})=0 \quad \text{(in~$\Fp$)?}
    \label{eq:summation-test}
  \end{equation}
\end{problem}

In a (single-prover) interactive proof system for a problem with an implicit input,
the implicit input is given to the verifier as an oracle.%
\footnote{In Ref.~\cite{BabForLun91CC}, the authors refer to the interactive proof system for the Summation Test Problem
  as an ``interactive oracle-protocol,''
  viewing the mapping~$h$ as an exponentially long certificate string which is given to the verifier as an oracle.
  However, for our purposes it will be more convenient to treat~$h$ as part of the input.}
The following variant of the summation test of Lund, Fortnow, Karloff, and Nisan~\cite{LunForKarNis92JACM}
is a special case of Lemma~3.5 in Ref.~\cite{BabForLun91CC}.

\begin{lemma}[Summation test~\cite{BabForLun91CC}] \label{lemma:summation-test}
  Suppose that $\abs{\Fp}\geq 2dm$. Then there exists a single-prover interactive proof system for the Summation Test Problem with perfect completeness and soundness error at most~$dm/\abs{\Fp}$.
  Moreover, in this interactive proof system, the verifier behaves as follows.
  First he chooses~$\vct{q}\in \Fp^m$ uniformly at random.
  Then he interacts with the prover.
  At the same time, he reads the value~$h(\vct{q})$ from the implicit input.
  Finally he accepts or rejects depending on~$\vct{q}$, $h(\vct{q})$, and the interaction with the prover.%
  \footnote{In particular, this implies that the verifier reads only one value~$h(\vct{q})$ from the implicit input
    and the position~$\vct{q}\in \Fp^m$ to read is chosen uniformly in~$\Fp^m$.
    Together with the soundness guarantee,
    this in turn implies that if the implicit input is~$\delta$-close to a polynomial function~$\tilde{h}$ of degree at most~$d$ in each variable
    and~$\tilde{h}$ fails to satisfy the equation~(\ref{eq:summation-test}),
    then the verifier accepts with probability at most~$\delta+dm/\abs{\Fp}$
    no matter what the prover does.}
\end{lemma}

To apply the summation test to Problem~\ref{problem:ora-3-sat-arith},
we have to consider exponentially many constraints instead of one.

\begin{problem}[AND Test Problem]
  \textit{Explicit input.}
  Integers~$k,d\in\N$ in unary, and an encoding scheme
  of a finite field~$\Fp$ of characteristic two.

  \textit{Implicit input.}
  A mapping~$h\colon \Fp^k\to \Fp$.

  \textit{Promise.}
  The given encoding scheme is valid,
  and the mapping~$h\colon \Fp^k\to \Fp$ is a polynomial function of degree at most~$d$ in each variable.

  \textit{Question.}
  Is~$h(\vct{i})=0$ (in~$\Fp$) for all~$\vct{i}\in\{0,1\}^k$?
\end{problem}

The idea for the following corollary is already explained in Section~7.1
of Ref.~\cite{BabForLun91CC}.
We will give a proof in Appendix~\ref{appendix:and-test}
for the sake of completeness.

\begin{corollary} \label{corollary:and-test}
  There exists a polynomial~$q\colon\N\times\N\to\N$ for which the following holds.
 There exists a single-prover interactive proof system for the AND Test Problem with perfect completeness and soundness error at most~$5/8+q(k,d)/\abs{\Fp}$.
  Moreover, in this interactive proof system, the verifier behaves as follows.
  First he chooses~$\vct{i}\in \Fp^k$ uniformly and independently at random.
  Then he interacts with the prover.
  At the same time, he reads the value~$h(\vct{i})$ from the implicit input.
  Finally he accepts or rejects depending on~$\vct{i}$, $h(\vct{i})$, and the interaction with the prover.
\end{corollary}


\section{A proof system for Succinct $3$-Colorability}\label{sec:protocol}

In this section we prove~Theorem~\ref{thm:nexp_main}, assuming the soundness of the multilinearity game (see Theorem~\ref{thm:lintestclose} below),
which will be proved in Sections~\ref{sec:multilin} and~\ref{sec:mainind}.
We first describe a three-prover poly-round proof system for the NEXP-complete Succinct $3$-satisfiability problem, Problem~\ref{problem:ora-3-sat-arith}, in Section~\ref{sec:3sat-protocol}. In Section~\ref{sec:3sat-comp} we show that the protocol has perfect completeness with classical provers, and in Section~\ref{sec:3sat-sound} we show that it has soundness error at most~$1-1/\poly$ with entangled provers. Theorem~\ref{thm:nexp_main} is then obtained by repeating this protocol sequentially.

\subsection{Description of the protocol}\label{sec:3sat-protocol}

We construct a three-prover poly-round proof system for Problem~\ref{problem:ora-3-sat-arith}. Our protocol follows that of~\cite{BabForLun91CC} for the Oracle-3-Satisfiability problem very closely.
In the protocol or~\cite{BabForLun91CC}, the verifier makes three queries to the oracle which answers a Boolean value.
Because our problem is Succinct 3-Colorability instead of Oracle-3-Satisfiability, the verifier would make two queries to the oracle which answers a ternary value.
We replace these two queries to the oracle by queries to two distinct provers.

Label the provers as~$P_1,P_2,P_3$. The protocol will be symmetric under any permutation of the three provers.
Let~$(r,n,f)$ be an instance of Problem~\ref{problem:ora-3-sat-arith}, as described in Section~\ref{sec:nexpproblems}. 
Let~$d_f$ be the maximum degree of~$f$ in any one variable.
Let~$m=r+2n$ and~$d=2d_f$.
Let~$0<c_0<1$ be a constant defined later
(in Theorem~\ref{thm:lintestclose}),
and~$p$ be the smallest number of the form~$p=2^{2\cdot 3^e}$
such that~$p>\max\{8q(m,d),n^{1/c_0+4}\}$,
where~$q$ is the polynomial appearing
in the statement of Corollary~\ref{corollary:and-test}.
Let~$\Fp$ be the finite field of size~$p$.
As was noted in Section~\ref{sec:summation-test},
an explicit encoding scheme for~$\Fp$ is known in this case.
In the protocol, all arithmetic operations in~$\Fp$ are performed using this encoding scheme.

In the protocol, each prover~$P_i$ is told explicitly to play one of the following two roles:
\begin{itemize}
\item
  \emph{Lookup prover}: $P_i$ receives an element of~$\Fp^n$, and responds with an element of~$\Fp$.
  In this case, the interaction between the verifier and~$P_i$ takes only one round.
\item
  \emph{AND-test prover}: $P_i$ acts as the prover in the protocol for the AND test (Corollary~\ref{corollary:and-test}).
  In this case, the interaction between the verifier and~$P_i$ takes polynomially many rounds.
\end{itemize}

The verifier performs one of the following five tests chosen uniformly at random:
\begin{itemize}
\item
  \emph{Consistency test.}
  He tells each prover to act as a lookup prover.
  He chooses~$\vct{x}\in \Fp^n$ uniformly at random
  and sends the same question~$\vct{x}$ to provers~$P_1,P_2,P_3$.
  He expects each prover to answer with an element of $\Fp$,
  and accepts if and only if all the answers are equal.
\item
  \emph{Linearity test.}
  He tells each prover to act as a lookup prover.
  He chooses~$i\in\{1,\dots,n\}$, $\vct{x}\in \Fp^n$ and~$y_i\neq z_i\in \Fp\backslash\{x_i\}$ uniformly at random,
  and sets~$y_j=z_j=x_j$ for every~$j\in\{1,\dots,n\}\setminus\{i\}$.
  He sends~$\vct{x}$ to~$P_1$, $\vct{y}$ to~$P_2$, and~$\vct{z}$ to~$P_3$.
  He receives $a,b,c\in\Fp$ from these three provers,
  and accepts if and only if
  \[
    \frac{b-a}{y_i-x_i}=\frac{c-b}{z_i-y_i}=\frac{c-a}{z_i-x_i}.
  \]
\item
  \emph{AND test with~$P_3$ as the AND-test prover.}
  He tells provers~$P_1$ and~$P_2$ to act as lookup provers,
  and~$P_3$ to act as an AND-test prover.
  He chooses~$\alpha\in\Fp\setminus\{0,1\}$ in some canonical way;
  e.g.\ set~$\alpha=t$ when~$\Fp$ is viewed as~$\Ftwo[t]/(t^{2\cdot3^e}+t^{3^e}+1)$.
  Then, the verifier simulates the interactive proof system from Corollary~\ref{corollary:and-test} with the explicit input~$(m,d)$ and prover $P_3$.
  When the verifier in Corollary~\ref{corollary:and-test}
  tries to read the value~$h(\vct{z},\vct{b}_1,\vct{b}_2)$ in the implicit input,
  where~$\vct{z}\in \Fp^r$ and~$\vct{b}_1,\vct{b}_2\in \Fp^n$,
  our verifier simulates this by sending~$\vct{b}_1$ to~$P_1$ and~$\vct{b}_2$ to~$P_2$.
  Upon obtaining answers~$a_1,a_2$ to his queries from these two provers,
  he evaluates~$f(\alpha,\vct{z},\allowbreak\vct{b}_1,\vct{b}_2,\allowbreak a_1,a_2)$ and uses the result as the value of~$h(\vct{z},\vct{b}_1,\vct{b}_2)$.
\item
  \emph{AND test with~$P_1$ as the AND-test prover.}
  The same as above, with~$P_1$ and~$P_3$ swapped.
\item
  \emph{AND test with~$P_2$ as the AND-test prover.}
  The same as above, with~$P_2$ and~$P_3$ swapped.
\end{itemize}

Note that each prover is asked a question~$\vct{x}\in \Fp^n$ distributed uniformly at random
except when he is told to act as an AND-test prover.

\subsection{Completeness}\label{sec:3sat-comp}

Let~$(r,n,f)$ be a yes-instance of Problem~\ref{problem:ora-3-sat-arith}.
Then there exists a mapping~$A\colon\{0,1\}^n\to\Fp$ such that Eq.~(\ref{eq:arithmetized-constraint}) is satisfied
for all~$\vct{z}\in\{0,1\}^r$ and all~$\vct{b}_1,\vct{b}_2\in\{0,1\}^n$ simultaneously.
Let~$g$ be the unique extension of~$A$ to a multilinear function~$g\colon \Fp^n\to \Fp$.
Each lookup prover answers~$g(\vct{b})$ on question~$\vct{b}\in \Fp^n$.
Then it is clear that this deterministic strategy is accepted with certainty in the consistency test and the linearity test.
In the AND test, note that the value of~$h(\vct{z},\vct{b}_1,\vct{b}_2)$
which the verifier uses is given by
\[
  h(\vct{z},\vct{b}_1,\vct{b}_2)
  =
  f(\alpha,\vct{z},\vct{b}_1,\vct{b}_2,g(\vct{b}_1),g(\vct{b}_2)),
\]
which is a polynomial in~$\vct{z},\vct{b}_1,\vct{b}_2$
of degree at most~$2d_f=d$ in each variable.
Therefore, the promise in Corollary~\ref{corollary:and-test} is satisfied
and the AND-test prover has a strategy which makes the verifier accept with certainty.

\subsection{Soundness}\label{sec:3sat-sound}

The soundness analysis is divided in two parts. First we analyze the consistency and linearity tests, which only involve the questions in~$\Fp^n$, and show that success in those tests implies the following. (We refer the reader to Section~\ref{sec:prelim} for some relevant notation and definitions.)

\begin{theorem}\label{thm:lintestclose}
There exist positive universal constants $c_0<1$, $c<1$, and~$C>1$ such that the following holds.
Let~$n\ge1$ be an integer.
Let~$\Fp$ be a finite field, and $(\ket{\Psi},\{A_{\vx}^a\})$ a (symmetric, projective) strategy for the provers in the three-player multilinearity game in~$n$ variables over~$\Fp$ (as defined in Definition~\ref{def:multilin-strategy} below) that passes both the consistency and the linearity tests with probability at least~$1-\eps$. Assume furthermore that $p:=\abs{\Fp}\geq n^4\eps^{-1/2}$ and $\eps \leq n^{-2/c_0}$. Then there exists a sub-measurement~$\{V^g\}$, indexed by multilinear $g\in\ML(\Fp^n, \Fp)$, such that
\beq\label{eq:mlclose}
 \Es{\vx} \sum_a \Tr_\rho\bigl((A_\vx^a - \sqrt{V_\vx^a})^2\bigr) \, \leq \, C\,\eps^{c},
 \eeq
where for every $\vx\in\Fp^n$ and $a\in\Fp$ we defined $V_\vx^a := \sum_{g:\,g(\vx)=a} V^g$.
\end{theorem}

The proof of Theorem~\ref{thm:lintestclose} is our main technical contribution, and it is given in Sections~\ref{sec:multilin} and~\ref{sec:mainind}. Assuming the theorem, we prove that our proof system has soundness error at most~$1-n^{-2/c_0}/5$,
provided~$n$ is larger than an absolute constant depending on~$c$, $c_0$, and~$C$.

Let~$(r,n,f)$ be a no-instance.
Toward contradiction, suppose that the provers have a symmetric\footnote{Lemma~\ref{lem:symmetry1} shows that we may assume this holds without loss of generality.} entangled strategy~$S$ whose acceptance probability is at least~$1-\varepsilon/5$,
where~$\varepsilon=n^{-2/c_0}$.
Let~$\ket\Psi\in\calP_1\otimes\calP_2\otimes\calP_3$ be the state used in the strategy~$S$.
Let~$(A_{\vct{x}}^a)_{a\in \Fp}$ be the projective measurements used by each of the three provers in the strategy~$S$
upon question~$x\in\Fp^n$ when he acts as a lookup prover.

The verifier can be viewed as playing the multilinearity game with probability~$2/5$
and performing something else, namely the AND test, with probability~$3/5$.
Therefore, the strategy~$S$ has winning probability at least~$1-\varepsilon/2$
in the multilinearity test.
Because~$\abs{\Fp}=p>n^{1/c_0+4}=n^4\eps^{-1/2}$,
Theorem~\ref{thm:lintestclose} implies that
there exists a sub-measurement~$\big\{V^g\big\}_{g\in\ML(\Fp^n,\Fp)}$ such that inequality~\eqref{eq:mlclose} holds,
where~$\rho$ is the reduced state of~$\ket\Psi\bra\Psi$ on~$\calP_1$.
For every $\vx\in\Fp^n$ and $a\in\Fp$, let
\[
  V_{\vct{x}}^a=\sum_{\substack{g\in\ML(\Fp^n,\Fp) \\ g(\vct{x})=a}}V^g.
\]

For~$0\le i\le 2$, let~$S_i$ be the entangled strategy obtained from~$S$ by replacing the measurement for the first~$i$ provers~$P_1,\dots,P_i$ for question~$\vct{x}\in\Fp^n$ by~$V_{\vct{x}}^a$.\footnote{Since $V$ is a sub-measurement, the $V_\vx^a$ may not sum to identity. In that case we introduce an additional outcome ``fail'', corresponding to the element $\Id - \sum_a V_\vx^a$. Whenever a prover obtains that outcome he aborts the protocol.}
Note~$S_0=S$.

Let~$V'$ be the verifier who performs one of the consistency test, the linearity test,
and the AND test with~$P_3$ as the AND-test prover each with probability~$1/3$.
Note that when interacting with~$V'$, provers~$P_1$ and~$P_2$ are always told to act as lookup provers.
For~$0\le i\le2$, let~$p_i$ be the probability that the strategy~$S_i$ is accepted by~$V'$.

By definition,~$p_0\ge1-\varepsilon/3$.
We prove the following.

\begin{claim}\label{claim:pi} For $i=1,2$, it holds that~$\abs{p_{i-1}-p_i}\le\sqrt{C\varepsilon^c}$.
\end{claim}

\begin{proof}
The only difference between strategies~$S_{i-1}$ and~$S_i$
is the measurements used by prover~$P_i$.
We call the message from the verifier to~$P_i$ as register~$\calA$,
and call everything other than~$\calA$ and the private space~$\calP_i$ for prover~$P_i$
as register~$\calB$.
Register~$\calA$ is classical, but we treat it as a quantum register which always contains a state in the computational basis.
Let~$\sigma$ be the global state before prover~$P_i$ performs his measurement,
and~$\sigma_A$ (resp.\ $\sigma_V$) be the global state after prover~$P_i$ performs the measurement $A_\vx$ (resp.\ $V$) on his share of the state, and then discards the post-measurement state.
Since the marginal distribution on the question to~$P_i$ is uniform, the state~$\sigma$ has the following form:
\[
  \sigma= \Es{\vct{x}\in \Fp^n}\ket{\vct{x}}\bra{\vct{x}}_{\calA}\otimes\sigma_{\vct{x}}^{\calP_i\calB},
\]
where~$\Tr_\calB\sigma_{\vct{x}}^{\calP_i\calB}=\sigma^{\calP_i}=\rho$ is independent of~$\vct{x}$. We want to bound~$(1/2)\norm{\sigma_W-\sigma_M}_1$, where
\begin{align*}
  \sigma_W
  &=
  \Tr_{\calP_i}\biggl[
  \Es{\vct{x}\in \Fp^n}\ket{\vct{x}}\bra{\vct{x}}_{\calA}\otimes
  \sum_{a\in \Fp}\ket{a}\bra{a}_{\calC}\otimes
  (A_{\vct{x}}^a\otimes I_{\calB})\sigma_{\vct{x}}^{\calP_i\calB}(A_{\vct{x}}^a\otimes I_{\calB})
  \biggr], \\
  \sigma_M
  &=
  \Tr_{\calP_i}\biggl[
  \Es{\vct{x}\in \Fp^n}\ket{\vct{x}}\bra{\vct{x}}_{\calA}\otimes
  \sum_{a\in \Fp}\ket{a}\bra{a}_{\calC}\otimes
  (\sqrt{V_{\vct{x}}^a}\otimes I_{\calB})\sigma_{\vct{x}}^{\calP_i\calB}(\sqrt{V_{\vct{x}}^a}\otimes I_{\calB})
  \biggr],
\end{align*}
and $\calC$ denotes the register used for prover $P_i$'s answers. For~$\vct{x}\in \Fp^n$, define isometries~$U_{\vct{x}},V_{\vct{x}}\colon\calP_i\otimes\calB\to\calP_i\otimes\calB\otimes\calC$ by
\begin{align*}
  U_{\vct{x}}
  &=
  \sum_{a\in \Fp}A_{\vct{x}}^a\otimes I_\calB\otimes\ket{a}_\calC, \\
  V_{\vct{x}}
  &=
  \sum_{a\in \Fp}\sqrt{V_{\vct{x}}^a}\otimes I_\calB\otimes\ket{a}_\calC.
\end{align*}
Then,
\begin{align*}
  &
  \norm{\sigma_W-\sigma_M}_1 \\
  &\le
  \norm[\bigg]{
    \Es{\vct{x}\in \Fp^n}\ket{\vct{x}}\bra{\vct{x}}_{\calA}\otimes
    \sum_{a\in \Fp}\ket{a}\bra{a}_{\calC}\otimes
    \bigl(
      (A_{\vct{x}}^a\otimes I_{\calB})\sigma_{\vct{x}}^{\calP_i\calB}(A_{\vct{x}}^a\otimes I_{\calB})
      -
      (\sqrt{V_{\vct{x}}^a}\otimes I_{\calB})\sigma_{\vct{x}}^{\calP_i\calB}(\sqrt{V_{\vct{x}}^a}\otimes I_{\calB})
    \bigr)
  }_1 \\
  &\le
  \Es{\vct{x}\in \Fp^n}
  \norm[\bigg]{
    \sum_{a\in \Fp}\ket{a}\bra{a}_{\calC}\otimes
    \bigl(
      (A_{\vct{x}}^a\otimes I_{\calB})\sigma_{\vct{x}}^{\calP_i\calB}(A_{\vct{x}}^a\otimes I_{\calB})
      -
      (\sqrt{V_{\vct{x}}^a}\otimes I_{\calB})\sigma_{\vct{x}}^{\calP_i\calB}(\sqrt{V_{\vct{x}}^a}\otimes I_{\calB})
    \bigr)
  }_1 \\
  &\le
  2\Es{\vct{x}\in \Fp^n}
  \sqrt{\sum_{a\in \Fp}    \Tr\bigl((A_{\vct{x}}^a-\sqrt{V_{\vct{x}}^a})^2 \rho\bigr)}
  \\
  &\le
  2\sqrt{
    \Es{\vct{x}\in \Fp^n}\sum_{a\in \Fp}
    \Tr\bigl((A_{\vct{x}}^a-\sqrt{V_{\vct{x}}^a})^2 \rho\bigr)
  } \\
  &\le
  2\sqrt{C\varepsilon^{c}},
\end{align*}
where the third inequality is by Lemma~\ref{lemma:gentle}, the fourth is by convexity and the last by~\eqref{eq:mlclose}.
Therefore, we have that~$\abs{p_{i-1}-p_i}\le(1/2)\norm{\sigma_W-\sigma_M}_1\le\sqrt{C\varepsilon^c}$ as claimed.
\end{proof}

By the triangle inequality, Claim~\ref{claim:pi} implies that~$\abs{p_0-p_2}\le 2\sqrt{C\varepsilon^c}$,
and therefore
\[
  p_2
  \ge p_0-2\sqrt{C\varepsilon^c}
  \ge 1-\varepsilon/3-2\sqrt{C\varepsilon^c}
  \ge1-3\sqrt{C\varepsilon^c},
\]
where the last inequality uses~$c\le1$ and $C\geq 1$.

Note that when the provers using strategy~$S_2$ interact with~$V'$,
both provers~$P_1$ and~$P_2$ can be implemented
so that they measure the prior entanglement without looking at their questions.
Since~$P_3$ is the only prover who might measure the prior entanglement after looking at his question,
strategy~$S_2$ can be implemented using shared randomness alone.

If~$P_1$ and~$P_2$ choose different multilinear functions,
then the provers pass in the consistency test with probability at most~$n/\abs{\Fp}\le1/6$
by the Schwartz-Zippel lemma~\cite{Schwartz80JACM,Zippel79} (see Lemma~\ref{lem:sz} in Appendix~\ref{sec:aux-lemmas} for a statement).
In strategy~$S_2$, they pass in the consistency test with probability at least~$1-15\sqrt{C\varepsilon^c}$.
Therefore, provers~$P_1$ and~$P_2$ choose the same multilinear function with probability at least~$1-15\sqrt{C\varepsilon^c}/(1-1/6)=1-18\sqrt{C\varepsilon^c}$.
This implies that if an oracle chooses a multilinear function in the same way as prover~$P_1$ and uses it for the two queries,
the distribution on their answers will differ by at most~$18\sqrt{C\varepsilon^c}$ in statistical distance.
Therefore, this oracle (which always implements a multilinear function) together with prover~$P_3$
is accepted in the interactive proof system of Corollary~\ref{corollary:and-test}
with probability at least~$1-15\sqrt{C\varepsilon^c}-18\sqrt{C\varepsilon^c}=1-33\sqrt{C\varepsilon^c}$.

Because~$(r,n,f)$ is a no-instance of Problem~\ref{problem:ora-3-sat-arith}
and~$\abs{\Fp}=p>8q(m,d)$,
the acceptance probability in the interactive proof system of Corollary~\ref{corollary:and-test}
is less than~$3/4$.
Comparing this with the lower bound in the previous paragraph, we obtain
\[
  1-33\sqrt{C\varepsilon^c}<\frac34,
\]
which implies
\[
  \varepsilon>\frac{1}{(132^2\cdot C)^{1/c}},
\]
contradicting the definition~$\eps = n^{-2/c_0}$ as soon as~$n$ is large enough. 
Since we obtained this contradiction
from the assumption that there exists an entangled strategy with acceptance probability at least~$1-\varepsilon/5$,
we have proved the claimed soundness guarantee against entangled provers.


\section{The multilinearity game}\label{sec:multilin}

In this section we analyze the combination of the consistency test and the linearity test described in Section~\ref{sec:protocol} as a stand-alone game played between a referee and $r\geq 3$ players,  which we call the \emph{$r$-player multilinearity game in~$n$ variables over~$\Fp$}. The game is parametrized by an integer $n$ and a finite field $\Fp$ of arbitrary size $p=|\Fp|$ (which is not necessarily a prime), and it is performed with $r$ players $P_1,\ldots,P_{r}$ treated symmetrically. The referee performs either of the following two tests with probability $1/2$ each:
\begin{itemize}
\item
  \emph{Consistency test.}
  The referee chooses~$\vct{x}\in\Fp^n$ uniformly at random
  and sends the same question~$\vct{x}$ to all players~$P_1,\ldots,P_{r}$.
  He expects each player to answer with an element of $\Fp$,
  and accepts if and only if all the answers are equal.
\item
  \emph{Linearity test.}
  The referee chooses~$i\in\{1,\dots,n\}$, $\vct{x}\in\Fp^n$ and~$y_i\neq z_i\in\Fp\backslash\{x_i\}$ uniformly at random,
  and sets~$y_j=z_j=x_j$ for every~$j\in\{1,\dots,n\}\setminus\{i\}$.
  He sends~$\vct{x},\vct{y},\vct{z}$ to three out of the $r$ players chosen at random,
  receives $a,b,c\in\Fp$,
  and accepts if and only if
  \[
    \frac{b-a}{y_i-x_i}=\frac{c-b}{z_i-y_i}=\frac{c-a}{z_i-x_i}.
  \]
\end{itemize}

We now define explicitly what we mean by a \emph{strategy} for the players in the multilinearity game. 

\begin{definition}\label{def:multilin-strategy}
A \emph{strategy for the players in the $r$-player multilinearity game in~$n$ variables over~$\Fp$} is given by the following. Finite-dimensional Hilbert spaces $\calP_1,\ldots,\calP_r$, a state $\ket{\Psi}\in\calP_1\otimes\cdots\otimes \calP_r$, and for every $i\in [r]$ and $\vx\in\Fp^n$ a measurement $\{(A_i)_\vx^a\}_{a\in \Fp}$ on $\calP_i$. It is understood that, upon receiving question $\vct{x}_i \in \Fp^n$, player~$P_i$ measures register $P_i$ corresponding to his share of $\ket{\Psi}$ using the measurement $\{(A_i)_{\vx_i}^a\}_{a\in \Fp}$, sending the outcome $a$ back to the verifier as his answer.

We will say that a strategy is \emph{symmetric} if $\calP_1\simeq\cdots\simeq\calP_r$, $(A_1)_\vx^a = \cdots = (A_r)_\vx^a$ for every $\vx$ and $a$ (in which case we will simply call the resulting measurement $\{A_\vx^a\}$), and $\ket{\Psi}$ is invariant with respect to arbitrary permutation of the registers $P_1,\ldots,P_r$.

Finally, a strategy will be called \emph{projective} if all measurements $\{(A_i)_{\vx}^a\}_{a\in \Fp}$ are projective.
\end{definition}

In case a strategy is symmetric, we will often abuse notation and use the symbol $\rho$ to denote the reduced density of $\ket{\Psi}$ on any $\bigotimes_{i\in S} \calP_i$, for $S\subseteq[r]$, without specifying explicitly which registers are understood: by symmetry only the number of registers matters, and this will always be clear in context.

The main result of this section is the following. We refer to Section~\ref{sec:prelim-notation} for definitions of the quantities appearing in the theorem, and to Lemma~\ref{lem:symmetry1} for a proof that the symmetry assumption made in the theorem is without loss of generality. 

\begin{theorem}\label{thm:lintest} There exists universal constants $0<c_0<1$, $C_0>1$ such that the following holds. 
Let $(\ket{\Psi},\{A_{\vx}^a\}_a)$ be a permutation-invariant projective strategy for $r\geq 3$ players in the $r$-player multilinearity game in~$n$ variables over~$\Fp$ with success probability at least~$1-\eps/2$. Assume furthermore that $p=|\Fp|\geq n^4\eps^{-1/2}$ and $\eps \leq n^{-2/c_0}$. Then there exists a sub-measurement~$\{V^g\}_{g\in\ML(\Fp^n,\Fp)}$, indexed by multilinear $g:\Fp^{n}\to \Fp$, such that
\begin{enumerate}
\item $V$ is consistent with $A$: $\inc(V,A) \leq C_0\,\eps^{c_0}$,
\item $\Trho(V)\geq 1 - C_0\,\eps^{c_0}$.
\end{enumerate}
\end{theorem}

The two items in the conclusion of the theorem intuitively state the following. Suppose that one of the players in the multilinearity game was to receive a question $\vx\in\Fp^n$, measure his share of the entangled state $\ket{\Psi}$ according to the projective measurement $\big\{A_\vx^a\big\}$, and answer the outcome he obtains (as he would in the original game). Now, suppose further that another player, upon receiving the same question $\vx\in\Fp^n$, instead of measuring her own share of $\ket{\Psi}$ according to $\big\{A_\vx^a\big\}$, was to perform the measurement $\{V^g, \Id-V\}$, where $V=\sum_g V^g$ (which is independent of $\vx$!). If she obtains the last outcome then she aborts the experiment. If, however, she obtains an outcome $g\in\ML(\Fp^n,\Fp)$, then she answers her question $\vx$ with $g(\vx)$. Item~1. above states that, on average over the choice of $\vx$, the probability that both players eventually produce different outcomes (conditioned on the second player not aborting) is at most $O(\eps^{c_0})$. Item~2. guarantees that, in the hypothetical scenario we just described, the second player does not abort too often: the probability that she obtains the outcome ``$\Id-V''$ is at most $C_0\,\eps^{c_0}$.

We will show that Theorem~\ref{thm:lintest} implies Theorem~\ref{thm:lintestclose} in Section~\ref{sec:pflintestclose}, while Theorem~\ref{thm:lintest} will be proved in Section~\ref{sec:mainind}. In the following section we prove a weaker version of the multilinearity test, the ``linearity test'', which implies Theorem~\ref{thm:lintest} for $n=1$. 

\subsection{Preliminary analysis: the linearity test}

Let $(\ket{\Psi},\{A_{\vx}^a\}_a)$ be a symmetric projective strategy for the players in the multilinearity game, as defined in Definition~\ref{def:multilin-strategy}. The following relations translate the assumption that the players succeed in the consistency test with probability $1-\eps$, and in the linearity test with probability $1-\eps$. 
\begin{align}
&\quad\,\,\Es{\vx} \sum_{a}\, \Tr_\rho\bigl( A_\vx^a \otimes A_\vx^{a} \bigr) \, \geq\, 1-\eps, \label{eq:cons}\\
\forall i\in[n],\quad&\underset{x_i\neq x'_i\neq x''_i,\vx_{\neg i}}{\textsc{E}} \sum_{\frac{a'-a}{x'_i-x_i} = \frac{a''-a'}{x''_i-x'_i}= \frac{a''-a}{x''_i-x_i}} \Tr_\rho\bigl( A_{x_i,\vx_{\neg i}}^a \otimes A_{x'_i,\vx_{\neg i}}^{a'} \otimes A_{x''_i,\vx_{\neg i}}^{a''} \bigr)  \,\geq\, 1-n\eps \,\geq\, 1-\sqrt{\eps},\label{eq:lin}
\end{align}
where all expectations are taken under the uniform distribution over the sets in which their indices range, and the last inequality follows from our assumption that $n \leq \eps^{-c_0/2} \leq \eps^{-1/2}$. 

The following claim proves the ``linearity'' part of the multilinearity test, thereby establishing the base case for the induction that will be performed in Section~\ref{sec:mainind}. It also illustrates some of the key techniques, in terms of the manipulation of measurement operators, that will be used throughout the paper. (The interested reader may thus wish to gain good familiarity with the proof of the claim before moving on to later sections, in which proofs will not always be as detailed.)

\begin{claim}\label{claim:lines} Let $i\in[n]$, and $\eps \geq p^{-1}$. Suppose that $(\ket{\Psi},\{A_\vx^a\})$ is a (symmetric, projective) strategy passing the consistency test with probability at least $1-\eps$, and the linearity test in the $i$-th direction with probability at least $1-\sqrt{\eps}$. Then there exists a family of measurements $\big\{B_{\vx_{\neg i}}^{\ell}\big\}_{\ell\in \ML(\Fp,\Fp)}$ of arity $1$ such that
\beq\label{eq:a-b-close}
\Es{\vx} \sum_a \Big\| A_\vx^a - \sum_{\ell:\,\ell(x_i)=a} B_{\vx_{\neg i}}^{\ell} \Big\|_\rho^2 \,=\, O\bigl(\sqrt{\eps}\bigr). 
\eeq
\end{claim}

We will often use the notation $B_\vx^a := \sum_{\ell:\,\ell(x_i)=a} B_{\vx_{\neg i}}^\ell$, leaving the dependence on $i$ implicit. We note for future use that the bound~\eqref{eq:a-b-close} implies that
$$\cons(A,B)\,\geq\, 1-O(\sqrt{\eps})\qquad\text{and}\qquad \cons(B) \,\geq\, 1-O(\sqrt{\eps}).$$
These inequalities can be deduced directly from~\eqref{eq:a-b-close}, but they will also be apparent from the proof of Claim~\ref{claim:lines}, which we now give. 

\begin{proof}
For any $\ell\in\ML(\Fp,\Fp)$, define 
$$B_{\vx_{\neg i}}^{\ell} \,:=\, \Es{x_i\neq x'_i} \, A_{x_i,\vx_{\neg i}}^{\ell(x_i)} A_{x'_i,\vx_{\neg i}}^{\ell(x'_i)} A_{x_i,\vx_{\neg i}}^{\ell(x_i)}.$$
Then $\big\{B_{\vx_{\neg i}}^{\ell} \big\}_{\ell}$ is a well-defined measurement: each operator is non-negative, and since for fixed $x_i\neq x'_i$, as $\ell$ ranges over $\ML(\Fp^2,\Fp)$ both $\ell(x_i)$ and $\ell(x'_i)$ independently range over $\Fp$, they sum to $\sum_a (A_\vx^a)^2 = \Id$ since, by assumption, for every $\vx$ and $a$ the measurement operator $A_\vx^a$ is a projector. Using the definition of $\|\cdot\|_\rho$, we can expand
\begin{align}
\Es{\vx} \sum_a \Big\| A_\vx^a - \sum_{\ell:\ell(x_i)=a} B_{\vx_{\neg i}}^{\ell} \Big\|_\rho^2 &= \Es{\vx} \sum_a \Trho\big( A_\vx^a \big) + \Es{\vx}\sum_{\substack{\ell,\ell'\\\ell(x_i)=\ell'(x_i)}} \Trho\big( B_{\vx_{\neg i}}^{\ell}   B_{\vx_{\neg i}}^{\ell'}\big)\notag \\
&\qquad- 2\,\Es{\vx}\sum_{a,\ell:\,\ell(x_i)=a} \Trho\big( A_\vx^a  B_{\vx_{\neg i}}^{\ell} \big).\label{eq:lines-0}
\end{align}
We first lower bound the last term above. Applying Lemma~\ref{lem:consmu} from Appendix~\ref{sec:consistency-lemmas} with $T_\vx^h = A_\vx^a$ and $Z_\vx^h = B_\vx^a$, we get
\beq\label{eq:lines-0b}
\Big| \Es{\vx}\sum_{a,\ell:\,\ell(x_i)=a} \Trho\big( A_\vx^a  B_{\vx_{\neg i}}^{\ell} \big) - \Es{\vx}\sum_{a,\ell:\,\ell(x_i)=a} \Trho\big( A_\vx^a \otimes B_{\vx_{\neg i}}^{\ell} \big) \Big| \,=\, O\big(\inc(A)^{1/2}\big)\,=\, O\big(\sqrt{\eps}\big)
\eeq
by~\eqref{eq:cons}, hence it will suffice to show a lower bound on $\Es{\vx}\sum_{a,\ell:\,\ell(x_i)=a} \Tr_\rho\bigl( A_\vx^a \otimes B_{\vx_{\neg i}}^{\ell}\bigr)$.  Using the definition of $B_{\vx_{\neg i}}^{\ell}$, we have
\begin{align}
\Es{\vx}\sum_{a,\ell:\,\ell(x_i)=a} &\Tr_\rho\bigl( A_\vx^a \otimes B_{\vx_{\neg i}}^{\ell} \bigr)\notag\\ &= \Es{\vx,x'_i\neq x''_i}\sum_{a,\ell:\,\ell(x_i)=a} \Tr_\rho\bigl( A_\vx^a \otimes A_{x'_i,\vx_{\neg i}}^{\ell(x'_i)} A_{x''_i,\vx_{\neg i}}^{\ell(x''_i)} A_{x'_i,\vx_{\neg i}}^{\ell(x'_i)} \bigr)\notag\\
&= \Es{\vx,x'_i\neq x''_i}\sum_{a,\ell:\,\ell(x_i)=a} \sum_{a'} \Tr_\rho\bigl( A_\vx^a \otimes A_{x'_i,\vx_{\neg i}}^{\ell(x'_i)} A_{x''_i,\vx_{\neg i}}^{\ell(x''_i)} A_{x'_i,\vx_{\neg i}}^{\ell(x'_i)} \otimes A_{x'_i,\vx_{\neg i}}^{a'}\bigr)\notag\\
&\leq \Es{\vx,x'_i\neq x''_i}\sum_{a,\ell:\,\ell(x_i)=a}  \Tr_\rho\bigl( A_\vx^a \otimes A_{x'_i,\vx_{\neg i}}^{\ell(x'_i)} A_{x''_i,\vx_{\neg i}}^{\ell(x''_i)} A_{x'_i,\vx_{\neg i}}^{\ell(x'_i)} \otimes A_{x'_i,\vx_{\neg i}}^{\ell(x'_i)}\bigr) + \eps \notag\\
&\leq \Es{\vx,x'_i\neq x''_i}\sum_{a,\ell:\ell(x_i)=a}   \sum_{a'} \Tr_\rho\bigl( A_\vx^a \otimes A_{x'_i,\vx_{\neg i}}^{a'} A_{x''_i,\vx_{\neg i}}^{\ell(x''_i)} A_{x'_i,\vx_{\neg i}}^{a'} \otimes A_{x'_i,\vx_{\neg i}}^{\ell(x'_i)}\bigr) + \eps,\label{eq:lines-1}
\end{align}
where the first equality simply uses that the $A_{x'_i,\vx_{\neg i}}^{a'}$ sum to identity over $a'$, the first inequality uses~\eqref{eq:cons} on the last two registers (together with $A_\vx^a \leq {\Id}$), and the last is by positivity. Let $\sigma:=\rho^{(3)}$ be the reduced density of $\ket{\Psi}$ on any $3$ of the provers, and apply Claim~\ref{claim:gentle} to the POVM $\{A_\vx^a\}_a$ for every $\vx$. Eq.~\eqref{eq:cons} implies that this POVM is consistent, hence 
$$ \Es{\vx} \Big\| \sum_a \bigl( A_\vx^a \otimes {\Id}\bigr)\,\rho^{(2)}\,\bigl(A_\vx^a \otimes {\Id}\bigr) - \rho^{(2)} \Big\|_1 \,=\, O\bigl(\sqrt{\eps}\bigr), $$
where we used that the $A_\vx^a$ are projectors.
Hence
\begin{align*}
\Es{\vx,x'_i\neq x''_i}&\sum_{a,\ell:\ell(x_i)=a}   \Big|\sum_{a'} \Tr_\rho\bigl( A_\vx^a \otimes \bigl(A_{x'_i,\vx_{\neg i}}^{a'} A_{x''_i,\vx_{\neg i}}^{\ell(x''_i)} A_{x'_i,\vx_{\neg i}}^{a'} - A_{x''_i,\vx_{\neg i}}^{\ell(x''_i)}\bigr) \otimes A_{x'_i,\vx_{\neg i}}^{\ell(x'_i)}\bigr)\Big|\\
&= \Es{\vx,x'_i\neq x''_i}\sum_{a,\ell:\,\ell(x_i)=a} \Big|\Tr\Bigl( \bigl(A_\vx^a \otimes A_{x''_i,\vx_{\neg i}}^{\ell(x''_i)} \otimes  A_{x'_i,\vx_{\neg i}}^{\ell(x'_i)} \bigr) \cdot \\
&\hskip2cm \Bigl(\sum_{a'} \bigl( {\Id} \otimes A_{x'_i,\vx_{\neg i}}^{a'} \otimes {\Id} \bigr) \,\rho\,\bigl( {\Id} \otimes A_{x'_i,\vx_{\neg i}}^{a'} \otimes {\Id} \bigr) - \rho \Bigr)\Bigr)\Big|\\
& \leq \Es{\vx,x'_i} \Big\|  \sum_{a'} A_{x'_i,\vx_{\neg i}}^{a'} \rho  A_{x'_i,\vx_{\neg i}}^{a'} - \rho \Big\|_1 \,=\,O\bigl(\sqrt{\eps}\bigr),
\end{align*}
where for the inequality we used that for every $\vx$ and $x_i\neq x''_i$, $\sum_{a,\ell:\ell(x_i)=a} A_\vx^a \otimes A_{x''_i,\vx_{\neg i}}^{\ell(x''_i)} \otimes  A_{x'_i,\vx_{\neg i}}^{\ell(x'_i)} \leq {\Id}$, and monotonicity of the trace distance. Combining this last bound with~\eqref{eq:lines-1}, we obtain
\begin{align*}
\Es{\vx}\sum_{a,\ell:\,\ell(x_i)=a} \Tr_\rho\bigl( A_\vx^a \otimes B_{\vx_{\neg i}}^{\ell} \bigr) &= \Es{\vx,x'_i\neq x''_i}\sum_{a,\ell:\,\ell(x_i)=a} \Tr_\rho\bigl( A_\vx^a \otimes A_{x''_i,\vx_{\neg i}}^{\ell(x''_i)} \otimes A_{x'_i,\vx_{\neg i}}^{\ell(x'_i)}\bigr) + O\bigl( \sqrt{\eps}\bigr)\\
&= \Es{\vx,x'_i\neq x''_i}\sum_{\ell} \Tr_\rho\bigl( A_\vx^{\ell(x_i)} \otimes A_{x''_i,\vx_{\neg i}}^{\ell(x''_i)} \otimes A_{x'_i,\vx_{\neg i}}^{\ell(x'_i)}\bigr) + O\bigl( \sqrt{\eps}\bigr).
\end{align*}
If $x_i=x'_i$ or $x_i=x''_i$, the last summation above evaluates to $1$. Hence the expectation is at least as large as the probability that the $\{A_\vx^a\}$ pass the linearity test along the $i$-th coordinate, which is at least $1-\sqrt{\eps}$ by~\eqref{eq:lin}, hence 
$$\Es{\vx}\sum_{a,\ell:\,\ell(x_i)=a} \Tr_\rho\bigl( A_\vx^a \otimes B_{\vx_{\neg i}}^{\ell} \bigr) \,\geq 1-O\bigl(\sqrt{\eps}\bigr).$$
Combining this inequality with~\eqref{eq:lines-0b} and using that the first two terms in~\eqref{eq:lines-0} are at most $1$ each proves the claim. 
\end{proof}

\subsection{Proof of Theorem~\ref{thm:lintestclose}}\label{sec:pflintestclose}

In this section we show how Theorem~\ref{thm:lintestclose}, which is the result we need in order to analyze the overall protocol from Section~\ref{sec:protocol}, follows from Theorem~\ref{thm:lintest}. Theorem~\ref{thm:lintest} is proved in Section~\ref{sec:mainind}.

\begin{proof}[Proof of Theorem~\ref{thm:lintestclose}]
Let $\big\{V^g\big\}_{g\in\ML(\Fp^n,\Fp)}$ be the sub-measurement guaranteed by Theorem~\ref{thm:lintest}. Expanding
\begin{align}
 \Es{\vx} \sum_a \Tr_\rho\bigl((A_\vx^a - \sqrt{V_\vx^a})^2\bigr) &= \Es{\vx} \sum_a \Bigl(\Tr_\rho\bigl( (A_\vx^a)^2\bigr) + \Tr_\rho\bigl( V_\vx^a\bigr) - 2\,\Tr_\rho\bigl(A_\vx^a \sqrt{V_\vx^a} \bigr)\Bigr)\notag\\
 &\leq 2-2\,\Es{\vx}\sum_a\Tr_\rho\bigl(A_\vx^a \sqrt{V_\vx^a} \bigr),\label{eq:corlin-2}
 \end{align}
 it will suffice to show that this last expectation is close to $1$. By applying Lemma~\ref{lem:consmu} from Appendix~\ref{sec:consistency-lemmas} with $T_\vx^h = A_\vx^a$ and $Z_\vx^h = \sqrt{V_\vx^a}$ we obtain that
$$
\Big|\Es{\vx}\sum_a\Tr_\rho\bigl(A_\vx^a \sqrt{V_\vx^a} \bigr) - \Es{\vx}\sum_a\Tr_\rho\bigl(A_\vx^a \otimes \sqrt{V_\vx^a} \bigr) \Big| \,=\, O\big(\inc(A)^{1/2}\big)\,=\, O\big(\sqrt{\eps}\big)
$$
by~\eqref{eq:cons}. Hence to upper-bound the right-hand-side of~\eqref{eq:corlin-2} it suffices to lower-bound
\begin{align*}
 \Es{\vx}\sum_a\Tr_\rho\bigl(A_\vx^a\otimes \sqrt{V_\vx^a}\bigr) &\geq \Es{\vx}\sum_a\Tr_\rho\bigl(A_\vx^a\otimes V_\vx^a\bigr) \\
&\geq  1   - C_0\eps^{c_0} - \inc(V,A)  \\
&\geq 1-2C_0\eps^{c_0},
\end{align*}
where the second inequality uses item~2 from Theorem~\ref{thm:lintest} and the definition of $\inc(V,A)$, and the last inequality follows from item~1.
Combined with~\eqref{eq:corlin-2}, this proves Theorem~\ref{thm:lintestclose}.
\end{proof}


\section{Soundness analysis of the multilinearity game}\label{sec:mainind}

In this section we prove our main result on the analysis of the multilinearity game in the presence of entanglement between the provers, Theorem~\ref{thm:lintest}. The proof proceeds by induction, and the key inductive step is summed up in the following proposition. (We refer to section~\ref{sec:prelim-notation} for a definition of the quantities that appear in the proposition.)

\begin{proposition}\label{prop:mainprop} There exists a universal constant $0<c_1<1/2$ such that the following holds. 
Suppose that~$(\ket{\Psi},\{A_\vx^a\}_a)$ is a symmetric projective strategy for the players in the $3$-player multilinearity game in $n$ variables over $\Fp$ that is accepted with probability at least~$1-\varepsilon$ in both the linearity test and the consistency test, for some $\varepsilon>0$. Let $p:=|\Fp|$ and $\delta>0$, and assume that $n^{-8/c_1^2} \geq \delta\geq \sqrt{n}\eps^{1/8} \geq n p^{-1/4}$. Let $1\leq k\leq n-1$ and $T$ be a given family of sub-measurements of arity $k$ such that $\inc(T,A)\leq\delta$. Then there exists a family of sub-measurements $V$ of arity $k+1$ such that 
\begin{enumerate}
\item $\inc(V,A) = O(\eps^{c_1})$,
\item For any family of sub-measurements $P$ of arity at least $k+1$, 
$$\big|\cons(P,V) - \cons(P,T)\big| \,=\, O\big(\delta^{c_1}+\inc(P,A)^{1/2}\big),$$
\item For any family of sub-measurements $P$, of arbitrary arity, 
$$\big|\cons(P,V)-\cons(P,T)\big| \,=\, O\big(\delta^{c_1}+\big|\cons(T,T)-\Trho(T)\big|^{1/2}\big).$$
\end{enumerate}
\end{proposition}

We first show that Theorem~\ref{thm:lintest} follows from Proposition~\ref{prop:mainprop}.

\begin{proof}[Proof of Theorem~\ref{thm:lintest}]
Starting from $V_0 = A$, let $V_1,\ldots,V_n$ be the sequence of measurements of increasing arity $1,\ldots,n$ given by Proposition~\ref{prop:mainprop}. By item~1, for every $i\in[n]$ we have $\inc(V_i,A) \leq C_1\eps^{c_1}$ for some universal constant $C_1$. Applying item~2 to $P=V_i$ and $V=V_i,V_{i-1},\ldots,V_0$, an easy induction shows that
$$ \big|\cons(V_i,V_i) - \cons(V_i,A) \big| \,=\, O\big(i\,\big(\eps^{{c_1}^2} + \eps^{c_1/2}\big)\big).$$
Hence using item~1. and $\inc(V,A)+\cons(V,A) = \Trho(V)$, since $A$ is a complete family of measurements, we also get
$$ \big|\cons(V_i,V_i) - \Trho(V_i) \big| \,=\, O\big(i\,\eps^{{c_1}^2}\big),$$
where we used $c_1 < 1/2$. Applying item~3 with $P=A$, an immediate induction then gives
$$\big|\cons(V_n,A) - \cons(A,A) \big| \,=\, O\big(n\sqrt{n}\,\eps^{{c_1}^2/2}\big).$$
But $\cons(A,A) \geq 1-\eps$ by~\eqref{eq:cons}, and using $\Trho(V_n) = \cons(V_n,A)+\inc(V_n,A)$ once more the theorem is proved for an appropriate choice of the constants $c_0,C_0$.  
\end{proof}

The proof of Proposition~\ref{prop:mainprop} itself proceeds by induction, and is based on two lemmas. The first is a quantum analogue of the ``self-improvement lemma''~\cite[Lemma~5.10]{BabForLun91CC}. It shows that, if a family of sub-measurements $\{R_{\vx_{\geq k}}^g\}$ is weakly consistent with $\{A_\vx^a\}$, \emph{and} it passes the consistency and linearity tests with high probability, then there exists an ``improved'' family of sub-measurements $\{T_{\vx_{\geq k}}^g\}$ that are highly consistent with $\{A_\vx^a\}$. (Item~3 in the conclusion of the lemma is not ultimately needed, but is required to combine Lemma~\ref{lem:lemma1} with Lemma~\ref{lem:lemma2} in the proof of Proposition~\ref{prop:mainprop}.) 

\begin{lemma}[Self-improvement lemma]\label{lem:lemma1} Let $(\ket{\Psi},\{A_\vx^a\}_a)$ be a (symmetric, projective) strategy for $3$ players in the multilinearity game, and $n^{-8} \geq \delta\geq \sqrt{n}\eps^{1/8} \geq 1/p$ such that the following hold:
\begin{enumerate}
\item The strategy $(\ket{\Psi},\{A_\vx^a\}_a)$ is accepted with probability at least~$1-\varepsilon/2$ in the multilinearity game,
\item There exists a family of sub-measurements $R$ of arity $k$ such that $\inc(R,A)\leq \delta$. 
 \end{enumerate}
Then there exists a family of sub-measurements $T$ of arity $k$, together with, for every $\vx\in\Fp^n$, a family of matrices $\{\hat{S}_{\vx}^g\big\}_g$, indexed by $g\in\ML(\Fp^{k-1},\Fp)$, such that the following hold: 
\begin{enumerate}
\item\label{claim2:1} $\inc(T,A) = O(\eps^{1/16})$,
\item\label{claim2:3} For any family of sub-measurements $P$, of arbitrary arity, $\big|\cons(P,R)-\cons(P,T)\big| = O\big(\sqrt{\delta}\big)$,
\item\label{claim2:2} For every $\vx$ and $a$, $\sum_{g:g(\vx_{<k})=a} \hat{S}_\vx^g \big( \hat{S}_\vx^g\big)^\dagger \leq A_\vx^a$, and for every $\vx_{\geq k}$ and $g$, $T_{\vx_{\geq k}}^g = \big(\Es{\vx_{<k}} \hat{S}_\vx^g\big)\big(\Es{\vx_{<k}} \hat{S}_\vx^g\big)^\dagger$ and 
$$\Es{\vx} \sum_g\, \Big\| \hat{S}_\vx^g - \sqrt{T_{\vx_{\geq k}}^g} \Big\|_\rho^2 \,\leq\, \delta.$$
\end{enumerate}
\end{lemma}

The second lemma is an analogue of the ``pasting lemma''~\cite[Lemma~5.11]{BabForLun91CC}. It shows how, starting from a family of sub-measurements $T$ of arity $k$ that is consistent with $A$, one may construct a family of sub-measurements $V$ of increased arity $k+1$ that is still somewhat consistent with $A$, as expressed in item~1 below. Items~2 and~3 are important to ensure that the new sub-measurement $V$ is not ``too incomplete'', which would render item~1 trivial. 

\begin{lemma}[Pasting lemma]\label{lem:lemma2} There exists a universal constant $0<c_2<1$ such that the following holds. Let $\eps,\delta>0$ be such that $np^{-1}\leq \eps\leq \delta^2$. Let $(\ket{\Psi},\{A_\vx^a\}_a)$ be a (symmetric, projective) strategy for $3$ players that is accepted with probability at least~$1-\varepsilon/2$ in the multilinearity game. Let $1\leq k\leq n-1$ and $T$ a family of sub-measurements of arity $k$ such that $\inc(T,A)\leq \delta$, and $T$ satisfies item 3. in the conclusion of Lemma~\ref{lem:lemma1}. Then there exists a family of sub-measurements $V$ of arity $k+1$ such that
  \begin{enumerate}
  \item $V$ is consistent with $A$: $\inc(V,A) = O\big(\delta^{c_2}\big)$,
 \item For any family of sub-measurements $P$ of arity at least $k+1$, 
$$\big|\cons(P,V) - \cons(P,T)\big| \,=\, O(\delta^{c_2} + \inc(P,A)^{1/2}),$$
\item For any family of sub-measurements $P$, of arbitrary arity, 
$$\big|\cons(P,V)-\cons(P,T)\big| \,=\, O\big(\delta^{c_2}+\big|\cons(T,T)-\Tr_\rho(T)\big|^{1/2}\big).$$
  \end{enumerate}
  \end{lemma}

Proposition~\ref{prop:mainprop} follows almost immediately by combining the two lemmas.

\begin{proof}[Proof of Proposition~\ref{prop:mainprop}]
Let $T$ be the family of sub-measurements given in the statement of the proposition. First apply Lemma~\ref{lem:lemma2} to $T$, obtaining a family of sub-measurements $R$ (called $V$ in the lemma) of arity $k+1$ such that items~1,~2 and~3 in the conclusion of the lemma hold. Next apply Lemma~\ref{lem:lemma1} to $R$, obtaining a family of sub-measurements $V$ of arity $k+1$ (called $T$ in the lemma) such that items~1 and~2 hold, where given our assumption $\inc(T,A)\leq\delta$ and item~1 from Lemma~\ref{lem:lemma2} the bound in item~2 is $O(\delta^{c_2/2})$. Item~1 from Lemma~\ref{lem:lemma1} implies item~1 in the proposition (provided $c_1$ is chosen small enough), and item~2 (resp. item~3) follows from combining item~2 from Lemma~\ref{lem:lemma1} with item~2 (resp. item~3) from Lemma~\ref{lem:lemma2}. 
\end{proof}


\subsection{The self-improvement lemma}\label{sec:lemma1}

In this section we prove Lemma~\ref{lem:lemma1}. Before proceeding with the details, we give some intuition and a high-level overview of how we will proceed. 

Consider the following simplified situation in $n=2$ dimensions. Although we will eventually require $p$ to be a large power of $2$, for the purposes of this overview it is sufficient to think about the case $p=2$, so that the players' answers are simply bits. For every $\vx\in \Fp^2$ we are given a two-outcome projective measurement $(A_\vx^0,A_\vx^1)$: picture two orthogonal ``planes'' of dimension $d/2$ each, where $d$ is the dimension of either players' private space and can be arbitrarily large. Our goal is to find a global ``refinement'' of these planes: a single measurement $\{T^g\}$, with outcomes in the set of bilinear functions $g:\Fp^2\to\Fp$, such that at every $\vx$ the approximation $A_\vx^a \approx_\eps \sum_{g:\,g(\vx)=a} T^g$ holds.\footnote{At this point we are being vague as to how the approximation is measured --- it will eventually be expressed solely in terms of the consistency between the two measurements.} In order to achieve this, we make two additional assumptions:
\begin{enumerate}
\item There exists another measurement $\{R^g\}$ which achieves an approximation of weaker quality, up to some $\delta \gg \eps$, than the one we are looking for,
\item The $\{A_\vx^a\}$ are very close to \emph{linear}: for every axis-parallel line $(x_1,\cdot)$ (resp. $(\cdot,x_2)$) there is a measurement $\{B_{x_1}^\ell\}_\ell$ (resp. $\{B_{x_2}^\ell\}_\ell$) with outcomes in the set of linear functions $\ell:\Fp\to \Fp$ such that $A^a_{(x_1,x_2)} \approx_\eps \sum_{\ell:\,\ell(x_2)=a} B_{x_1}^\ell$ (resp. $A^a_{(x_1,x_2)} \approx_\eps \sum_{\ell:\,\ell(x_1)=a} B_{x_2}^\ell$).
\end{enumerate}
The goal is to use the high quality of the approximation along lines to improve the quality of the overall ``bilinear'' approximation. Let's trust that an ideal measurement $\{T^g\}$, achieving an approximation of order $\eps$, exists, and think of $\{R^g\}$ as an adversarially ``corrupted'' version of $\{T^g\}$. There are two main ways in which $\{T^g\}$ can be corrupted: the first is by applying an arbitrary (but not too large) rotation on the whole space. The second is by ``mislabeling'' some of the measurement elements: e.g.\ for some $g$, a subspace of the space on which the ideal operator $T^g$ projects could have been labeled as a subspace of $R^{g'}$ for some $g'\neq g$. Note that the first type of error is unique to the quantum setting, and did not arise in the setting of Babai et al.'s ``self-improvement'' lemma~\cite{BabForLun91CC}. Indeed,  while quantum measurements are subject to arbitrarily small perturbations that may add up over time, nothing short of flipping the output of a binary function will suffice to corrupt it.

We devise a procedure which recovers from the first type of perturbation, but not the second. This appears unavoidable: if some components of the measurement $\{R^g\}$ are mis-labeled (say by completely re-shuffling the part of each measurement element that falls in a small-dimensional subspace of the whole space), there is no generic way to recover the corresponding ideal measurement elements. This is the main reason why the measurements we construct ``shrink'' at every step of the induction, and we have to work with sub-measurements instead: any ``mislabeled'' portions of space will have to be ignored. Since we cannot recover from such errors, it is crucial that they do not add up to too much throughout the whole induction process. 

\medskip

To correct the first type of error, we introduce the following procedure:
\begin{enumerate}
\item For every $\vx$, find the measurement $\{S_\vx^g\}_g$ which is closest to $\{R^g\}$ while being \emph{perfectly} consistent with $\{A_\vx^a\}$: that is, $\sum_{g:g(\vx)=a} S_\vx^g = A_\vx^a$. This is possible only because the elements $S_\vx^g$ are allowed to depend on $\vx$. We define the $\{S_\vx^g\}$ as the optimum solution to a specific convex program (see~\eqref{eq:conv1} below). Intuitively, $S_\vx^g$ is obtained as the ``projection'' of $R^g$ on the subspace $A_\vx^{g(\vx)}$.  
\item  Show that $\{S_\vx^g\}_g$ in fact only depends on $\vx$ up to some error depending on $\eps$ only (and not $\delta$), so that defining $T^g := \Es{\vx} S_\vx^g$ leads to the consistent measurement we are looking for. 
\end{enumerate}
The second step is crucial: why would the $\{S_\vx^g\}$ be (almost) independent of $\vx$? Here the linearity relations satisfied by the $\{A_\vx^a\}$ come into play. Using the perfect consistency of $S$ and $A$, together with the linearity of $A$, we are able to conclude that the $\{S_\vx^g\}$ should not vary too much \emph{along any axis-parallel line}. That is, $S_{(x_1,x_2)}^g \approx_\eps S_{(x_1,x'_2)}^g$ for any $x_1$ and $x_2,x'_2$ (and similarly in the other direction). This step depends on the specific optimization problem that was introduced in order to define $\{S_\vx^g\}_g$ (see~\eqref{eq:conv1} below). 
This invariance along axis-parallel lines can then be combined with the (reasonably) good expansion properties of the hypercube to conclude that the $\{S_\vx^g\}$ are in fact globally invariant, leading to the ``corrected'' measurement $\{T^g\}$. (We note that the fact that invariance along axis-parallel lines implies global invariance was already used in~\cite{BabForLun91CC}.)

\medskip

We proceed with the details. In the following section we introduce the optimization procedure that is used to define the operators $\big\{S_{\vx}^g\big\}_g$. In Section~\ref{sec:inv-s} we show that the $\{S_\vx^g\}$ are close to being independent of $\vx$, leading to the definition of the family of sub-measurements $\{T_{\vx_{\geq k}}^g\}$. In Section~\ref{sec:pfl2} we show that $T$ satisfies the conclusions of Lemma~\ref{lem:lemma1}.

\subsubsection{A convex optimization problem}

Let $\big\{R_{\vx_{\geq k}}^g\big\}_g$ be the family of sub-measurements promised in the assumptions of Lemma~\ref{lem:lemma1}. Let $\{\hat{S}_\vx^g\}_{g}$, where $\vx\in\Fp^n$ and $g\in\ML(\Fp^{k-1},\Fp)$, be an optimal solution to the following convex optimization problem: 
\begin{center}
 \centerline{\underline{Convex program for self-improvement}}\vskip-4mm
\begin{align}
&\omega\,:=\,\min \, \Es{\vx} \sum_g \Big\|\hat{S}_\vx^g - \sqrt{R_{\vx_{\geq k}}^g}\Big\|_{\rho}^2\label{eq:conv1}\\
& \forall\vx,a,\, \sum_{g:g(\vx_{\leq k})=a}  \hat{S}_\vx^g (\hat{S}_\vx^g)^\dagger \leq A_{\vx}^a,\notag
\end{align}
\end{center}
where $\sqrt{R_{\vx_{\geq k}}^g}$ is the positive square root of $R_{\vx_{\geq k}}^g$. Let $S_\vx^g := \hat{S}_\vx^g \big(\hat{S}_\vx^g)^\dagger$.\footnote{We will usually use a hat, as in $\hat{S}$, to denote matrices which we think of as factorizations of positive semidefinite matrices, but are not necessarily positive themselves. In general, the relation between $\hat{X}$ and $X$ will always be that $X = \hat{X}\hat{X}^\dagger$.} Our first claim shows that the optimum of~\eqref{eq:conv1} is bounded as a function of the inconsistency of $R$ and $A$. 

\begin{claim}\label{claim:distdelta}
Suppose that the $\big\{R_{\vx_{\geq k}}^g\big\}_g$ satisfy the assumptions of Lemma~\ref{lem:lemma1}. Then the optimum $\omega$ of~\eqref{eq:conv1} is at most $\inc(A,R)+O\big(\sqrt{\eps}\big)$.
\end{claim}

\begin{proof} We construct a feasible solution achieving the claimed value. 
Let $\hat{S}_\vx^g := A_{\vx}^{g(\vx_{<k})} \sqrt{R_{\vx_{\geq k}}^g}$. Then by definition $\{\hat{S}_\vx^g\}$ is  a feasible solution to~\eqref{eq:conv1}. To upper-bound its value, we first evaluate 
\begin{align*}
\Es{\vx}\sum_g \Big(\Tr_\rho\big(\hat{S}_\vx^g \sqrt{R_{\vx_{\geq k}}^g}\big) - \Tr_\rho\big(R_{\vx_{\geq k}}^g\big)\Big) &= \Es{\vx}\sum_g \Tr_\rho\big(  \big(A_{\vx}^{g(\vx_{<k})}-\Id\big) R_{\vx_{\geq k}}^g\big)\\
&= \Es{\vx}\sum_a \Trho\Big( A_\vx^a \Big( \sum_{g:g(\vx)\neq a} R_{\vx_{\geq k}}^g\Big)\Big)\\
&= \Es{\vx}\sum_g \Trho\big( R_{\vx_{\geq k}}^g \otimes A_\vx^{g(\vx_{<k})}\big) + O\big(\inc(A)^{1/2}\big), 
\end{align*}
where the second equality uses that $\sum_a A_\vx^a = \Id$ for every $\vx$, and the last follows from an application of Lemma~\ref{lem:consmu}. A similar calculation shows that
\begin{align*}
\Es{\vx}\sum_g \Tr_\rho\big(\hat{S}_\vx^g \big(\hat{S}_\vx^g\big)^\dagger \big) &= \Es{\vx}\sum_g \Tr_\rho\big(R_{\vx_{\geq k}}^g \otimes A_\vx^{g(\vx_{<k})}\big) + O\big(\inc(A)^{1/2}\big).
\end{align*}
To conclude, expand $\big\|\hat{S}_\vx^g - \sqrt{R_{\vx_{\geq k}}^g}\big\|_{\rho}^2$ and use
$$ \Es{\vx}\sum_g \Trho\big( R_{\vx_{\geq k}}^g \otimes A_\vx^{g(\vx_{<k})}\big) = \Trho(R) - \inc(A,R)$$
by definition, together with the bound $\inc(A)\leq\eps$ from~\eqref{eq:cons}.
\end{proof}

\subsubsection{Constructing a family of sub-measurements independent of $\vx_{< k}$}\label{sec:inv-s}

As a first step in showing that any optimal solution to~\eqref{eq:conv1} must be close to one that does not depend on $\vx_{<k}$, we show that such an optimal solution must be close to another feasible solution which is furthermore close to being invariant along the direction of any axis-parallel line in direction $i<k$. Precisely, we have the following. 

\begin{claim}\label{claim:almostopt} Assume $p^{-1} \leq \eps$. 
For every $i<k$ there exists a feasible solution $\big\{ \hat{Z}_\vx^g \big\}_g$ to~\eqref{eq:conv1}, with objective value  at most $\omega + O\big(\eps^{1/4}\big)$, such that 
$$ \Es{\vx}\sum_g  \big\| \hat{Z}_\vx^g - \Es{x'_i} \hat{Z}_{\vx_{\neg i},x'_i}^g \big\|_\rho^2 \,=\, O\big(\sqrt{\eps}\big).$$
\end{claim}

\begin{proof}
Let $\{\hat{S}_\vx^g\}$ be an optimal solution to~\eqref{eq:conv1}, and for any $i<k$ let 
$$\hat{Y}_{\vx_{\neg i}}^g\,:=\, B_{\vx_{\neg i}}^{g_{|\ell_i(\vx)}} \big(\Es{x_i}\,\hat{S}_\vx^g\big),$$
where $\ell_i(\vx)$ is the line going through $\vx$ and parallel to the $i$-th axis, and $\{B_{\vx_{\neg i}}^\ell\}_{\ell}$ is the ``lines'' family of measurements introduced in Claim~\ref{claim:lines}. We first claim that the $\hat{Y}_{\vx_{\neg i}}^g$, while not strictly feasible, achieve an objective value in~\eqref{eq:conv1} of at most $\omega + O(\eps^{1/4})$. 

\medskip

Towards proving this, we first show that $B_\vx^{g(\vx_{\leq k})}\hat{S}_\vx^g$ is close to $\hat{S}_\vx^g$. Recall the definition of $B_\vx^a = \sum_{\ell:\, \ell(x_i)=a} B_{\vx_{\neg i}}^{\ell}$. Using the fact that, since $\{S_\vx^g\}$ is feasible, $A_\vx^{g(\vx_{\leq k})} \hat{S}_\vx^g = \hat{S}_\vx^g$, we get
\begin{align}
\Es{\vx}\sum_g \big\|B_\vx^{g(\vx_{\leq k})}\hat{S}_\vx^g - \hat{S}_\vx^g\big\|_\rho^2 
&= \Es{\vx}\sum_g \Trho\big(  \big(B_\vx^{g(\vx_{\leq k})}- A_\vx^{g(\vx_{\leq k})}\big) S_\vx^g  \big(B_\vx^{g(\vx_{\leq k})}- A_\vx^{g(\vx_{\leq k})}\big) \big)\notag\\
&\leq \Es{\vx}\sum_a \big\|B_\vx^{a}- A_\vx^{a}\big\|_\rho^2\notag\\
& = O\big(\sqrt{\eps}\big)\label{eq:almostopt-0}
\end{align}
by Claim~\ref{claim:lines}. Using the triangle inequality and convexity, the following (not necessarily feasible) operators
$$\tilde{Y}_{\vx_{\neg i}}^g\,:=\,\Es{x_i}\, B_\vx^{g(\vx_{\leq k})} \hat{S}_\vx^g$$
also achieve a value $\omega + O(\sqrt{\eps})$ in~\eqref{eq:conv1}. 

\medskip

Next we show that the $\tilde{Y}_{\vx_{\neg i}}^g$ are close to the $\hat{Y}_{\vx_{\neg i}}^g := B_{\vx_{\neg i}}^{g_{|\ell_i(\vx)}}\Es{x_i}\hat{S}_\vx^g$. From the definition,
\begin{align*}
\tilde{Y}_{\vx_{\neg i}}^g \,=\, B_{\vx_{\neg i}}^{g_{|\ell_i(\vx)}}\,\big(\Es{x_i}\,\hat{S}_\vx^g \big)+ \Es{x_i}\sum_{\substack{\ell:\,\ell(x_i)=g(\vx_{\leq k})\\ \ell\neq g_{|\ell_i(\vx)}} } B_{\vx_{\neg i}}^{\ell}\hat{S}_\vx^g.
\end{align*}
The norm of the second term can be expanded as follows:
\begin{align*}
\Es{\vx_{\neg i}} \sum_g&  \Big\| \Es{x_i}\sum_{\substack{\ell:\,\ell(x_i)=g(\vx_{\leq k})\\ \ell\neq g_{|\ell_i(\vx)}} } B_{\vx_{\neg i}}^{\ell}\hat{S}_\vx^g \Big\|_\rho^2\\
&=\Es{\vx_{\neg i}} \sum_g \Es{x_i,y_i} \sum_{\substack{\ell:\,\ell(x_i)=g(\vx_{\leq k})\\ \ell\neq g_{|\ell_i(\vx)}} } \sum_{\substack{\ell':\,\ell'(y_i)=g(\vx_{\leq k})\\ \ell'\neq g_{|\ell_i(\vx)}} } \Trho\big(B_{\vx_{\neg i}}^{\ell}\hat{S}_{\vx_{\neg i},x_i}^g(\hat{S}_{\vx_{\neg i},y_i}^g)^\dagger B_{\vx_{\neg i}}^{\ell'} \big).
\end{align*}
Eq.~\eqref{eq:consmu-2} from Lemma~\ref{lem:consmu} shows that the contribution of all terms such that $\ell\neq \ell'$ is at most $O\big(\sqrt{\inc(B)}\big) = O\big(\eps^{1/4}\big)$ by Claim~\ref{claim:lines}. But the only possibility for $\ell=\ell'$ is that also $x_i=y_i$, since two distinct linear functions on $\Fp$ intersect in at most one point. Hence we have that
\begin{align*}
\Es{\vx_{\neg i}} \sum_g \Big\| \Es{x_i}\sum_{\substack{\ell:\,\ell(x_i)=g(\vx_{\leq k})\\ \ell\neq g_{|\ell_i(\vx)}} } B_{\vx_{\neg i}}^{\ell}\hat{S}_\vx^g \Big\|_\rho^2&= \Es{\vx_{\neg i}} \sum_g \frac{1}{p} \Es{x_i}\sum_{\substack{\ell:\,\ell(x_i)=g(\vx_{\leq k})\\ \ell\neq g_{|\ell_i(\vx)}} }  \Trho\big(B_{\vx_{\neg i}}^{\ell}S_\vx^g B_{\vx_{\neg i}}^{\ell} \big) + O\big(\eps^{1/4}\big)\\
&\leq \frac{4}{p} + O\big(\eps^{1/4}\big).
\end{align*}
Given our assumption on $p$, this implies 
$$ \Es{\vx_{\neg i}} \sum_g \big\| \tilde{Y}_{\vx_{\neg i}}^g-\hat{Y}_{\vx_{\neg i}}^g\big\|_\rho^2 \,=\, O\big(\eps^{1/4}\big),$$
and hence the $\hat{Y}_{\vx_{\neg i}}^g$, while still not necessarily feasible, achieve an objective value in~\eqref{eq:conv1} of $\omega + O\big(\eps^{1/4}\big)$. 

\medskip

Finally, define $\hat{Z}_{\vx}^g:= A_\vx^{g(\vx_{\leq k})} B_{\vx_{\neg i}}^{g_{|\ell_i(\vx)}} \,\big(\Es{x_i}\,\hat{S}_\vx^g\big)$. Then the $\big\{\hat{Z}_{\vx}^g\big\}$ are feasible in~\eqref{eq:conv1}, and the fact that 
\beq\label{eq:almostopt-1}
\Es{\vx} \sum_g \big\| \hat{Z}_{\vx}^g - \hat{Y}_{\vx_{\neg i}}^g \big\|_\rho^2 \,=\, O\big(\sqrt{\eps}\big)
\eeq
follows from arguments similar to those used in the proof of Claim~\ref{claim:distdelta}. Hence the $\big\{\hat{Z}_{\vx}^g\big\}$ are a feasible solution to~\eqref{eq:conv1} with objective value at most $\omega + O\big(\eps^{1/4}\big)$. Finally, by convexity~\eqref{eq:almostopt-1} implies that 
$$\Es{\vx_{\neg i}} \sum_g \big\| \Es{x_i} \hat{Z}_{\vx}^g - \hat{Y}_{\vx_{\neg i}}^g \big\|_\rho^2 \,=\, O\big(\sqrt{\eps}\big),$$
which together with the triangle inequality and~\eqref{eq:almostopt-1} shows that the $\{\hat{Z}_\vx^g\}$ are close to their expectation on any axis-parallel line in the $i$-th direction, proving the claim.  
\end{proof}

Using convexity of $X\to \|X-A\|_{\rho}^2$ for fixed $A$, the following follows from Claims~\ref{claim:distdelta} and~\ref{claim:almostopt}.

\begin{claim}\label{claim:closebasis}
Let $\big\{ \hat{S}_\vx^g \big\}$ be an optimal solution to~\eqref{eq:conv1}. Then
$$\Es{\vx,i<k} \sum_g \| \hat{S}_\vx^g - \Es{x'_i}\hat{S}_{\vx_{\neg i}x'_i}^g\|_{\rho}^2 = O\big(\eps^{1/4}\big).$$ 
\end{claim}

\begin{proof} We show that the two solutions constructed to~\eqref{eq:conv1}, $\big\{\hat{S}_\vx^g\big\}$ and $\big\{\hat{Z}_\vx^g\big\}$ from Claim~\ref{claim:almostopt}, must be close:\footnote{Note that $\hat{Z}_\vx^g$ implicitly depends on $i$, and the following equation is measuring the distance on average over the $k-1$ different constructions of $\hat{Z}_\vx^g$ obtained for all $1\leq i<k$.}
 \beq\label{eq:cb-2}
\Es{\vx,i<k} \sum_g \big\|\hat{Z}_{\vx}^g - \hat{S}_{\vx}^g \big\|_\rho^2 \,=\, O\big( \eps^{1/4}\big).
\eeq
The claim then follows by using the triangle inequality to combine this bound with the fact, proved in Claim~\ref{claim:almostopt}, that the $\hat{Z}_\vx^g$ themselves are close to their expectation along any axis-parallel line in the $i$-th direction. Hence it suffices to prove~\eqref{eq:cb-2}. Since the feasible set of~\eqref{eq:conv1} is convex, for any $0\leq t\leq 1$ the elements $\{(1-t) \hat{S}_\vx^g + t \hat{Z}_\vx^g\}$ also constitute a feasible solution. By optimality of $\big\{\hat{S}_\vx^g\big\}$, the resulting objective value must be at least $\omega$: for every $0\leq t \leq 1$, 
\begin{align*}
\Es{\vx} \sum_g \Big\|\hat{S}_\vx^g - \sqrt{R_{\vx_{\geq k}}^g} \Big\|_\rho^2 &\leq \Es{\vx} \sum_g \Big\|(1-t)\hat{S}_\vx^g+t\hat{Z}_\vx^g - \sqrt{R_{\vx_{\geq k}}^g} \Big\|_\rho^2\\
&= t^2 \,\Es{\vx} \sum_g \Big\| \hat{Z}_{\vx}^g-\hat{S}_\vx^g  \Big\|_\rho^2 + \Es{\vx} \sum_g \Big\|\hat{S}_\vx^g - \sqrt{R_{\vx_{\geq k}}^g} \Big\|_\rho^2 \\
&\qquad +2\,t \,\Es{\vx} \sum_g \Trho\Big( \big( \hat{Z}_{\vx}^g-\hat{S}_\vx^g \big)\big(  \hat{S}_\vx^g - \sqrt{R_{\vx_{\geq k}}^g} \big)^\dagger\Big).
\end{align*}
Using the known objective values, re-arranging and making $t\to 0$, we obtain that 
$$ \Es{\vx} \sum_g \Trho\Big( \big(\hat{Z}_{\vx}^g -  \hat{S}_\vx^g\big)\big( \sqrt{R_{\vx_{\geq k}}^g}  - \hat{S}_\vx^g\big)^\dagger\Big) = O\big(\eps^{1/4}\big).$$
Hence
\begin{align*}
 \Es{\vx} \sum_g \big\|\hat{S}_\vx^g - \hat{Z}_{\vx}^g\big\|_\rho^2 &=  \Es{\vx} \sum_g \Big( \Big\| \hat{Z}_{\vx}^g - \sqrt{R_{\vx_{\geq k}}^g}\Big\|_\rho^2 -  \Big\| \hat{S}_\vx^g - \sqrt{R_{\vx_{\geq k}}^g}\Big\|_\rho^2 \\
&\qquad\qquad+ 2\,\Trho\Big(\big( \hat{Z}_{\vx}^g -  \hat{S}_\vx^g\big)\big( \sqrt{R_{\vx_{\geq k}}^g}  - \hat{S}_\vx^g\big)^\dagger\Big)\Big)\\
&= O\big(\eps^{1/4}\big),
\end{align*}
proving~\eqref{eq:cb-2}. 
\end{proof}

Claim~\ref{claim:closebasis} shows that the $\{\hat{S}_\vx^g\}_g$ do not vary much along any axis-parallel line in the $i$-th direction. Using the expansion properties of the hypercube, we can deduce that the $\{\hat{S}_\vx^g\}_g$ are close (in the squared $\|\cdot\|_\rho$ norm) to a single operator, independent of the first $(k-1)$ coordinates.

\begin{claim}\label{claim:existz} For every $\vx_{\geq k}$ and $g$, let $\hat{T}_{\vx_{\geq k}}^g:= \Es{\vx_{<k}} \hat{S}_\vx^g$. Then 
$$\Es{\vx} \,\sum_g\, \big\|\hat{S}_\vx^g - \hat{T}_{\vx_{\geq k}}^g\big\|_{\rho}^2 \,=\, O\big(n\eps^{1/4}\big).$$ 
\end{claim}

\begin{proof} This is a direct consequence of the expansion properties of the hypercube, as expressed in Claim~\ref{claim:expand}. 
\end{proof}

\subsubsection{Proof of Lemma~\ref{lem:lemma2}}\label{sec:pfl2}

We conclude the proof of Lemma~\ref{lem:lemma2} by showing that the non-negative operators 
$$T_{\vx_{\geq k}}^g\,:=\, \hat{T}_{\vx_{\geq k}}^g\big(\hat{T}_{\vx_{\geq k}}^g\big)^\dagger,$$
where for any $\vx_{\geq k}$ and $g$ the matrix $\hat{T}_{\vx_{\geq k}}^g$ is defined in Claim~\ref{claim:existz} in the previous section, satisfy the conclusions of the lemma. First note that item~\ref{claim2:2} follows directly from Claim~\ref{claim:existz}, so it will suffice to verify that items~\ref{claim2:1} and~\ref{claim2:3} hold. Regarding item~\ref{claim2:1}, we can bound
\begin{align*}
\inc(T,A)&= \Es{\vx} \sum_{g,a\neq g(\vx_{<k})} \Trho\big( T_{\vx_{\geq k}}^g \otimes A_\vx^a \big) \\
&= \Es{\vx} \sum_{g,a\neq g(\vx_{<k})} \Trho\big( S_{\vx}^g \otimes A_\vx^a \big) + O\big(\sqrt{n}\eps^{1/8}\big)\\
&\leq \Es{\vx} \sum_{a\neq b} \Trho\big(A_\vx^a \otimes A_\vx^b\big)  + O\big(\sqrt{n}\eps^{1/8}\big)\\
&= O\big(\sqrt{n}\eps^{1/8}\big),
\end{align*}
where the second equality follows from Cauchy-Schwarz and Claim~\ref{claim:existz}, the inequality follows from the fact that the $\hat{S}_\vx^g$ are a feasible solution to~\eqref{eq:conv1}, and the last uses self-consistency of $A$ as in~\eqref{eq:cons}. 

Item~\ref{claim2:3} is proved in a similar way. Let $P$ be a family of sub-measurements of arity $\ell$, and assume that $\ell \leq k$, the other case being treated symmetrically. By definition,
\begin{align*}
\big|\cons(P,T)-\cons(P,R)\big|
&= \Es{\vx} \sum_{f,g:\, g_{|\vx_{\ell\cdots k-1}}=f} \Trho\big( P_{\vx_{\geq l}}^f \otimes \big(T_{\vx_{\geq k}}^g-R_{\vx_{\geq k}}^g \big)\big)\\
&\leq \Big( \Es{\vx} \sum_{f,g:\, g_{|\vx_{\ell\cdots k-1}}=f} \Trho\big(P_{\vx_{\geq l}}^f \otimes \big( \hat{T}_{\vx_{\geq k}}^g-\sqrt{R_{\vx_{\geq k}}^g}\big)\big(\hat{T}_{\vx_{\geq k}}^g-\sqrt{R_{\vx_{\geq k}}^g}\big)^\dagger\big)\Big)^{1/2}\\
&\hskip0.7cm\cdot \Big(\Es{\vx} \sum_{f,g:\, g_{|\vx_{\ell\cdots k-1}}=f} \Trho\big(P_{\vx_{\geq l}}^f \otimes \big( \hat{T}_{\vx_{\geq k}}^g+\sqrt{R_{\vx_{\geq k}}^g}\big)\big(\hat{T}_{\vx_{\geq k}}^g+\sqrt{R_{\vx_{\geq k}}^g}\big)^\dagger \big)\Big)^{1/2}\\
&\leq \sqrt{2} \Big( \Es{\vx} \sum_g \big\|\hat{T}_{\vx_{\geq k}}^g-\sqrt{R_{\vx_{\geq k}}^g}\big\|_\rho^2 \Big)^{1/2}\\
&= O\big( \inc(R,A)^{1/2} + \sqrt{n}\eps^{1/8}\big),
\end{align*}
where the first inequality is by Cauchy-Schwarz, the second uses that $\sum_f P_{\vx_{\geq l}}^f \leq \Id$ for every $\vx_{\geq l}$, and the last follows from the bounds proved in Claim~\ref{claim:distdelta} and Claim~\ref{claim:existz}.


\subsection{The pasting lemma}\label{sec:lemma2}

In this section we prove Lemma~\ref{lem:lemma2}. Let $T$ be the family of sub-measurements whose existence is promised in the lemma's assumptions. For every $\vx$, let $\big\{\hat{S}_{\vx}^h\big\}_h$ and $\big\{T_{\vx_{\geq k}}^h\big\}_h$ be as in item~\ref{claim2:2} of Lemma~\ref{lem:lemma1}. Let $\delta$ be such that 
\beq\label{eq:deltadef}
\max\Big\{ \inc(T,B),\,\inc(T,A),\,\Es{\vx}\sum_h \Big\|\hat{S}_\vx^h - \sqrt{T_{\vx_{\geq k}}^h}\Big\|_\rho^2\Big\} \,\leq\, \delta,
\eeq
where here $\{B_{\vx_{\neg k}}^\ell\}_\ell$ are the ``lines'' measurements in the $k$-th direction, as defined in Claim~\ref{claim:lines}. Note that Claim~\ref{claim:lines} implies that $\inc(T,B)\leq\inc(T,A)+O(\eps^{1/4})$, which justifies including $\inc(T,B)$ in~\eqref{eq:deltadef}.

Our goal is to define a new family of sub-measurements $V$, depending on one less coordinate of $\vx$ than $T$, but such that $V$ is still consistent with $A$, and moreover $V$ is not ``too small'', as measured by items~2 and~3 in the lemma. The main idea is to define $\{V_{\vx_{>k}}^g\}$ as (roughly) corresponding to the sequential application of $\{T_{\vx_{\geq k}}^h\}$ twice, for two random choices of $x_k$. This will produce two $(k-1)$-multilinear functions $h$ and $h'$, from which a $k$-multilinear function $g$ can be recovered by interpolation. This is essentially the same method as was used to define the ``line'' operators $B$ from the ``point'' operators $A$ in Claim~\ref{claim:lines}. Here the main additional difficulty is that we are starting with a family of sub-measurements, instead of complete, projective measurements as was the case in Claim~\ref{claim:lines}. 

This section is organized as follows. We start with some preliminary observations in Section~\ref{sec:pasting-preproc}. The family of sub-measurements $V$ is defined in Section~\ref{sec:pasting-def}. Item 1 in the conclusion of Lemma~\ref{lem:lemma2} is proved in Section~\ref{sec:pasting-cons}, and items 2 and 3 are proved in Section~\ref{sec:pasting-cons2}. 

\subsubsection{Pre-processing}\label{sec:pasting-preproc}

In this section we prove a preliminary claim, Claim~\ref{claim:tm-cons} below, which lets us modify the family of sub-measurements $T$ into another family $Q$ that has useful properties. The important property is item 3. in the claim, which establishes a form of commutation between $Q$ and the ``line'' measurements $B$. Intuitively, that such a property would hold for $Q$ equal to $T$ should follow from the consistency between the families of sub-measurements defined by $T$ and $B$: consistent measurements are ``compatible'', and by the gentle measurement lemma (cf. Lemma~\ref{lemma:gentle}) the order in which they are performed does not matter. However, we could not show directly that item 3 below holds for the family of sub-measurements $T$ itself; hence we need to modify it slightly. 

\begin{claim}\label{claim:tm-cons} 
Let $T$ be the family of sub-measurements satisfying the assumptions of Lemma~\ref{lem:lemma2}, and $\delta$ be as in~\eqref{eq:deltadef}. There exists a family of sub-measurements $\{Q_{\vx_{\geq k}}^h\}$ such that the following hold:
\begin{enumerate}
\item $\Trho(Q) \geq \Trho(T) - O(\delta^{c_4})$,
\item For every $\vx_{\geq k}$ and $h$, $Q_{\vx_{\geq k}}^h = B_{\vx_{\geq k}}^h \tilde{Q}_{\vx_{\geq k}}^h B_{\vx_{\geq k}}^h$ for some family of sub-measurements $\{\tilde{Q}_{\vx_{\geq k}}^h\}$ (and in particular $\inc(Q,A) = O(\eps^{1/2})$),
\item Let $Q_{\vx_{>k}} = \Es{x_k} \sum_h Q_{\vx_{\geq k}}^h$. For any $r\geq 1$,
$$ \Es{\vx} \Trho\Big( \Big( (Q_{\vx_{>k}})^{r} - \sum_\ell B_{\vx_{\neg k}}^{\ell} (Q_{\vx_{> k}})^r B_{\vx_{\neg k}}^{\ell}  \Big)^2\Big)\,=\, O\big( r^2\delta^{c_4}\big),$$
\end{enumerate}
where $c_4>0$ is a universal constant. 
\end{claim}

\begin{proof} For any $\vx_{<k}$ define a ``pinching'' map 
$$\mathcal{E}_{\vx_{<k}}:\, T_{\vx_{\geq k}}^h \,\mapsto\,  B_{\vx}^{h(\vx_{<k})}T_{\vx_{\geq k}}^h B_{\vx}^{h(\vx_{<k})}.$$
Note that $\mathcal{E}_{\vx_{<k}}$ also implicitly depends on $h(\vx_{<k})$, but this dependence will always be clear from the context. Let $\mathcal{E}(\cdot):= \Es{\vx_{<k}} \mathcal{E}_{\vx_{<k}}(\cdot)$. The idea for the definition of $Q$ consists in applying the map $\mathcal{E}$ to $T$ a certain number of times, leveraging a certain stability property that will follow after sufficiently many applications. 

Let $M$ be an integer to be fixed later, and for every $\vx_{\geq k}$ and $h$ let $R_{\vx_{\geq k}}^h := \mathcal{E}^M(T_{\vx_{\geq k}}^h)$, where $\mathcal{E}^M$ denotes the sequential composition of $\mathcal{E}$ with itself $M$ times. Using the Schwarz-Zippel lemma (Lemma~\ref{lem:sz}) it is not hard to verify that, as long as $M\geq 1$, $\inc(R,B) = O(\inc(B,B)+n/p) = O(\eps^{1/2})$. The proof of Claim~\ref{claim:tm-cons} is based on the following sequence of facts.

\begin{fact}\label{fact:tm-f0} There is a choice of $M \leq \delta^{-1/4}$ for which the following holds:
$$ \Es{\vx} \sum_h \Trho\big( \big(R_{\vx_{\geq k}}^h - \mathcal{E}_{\vx_{<k}}(R_{\vx_{\geq k}}^h)\big)^2\big)\,=\,O(\delta^{1/4}).$$
\end{fact}

\begin{proof}
Let $J_1=\delta^{-1/4}$. The proof is based on the use of the potential function  
$$ \Phi_i \,:=\, \Es{\vx_{\geq k}}\sum_h \Trho\big( \mathcal{E}^{J_1-i}\big( (\mathcal{E}^i(T_{\vx_{\geq k}}^h))^2\big)\big),$$
defined for all $0\leq i\leq J_1$. Note that $\Phi_i$ is non-negative, always at most $1$, and by the pinching inequality $(\mathcal{E}(X))^2 \leq \mathcal{E}(X^2)$ for any positive semidefinite $X$, $\Phi_i$ is non-increasing with $i$. Let $i_1$ the smallest index $i$ for which it holds that 
\beq\label{eq:tm-f0-1} \Es{\vx}\sum_h \Trho\Big(\mathcal{E}^{J_1-i}\Big( \mathcal{E}_{\vx_{< k}}\big( \big(\mathcal{E}^{i-1}(T_{\vx_{\geq k}}^h)\big)^2 \big)-  \big(\mathcal{E}_{\vx_{< k}}(\mathcal{E}^{i-1}(T_{\vx_{\geq k}}^h))\big)^2\Big)\Big) \,\leq\,\delta^{1/4}.
\eeq
Using operator convexity of the square function, this inequality not being satisfied for some $i$ implies that $\Phi_{i-1}-\Phi_{i} > \delta^{1/4}$. Since this can happen for at most $\delta^{-1/4}$ indices $i$, an $0\leq i_1\leq \delta^{1/4}$ such that~\eqref{eq:tm-f0-1} is satisfied for $i=i_1$ must exist. Using self-consistency of $B$ $(J_1-i_1)$ times, and consistency of $T$ and $B$,~\eqref{eq:tm-f0-1} is seen to imply
$$ \Es{\vx} \sum_h \Trho\big( \big( \mathcal{E}^{i_1-1}(T_{\vx_{\geq k}}^h) - \mathcal{E}_{\vx_{<k}}\big(\mathcal{E}^{i_1-1}(T_{\vx_{\geq k}}^h)\big) \big)^2\big)\,\leq\, O\big(\delta^{-1/4}\inc(B)^{1/2}+\delta^{1/4}\big).$$
To conclude, we set $M:=i_1-1$ and use $\inc(B) = O(\eps^{1/2}) \leq \delta^{1/2}$.
\end{proof}

The following is a consequence of Fact~\ref{fact:tm-f0}.

\begin{fact}\label{fact:tm-f0b}
The following holds
$$\Es{\vx_{\neg k},x_k\neq y_k} \sum_{g,\ell:\, \ell\neq g_{|\vx_{<k}}} \Trho\big( B_{\vx_{\neg k}}^\ell R_{x_k\vx_{>k}}^{g_{|x_k}}R_{y_k\vx_{>k}}^{g_{|y_k}}R_{x_k\vx_{>k}}^{g_{|x_k}} B_{\vx_{\neg k}}^\ell\big) \,=\, O\big(\inc(R,B)^{1/2}+\delta^{1/8}\big).$$
\end{fact}

\begin{proof}
By definition of $R$,
\begin{align*}
\Es{\vx_{> k},x_k\neq y_k}\sum_g &\,\Trho\big(R_{x_k\vx_{>k}}^{g_{|x_k}}R_{y_k\vx_{>k}}^{g_{|y_k}}R_{x_k\vx_{>k}}^{g_{|x_k}}\big) \\
&= \Es{\vx_{\neg k},x_k\neq y_k}\sum_g \Trho\big(R_{x_k\vx_{>k}}^{g_{|x_k}} B_{\vx_{\neg k}y_k}^{g(\vx_{<k}y_k)}R_{y_k\vx_{>k}}^{g_{|y_k}} B_{\vx_{\neg k}y_k}^{g(\vx_{<k}y_k)}R_{x_k\vx_{>k}}^{g_{|x_k}}\big)\\
&= \Es{\vx_{\neg k},x_k\neq y_k}\sum_g \Trho\big( \big(B_{\vx_{\neg k}x_k}^{g(\vx_{\leq k})} R_{x_k\vx_{>k}}^{g_{|x_k}}B_{\vx_{\neg k}x_k}^{g(\vx_{\leq k})}\big)B_{\vx_{\neg k}y_k}^{g(\vx_{<k}y_k)}R_{y_k\vx_{>k}}^{g_{|y_k}} B_{\vx_{\neg k}y_k}^{g(\vx_{<k}y_k)}R_{x_k\vx_{>k}}^{g_{|x_k}} \big) + O\big(\delta^{1/8}\big), 
\end{align*}
where the second equality follows from Fact~\ref{fact:tm-f0}. Using that, by definition, for $x_k\neq y_k$, $B_{\vx_{\neg k}x_k}^{g(\vx_{< k}x_k)}B_{\vx_{\neg k}y_k}^{g(\vx_{<k}y_k)} = B_{\vx_{\neg k}}^{g_{|\vx_{< k}}}$ and consistency of $R$ and $B$, we get
\begin{align*}
\Es{\vx_{> k},x_k\neq y_k}\sum_g &\,\Trho\big(R_{x_k\vx_{>k}}^{g_{|x_k}}R_{y_k\vx_{>k}}^{g_{|y_k}}R_{x_k\vx_{>k}}^{g_{|x_k}}\big) \\
&= \Es{\vx_{\neg k},x_k\neq y_k}\sum_g \Trho\big( R_{x_k\vx_{>k}}^{g_{|x_k}} B_{\vx_{\neg k}}^{g_{|\vx_{< k}}} R_{y_k\vx_{>k}}^{g_{|y_k}} B_{\vx_{\neg k}}^{g_{|\vx_{< k}}} R_{x_k\vx_{>k}}^{g_{|x_k}} \big) + O\big(\inc(R,B)^{1/2}+\delta^{1/8}\big)\\
&= \Es{\vx_{\neg k},x_k\neq y_k}\sum_g \Trho\big( B_{\vx_{\neg k}}^{g_{|\vx_{< k}}} R_{x_k\vx_{>k}}^{g_{|x_k}} R_{y_k\vx_{>k}}^{g_{|y_k}} R_{x_k\vx_{>k}}^{g_{|x_k}}B_{\vx_{\neg k}}^{g_{|\vx_{< k}}} \big) + O\big(\inc(R,B)^{1/2}+\delta^{1/8}\big), 
\end{align*}
where the last equality again follows (after a little work) from Fact~\ref{fact:tm-f0}.
\end{proof}

We will also use the following.

\begin{fact}\label{fact:tm-f1a} Let $\{S_{\vx_{> k}}^g\}_g$ be an arbitrary family of sub-measurements and $\mu_2>0$. There exists an $i_2 \leq \mu_2^{-1}$ such that 
$$ \Es{\vx}\sum_g \Trho\Big(\big(\mathcal{E}_{\vx_{>k}}\big(\mathcal{E}^{i_2-1}(S_{\vx_{> k}}^g)\big) -  \mathcal{E}^{i_2-1}(S_{\vx_{> k}}^g) \big)^2\Big) \,=\, O\big(\mu_2 + \mu_2^{-1}\,\inc(B)^{1/2} \Big),$$
where here we denote $\mathcal{E}(S_{\vx_{> k}}^g) = \Es{\vx_{< k}} B_{\vx_{\neg k}}^{g_{|\vy_{<k}}}S_{\vx_{> k}}^g B_{\vx_{\neg k}}^{g_{|\vy_{<k}}}$. Moreover, for all $i\geq i_2$ it holds that 
$$ \Es{\vx}\sum_g \Trho\big(\big(\mathcal{E}^{i+1}(S_{\vx_{> k}}^g) -  \mathcal{E}^i(S_{\vx_{> k}}^g) \big)^2\big) \,=\, O\big(\mu_2 + i\,\inc(B)^{1/2} \big).$$
\end{fact}

\begin{proof}
The proof is very similar to that of Fact~\ref{fact:tm-f0}, and is based on the use of the potential function  
$$ \Phi_i \,:=\, \Es{\vx}\sum_g \Trho\big( \mathcal{E}^{J_2-i}\big( (\mathcal{E}^i(S_{\vx_{> k}}^g))^2\big)\big),$$
defined for all $0\leq i\leq J_2$, where $J_2=\mu_2^{-1}$. Note that $\Phi_i$ is always at most $1$, and by the pinching inequality $\mathcal{E}(X)^2 \leq \mathcal{E}(X^2)$ for any positive semidefinite $X$, $\Phi_i$ is non-increasing with $i$. Let $i_2$ the smallest index such that $\Phi_{i_2-1}-\Phi_{i_2} \leq \mu_2$; as long as $J_2\geq \mu_2^{-1}$ such an $0\leq i_2\leq J_2$ must exist. By definition, it then holds that
$$ \Es{\vx}\sum_g \Trho\Big(\mathcal{E}^{J_2-i_2}\Big( \mathcal{E}_{\vx_{<k}}\big( \big(\mathcal{E}^{i_2-1}(S_{\vx_{> k}}^g)\big)^2 \big)-  \big(\mathcal{E}_{\vx_{<k}}(\mathcal{E}^{i_2-1}(S_{\vx_{> k}}^g))\big)^2\Big)\Big) \,\leq\,\mu_2.$$
Using self-consistency of $B$ $(J_2-i_2)$ times, we obtain 
$$ \Es{\vx}\sum_g \Trho\Big( \big(\mathcal{E}_{\vx_{<k}}\big(\mathcal{E}^{i_2-1}(S_{\vx_{> k}}^g)\big) -  \mathcal{E}^{i_2-1}(S_{\vx_{> k}}^g)\big)^2\Big) \,\leq\,\mu_2+O(J_2\,\inc(B)^{1/2}).$$
To conclude the proof, it suffices to use the operator convexity of the square function to move the expectation over $\vx_{<k}$ inside the square, and then observe that 
\begin{align*}
\Es{\vx}\sum_g \Trho\big(& \big( \mathcal{E}( \big(\mathcal{E}^{i_2-1}(S_{\vx_{> k}}^g)) - \mathcal{E}^2( \big(\mathcal{E}^{i_2-1}(S_{\vx_{> k}}^g))\big)^2\big)\\ 
&\leq \Es{\vx}\sum_g \Trho\big( \mathcal{E}(\big(  \big(\mathcal{E}^{i_2-1}(S_{\vx_{> k}}^g)) - \mathcal{E}( \big(\mathcal{E}^{i_2-1}(S_{\vx_{> k}}^g))\big)^2\big)\\
&\leq \Es{\vx}\sum_g \Trho\big( \big(  \big(\mathcal{E}^{i_2-1}(S_{\vx_{> k}}^g) - \mathcal{E}( \big(\mathcal{E}^{i_2-1}(S_{\vx_{> k}}^g))\big)^2\big) + O\big(\inc(B)^{1/2}\big),
\end{align*}
again using self-consistency of $B$.
\end{proof}

Let $M'$ be an integer to be fixed later, and for every $\vx_{\geq k}$ and $h$ define
$$Q_{\vx_{\geq k}}^h := \big(\Es{\vx_{<k}} B_{\vx}^{h(\vx_{< k})}\big)^{M'}(R_{\vx_{\geq k}}^h)\big(\Es{\vx_{<k}} B_{\vx}^{h(\vx_{< k})}\big)^{M'}.$$
Observe that, as before, as long as $M'\geq 1$ it holds that $\inc(Q,B)=O(\inc(B,B)+n/p) = O(\eps^{1/2})$. For any $g\in\ML(\Fp^k,\Fp)$, let $Q_{\vx_{>k}}^g := \Es{x_k\neq y_k} Q_{x_k\vx_{\geq k}}^{g_{|x_k}}Q_{y_k\vx_{\geq k}}^{g_{|y_k}}$. 
The following implies item 3. in Claim~\ref{claim:tm-cons}: for any $r\geq 1$,
\beq\label{eq:tm-cons-ind} \Es{\vx} \Trho\Big( \Big( (Q_{\vx_{>k}})^{r} - \sum_g \big(B_{\vx_{\neg k}}^{g_{|\vx_{<k}}} Q_{\vx_{> k}}^g B_{\vx_{\neg k}}^{g_{|\vx_{<k}}} \big)^r \Big)^2\Big)\,=\, O\big( r^2\delta^{c_3}\big),
\eeq
where $c_3>0$ is a universal constant. Eq.~\eqref{eq:tm-cons-ind} is proved by induction on $r$. The case $r=1$ is stated in the following claim.

\begin{fact}\label{fact:tm-f3}
The following holds 
\beq\label{eq:tm-f3-1}
\Es{\vx} \Trho\Big( \Big( Q_{\vx_{>k}} - \sum_g B_{\vx_{\neg k}}^{g_{|\vx_{<k}}} Q_{\vx_{> k}}^g B_{\vx_{\neg k}}^{g_{|\vx_{<k}}}  \Big)^2\Big)\,=\, O\big(\delta^{1/16}+ M'\delta^{1/8}\big).
\eeq
\end{fact}

\begin{proof} 
Fact~\ref{fact:tm-f0b} implies that $\{Q_{\vx_{>k}}^g\}$ and $B$ are $O(\delta^{1/8})$-consistent, from which it follows that 
\begin{align*}
\Es{\vx}\sum_{g,g'} \Trho\big( B_{\vx_{\neg k}}^{g_{|\vx_{<k}}} Q_{\vx_{> k}}^g B_{\vx_{\neg k}}^{g_{|\vx_{<k}}}  B_{\vx_{\neg k}}^{g'_{|\vx_{<k}}} Q_{\vx_{> k}}^{g'} B_{\vx_{\neg k}}^{g'_{|\vx_{<k}}}  \big)
&= \Es{\vx}\sum_{g,g'} \Trho\big( Q_{\vx_{> k}}^g B_{\vx_{\neg k}}^{g_{|\vx_{<k}}}  B_{\vx_{\neg k}}^{g'_{|\vx_{<k}}} Q_{\vx_{> k}}^{g'}\big) + O(\delta^{1/8})\\
&= \Es{\vx}\sum_{g,g'} \Trho\big( Q_{\vx_{> k}}^g    Q_{\vx_{> k}}^{g'}\otimes B_{\vx_{\neg k}}^{g_{|\vx_{<k}}}\otimes B_{\vx_{\neg k}}^{g'_{|\vx_{<k}}}\big) + O(\delta^{1/16}+M'\delta^{1/8})\\
&= \Es{\vx} \Trho\big( (Q_{\vx_{> k}})^2\big) + O(\delta^{1/16}+M'\delta^{1/8}).
\end{align*}
Here the second equality follows by applying Fact~\ref{fact:tm-f1a} with $\mu_2 := \delta^{1/16}$ and $S_{\vx}^g$ chosen as $Q_{\vx_{> k}}^g$ to move the term $ B_{\vx_{\neg k}}^{g_{|\vx_{<k}}}$ on the outside, and holds as long as $M'\geq \mu_2^{-1} = \delta^{-1/16}$. (The third uses consistency of $Q$ and $B$.) Expanding out the square in~\eqref{eq:tm-f3-1}, all four terms can be related up to $O(M'\delta^{1/8})$ by using similar arguments. 
\end{proof}

The induction step required to prove Eq.~\eqref{eq:tm-cons-ind} uses arguments similar to that of the proof of Fact~\ref{fact:tm-f3}, and we leave the details to the reader. Once that equation is established, choosing $M' = \delta^{-1/16}$ item 3 in Claim~\ref{claim:tm-cons} follows. Items 1 and 2 in the claim are simple consequences of the definition of $Q$ from $R$, and of $R$ from $T$; again we omit the details.
\end{proof}

\subsubsection{Construction of the pasted family of sub-measurements}\label{sec:pasting-def}

In this section and for the remainder of the proof of Lemma~\ref{lem:lemma2} we rename the family of sub-measurements $\{Q_{\vx_{\geq k}}^h\}$ constructed in the previous section into $\{T_{\vx_{\geq k}}^h\}$. The only properties of that family that we will need are those stated in Claim~\ref{claim:tm-cons}. In order to define the pasted sub-measurements $V$, we first introduce a ``pseudo-inverse'' $\tilde{T}$ as follows. As usual, let $T_{\vx_{>k}} = \Es{x_k} \sum_h T_{\vx_{\geq k}}^h$ and $\eta>0$ a small parameter to be fixed later. Define
\beq\label{eq:def-tildet}
\tilde{T}_{\vx_{>k}}\,:=\, \Big(\sum_{r=0}^R  \big(\Id- T_{\vx_{>k}}\big)^r\Big)^{1/2},
\eeq
where $R := (10/\eta)\log(1/\eta)$ is chosen so that $T_{\vx_{>k}}(1-T_{\vx_{>k}}\tilde{T}_{\vx_{>k}}^2) \leq \eta \Id$ (note that, by definition, $\tilde{T}_{\vx_{>k}}$ commutes with $T_{\vx_{>k}}$). Expanding out the series in the definition of $\tilde{T}_{\vx_{>k}}$, Item 3 from Claim~\ref{claim:tm-cons} implies that the following equation holds:
\beq\label{eq:tilde-com}
\Es{\vx} \Trho\big( (\tilde{T}_{\vx_{>k}} - \sum_\ell B_{\vx_{\neg k}}^\ell \tilde{T}_{\vx_{>k}} B_{\vx_{\neg k}}^\ell)^2 \big)\,=\, O\big( (\delta/\eta)^{c_5}\big),
\eeq
where $c_5>0$ is a sufficiently small constant. For every $\vx_{>k}$ and $g\in\ML(\Fp^k,\Fp)$, define
$$ V_{\vx_{>k}}^g \,:=\, \Big(1+\frac{R}{p}\Big)^{-1}\Es{y_k} \tilde{T}_{\vx_{>k}}\Big(\Es{x_k\neq y_k} T_{x_k\vx_{>k}}^{g_{|x_k}}\Big) \tilde{T}_{\vx_{>k}} T_{y_k\vx_{>k}}^{g_{|y_k}} \tilde{T}_{\vx_{>k}}\Big(\Es{x_k\neq y_k} T_{x_k\vx_{>k}}^{g_{|x_k}}\Big) \tilde{T}_{\vx_{>k}}.$$
The scaling factor $(1+R/p)^{-1}$ is necessary to ensure that the $\{V_{\vx_{>k}}^g\}$ sum to at most identity. It induces an extra error term in all our estimates; however our choice of $\eta = \delta^{c'}$ for some $c'>0$ will ensure that this error term is of the same order as ones that already appear; for clarity in the remainder of this section we will neglect it.

\begin{claim} The $\big\{V_{\vx_{>k}}^g\big\}_g$ form a family of sub-measurements of arity $k+1$.
\end{claim}

\begin{proof}
It is clear that $V_{\vx_{>k}}^g \geq 0$ for every $g$. When the variable $g$ runs over $\ML(\Fp^{k},\Fp)$, for $x_k\neq y_k\in\Fp$ the restrictions $g_{|x_k}$ and $g_{|y_k}$ independently run over $\ML(\Fp^{k-1},\Fp)$. Hence, using convexity of the map $A\mapsto AXA^\dagger$ for any $A$ and $X\geq 0$,
\begin{align*}
\sum_g V_{\vx_{>k}}^g &\leq \Big(1+\frac{R}{p}\Big)^{-1} \frac{1}{p^2}\sum_{y_k\neq x_k} \sum_{h,h'} \tilde{T}_{\vx_{>k}}T_{x_k\vx_{>k}}^{h} \tilde{T}_{\vx_{>k}} T_{y_k\vx_{>k}}^{h'} \tilde{T}_{\vx_{>k}} T_{x_k\vx_{>k}}^{h}\tilde{T}_{\vx_{>k}}\\
&= \Big(1+\frac{R}{p}\Big)^{-1}\frac{1}{p}\sum_{x_k} \sum_{h} \tilde{T}_{\vx_{>k}}T_{x_k\vx_{>k}}^{h} \tilde{T}_{\vx_{>k}} T_{\vx_{>k}} \tilde{T}_{\vx_{>k}} T_{x_k\vx_{>k}}^{h}\tilde{T}_{\vx_{>k}} \\
&\qquad - \Big(1+\frac{R}{p}\Big)^{-1}\frac{1}{p^2}\sum_{x_k} \sum_{h} \tilde{T}_{\vx_{>k}}T_{x_k\vx_{>k}}^{h} \tilde{T}_{\vx_{>k}} T_{x_k\vx_{>k}} \tilde{T}_{\vx_{>k}} T_{x_k\vx_{>k}}^{h}\tilde{T}_{\vx_{>k}}\\
&\leq \Big(1+\frac{R}{p}\Big)^{-1}\Big(\Id+ \frac{R}{p}\Id\Big) \,\leq\,\Id,
\end{align*}
where to obtain the last line we used $(T_{x_k\vx_{>k}}^{h} )^2\leq T_{x_k\vx_{>k}}^{h} $ as well as $\tilde{T}_{\vx_{>k}} \leq R^{1/2}\Id$ and $\tilde{T}_{\vx_{>k}} T_{\vx_{>k}} \tilde{T}_{\vx_{>k}} \leq \Id$.
\end{proof}

\subsubsection{Consistency}\label{sec:pasting-cons}

In this section we show that the ``pasted'' sub-measurement $V$ is consistent with $A$, proving item~1 of Lemma~\ref{lem:lemma2}. It will be convenient to introduce the shorthand 
\beq\label{eq:def-ww}
W_{\vx_{\geq k}}^{h} \,:=\, \tilde{T}_{\vx_{>k}} T_{\vx_{\geq k}}^h \tilde{T}_{\vx_{>k}}.
\eeq
We also let $\delta_W:=\max(\delta,\inc(W,A))$. 

\begin{claim}\label{claim:t-s-cons} The following holds
$$ \inc(W,A)\,=\,\Es{\vx} \sum_{h,a\neq h(\vx_{<k})} \Trho\big( \tilde{T}_{\vx_{>k}} T_{\vx_{\geq k}}^h  \tilde{T}_{\vx_{>k}} \otimes A_{\vx}^a \big)\,=\,O\big((\delta/\eta)^{c_5}\big),$$
where $c_5>0$ is the constant that appears in~\eqref{eq:tilde-com}.
\end{claim}

\begin{proof} We have
\begin{align*}
\inc(W,A) &= \Es{\vx_{\geq k}} \sum_{h} \Trho\big( \tilde{T}_{\vx_{>k}} T_{\vx_{\geq k}}^h  \tilde{T}_{\vx_{>k}} \otimes (\Id -A_{\vx_{\geq k}}^h) \big)\\
&=  \Es{\vx} \sum_{h,a} \Trho\big( \tilde{T}_{\vx_{>k}}  B_{\vx_{\geq k}}^h R_{\vx_{\geq k}}^h B_{\vx_{\geq k}}^h \tilde{T}_{\vx_{>k}} \otimes (\Id -A_{\vx_{\geq k}}^h) \big)\\
& =\Es{\vx} \sum_{h} \Trho\big( B_{\vx_{\geq k}}^h\tilde{T}_{\vx_{>k}} R_{\vx_{\geq k}}^h \tilde{T}_{\vx_{>k}} B_{\vx_{\geq k}}^h \otimes (\Id -A_{\vx_{\geq k}}^h) \big) + O\big((\delta/\eta)^{c_5}\big)\\
&= \Es{\vx} \sum_{h} \Trho\big( \tilde{T}_{\vx_{>k}}R_{\vx_{\geq k}}^h\tilde{T}_{\vx_{>k}} \otimes (\Id -A_{\vx_{\geq k}}^h) \otimes A_{\vx_{\geq k}}^h \big) + O\big((\delta/\eta)^{c_5}\big)\\
&= O\big((\delta/\eta)^{c_5}\big),
\end{align*}
where the second equality follows from item 2 in Claim~\ref{claim:tm-cons} (and some sub-measurement $\{R_{\vx_{\geq k}}^h\}$), the third follows from~\eqref{eq:tilde-com}, the fourth uses Lemma~\ref{lem:consmu} together with consistency of $B$ and $A$ as in Claim~\ref{claim:lines}, and the last again follows from self-consistency of $A$, together with $\tilde{T}_{\vx_{>k}} \leq R^{1/2}\Id \leq \eta^{-1}\Id$ for small enough $\eta$. 
\end{proof}

\begin{claim}\label{claim:vcons} The family of sub-measurements $V$ is consistent with $A$: 
$$\inc(V,A)\,=\, O\big(\delta_W^{1/2}\big).$$
\end{claim}

\begin{proof}
By definition,
\begin{align*}
\inc(V,A) &= \Es{\vx,x'_k,x''_k\neq y_k} \sum_{g,a\neq g(\vx_{\leq k})} \Trho\big( W_{x'_k \vx_{>k}}^{g_{|x'_k}} T_{y_k\vx_{>k}}^{g_{|y_k}} W_{x''_k \vx_{>k}}^{g_{|x''_k}} \otimes A_\vx^a \big)\\
&= \Es{\vx,x'_k,x''_k\neq y_k} \sum_{g,a\neq g(\vx_{\leq k})} \Trho\big( W_{x'_k \vx_{>k}}^{g_{|x'_k}} T_{y_k\vx_{>k}}^{g_{|y_k}} W_{x''_k \vx_{>k}}^{g_{|x''_k}} \otimes A_\vx^a \otimes A_{\vx_{\neg k}x'_k}^{g(\vx_{<k},x'_k)}\big) + O(\delta_W^{1/2}) \\
&= \Es{\vx,x'_k,x''_k\neq y_k} \sum_{g,a\neq g(\vx_{\leq k})} \Trho\big( W_{x'_k \vx_{>k}}^{g_{|x'_k}} T_{y_k\vx_{>k}}^{g_{|y_k}} W_{x''_k \vx_{>k}}^{g_{|x''_k}}A_{\vx_{\neg k}x''_k}^{g(\vx_{<k},x''_k)} \otimes A_\vx^a \otimes A_{\vx_{\neg k}x'_k}^{g(\vx_{<k},x'_k)}\big) + O(\delta_W^{1/2}) \\
&= O(\eps^{1/2}+\delta_W^{1/2}),
\end{align*}
where the second and third equalities each follow from an application of Lemma~\ref{lem:cons1} and the definition of $\delta_W$, and the last follows by applying Cauchy-Schwarz and using linearity of $A$ in the $k$-th direction, as in~\eqref{eq:lin}.
\end{proof}

\subsubsection{Consistency with arbitrary sub-measurements}\label{sec:pasting-cons2}

We now show that items 2 and 3 in the conclusion of Lemma~\ref{lem:lemma2} hold. We will make use of the bound
\begin{align}\label{eq:ww-t}
\Es{\vx_{>k}} \Trho\big( \big( \Id - W_{\vx_{>k}} \big) T_{\vx_{>k}}\big( \Id - W_{\vx_{>k}} \big) \big) \,=\, O\big(\eta\big),
\end{align}
which holds by definition of $\{W_{\vx_{\geq k}}^{h}\}$ (cf.~\eqref{eq:def-ww}) and of $\tilde{T}_{\vx_{>k}}$ (cf.~\eqref{eq:def-tildet}).

\begin{claim}\label{claim:item2}
For any family of sub-measurements $P$ of arity at least $k+1$, 
$$\big|\cons(P,V) - \cons(P,T)\big| \,=\, O\big(\delta_W^{1/2} + \inc(P,A)^{1/2} + \eta^{1/2}\big).$$
\end{claim}

\begin{proof} Let $P$ be an arbitrary family of sub-measurements of arity $\ell \geq k+1$. We prove the claim in case $\ell = k+1$, the other cases being exactly similar. Then $P = \{P_{\vx_{>k}}^g\}_{g\in\ML(\Fp^{k},\Fp)}$, and by definition
\begin{align}
\cons(P,V) &= \Es{\vx} \sum_{g} \Trho\big(P_\vx^g \otimes V_\vx^g \big)\notag\\
&= \Es{\vx_{\neg k},x_k,x'_k\neq y_k} \sum_g \Trho\big( P_{\vx_{>k}}^g \otimes W_{x_k \vx_{>k}}^{g_{|x_k}} T_{y_k\vx_{>k}}^{g_{|y_k}} W_{x'_k \vx_{>k}}^{g_{|x'_k}} \big)\notag\\
&= \Es{\vx,x'_k,x''_k\neq y_k} \sum_g \Trho\big( P_{\vx_{>k}}^g \otimes W_{x'_k \vx_{>k}}^{g_{|x'_k}} T_{y_k\vx_{>k}}^{g_{|y_k}} W_{x''_k \vx_{>k}}^{g_{|x''_k}} \otimes A_{\vx_{\neg k} x'_k}^{g(\vx_{<k}x'_k)}  \big) + O\big(\delta_W^{1/2}\big),\label{eq:item2-1}
\end{align}
where the last equality follows from Lemma~\ref{lem:cons1} and the definition of $\delta_W$. We can then write
\begin{align}
&\Es{\vx,x'_k,x''_k\neq y_k} \sum_{g,h\neq g_{|x'_k}} \Trho\big( P_{\vx_{>k}}^g \otimes W_{x'_k \vx_{>k}}^{h} T_{y_k\vx_{>k}}^{g_{|y_k}} W_{x''_k \vx_{>k}}^{g_{|x''_k}} \otimes A_{\vx_{\neg k} x'_k}^{g(\vx_{<k}x'_k)}  \big) \notag\\
&= \Es{\vx,x'_k,x''_k\neq y_k} \sum_{g,h\neq g_{|x'_k}} \Trho\big( A_{\vx_{\neg k} x'_k}^{h(\vx_{<k}x'_k)} P_{\vx_{>k}}^g A_{\vx_{\neg k} x'_k}^{h(\vx_{<k}x'_k)} \otimes W_{x'_k \vx_{>k}}^{h} T_{y_k\vx_{>k}}^{g_{|y_k}} W_{x''_k \vx_{>k}}^{g_{|x''_k}} \otimes A_{\vx_{\neg k} x'_k}^{g(\vx_{<k}x'_k)}  \big)  + O\big(\delta_W^{1/2}\big)\notag\\
&= \Es{\vx,x'_k,x''_k\neq y_k} \sum_{\substack{h\neq g_{|x'_k}\\ h(\vx_{<k})= g(\vx_{<k}x'_k)}} \Trho\big( P_{\vx_{>k}}^g \otimes W_{x'_k \vx_{>k}}^{h} T_{y_k\vx_{>k}}^{g_{|y_k}} W_{x''_k \vx_{>k}}^{g_{|x''_k}} \otimes A_{\vx_{\neg k} x'_k}^{g(\vx_{<k}x'_k)}  \big)  + O\big(\eps^{1/2}+\delta_W^{1/2}\big),\label{eq:item2-2}
\end{align}
where the first equality again uses Lemma~\ref{lem:cons1} and the definition of $\delta_W$, and the last follows from an application of the Cauchy-Schwarz inequality and self-consistency of $A$. In the last expression, $h$ and $g_{|x'_k}$ are two distinct $(k-1)$-linear functions over $\Fp$: by the Schwartz-Zippel lemma (see Lemma~\ref{lem:sz} for a statement) they intersect in a fraction at most $O(k/|\Fp|)= O(k/p)$ points. Hence, applying the Cauchy-Schwarz inequality to recover a non-negative expression, we can upper bound~\eqref{eq:item2-2} by $O(\sqrt{n/p}+\eps^{1/2}+\delta_W^{1/2}) = O(\delta_W^{1/2})$ since $\delta_W \geq\delta\geq np^{-1}$. Together with~\eqref{eq:item2-1}, this shows that
\begin{align*}
\cons(P,V) &= \Es{\vx,x'_k,x''_k\neq y_k} \sum_g \Trho\big( P_{\vx_{>k}}^g \otimes W_{x'_k \vx_{>k}} T_{y_k\vx_{>k}}^{g_{|y_k}} W_{x''_k \vx_{>k}}^{g_{|x''_k}} \otimes A_{\vx_{\neg k} x'_k}^{g(\vx_{<k}x'_k)}   \big)  + O\big(\delta_W^{1/2}\big)\\
 &= \Es{\vx,x''_k\neq y_k} \sum_g \Trho\big( P_{\vx_{>k}}^g \otimes T_{y_k\vx_{>k}}^{g_{|y_k}} W_{x''_k \vx_{>k}}^{g_{|x''_k}} \otimes A_{\vx_{\neg k} x'_k}^{g(\vx_{<k}x'_k)}  \big)  + O\big(\delta_W^{1/2}+\eta^{1/2}\big),
\end{align*}
where the second equality follows from the Cauchy-Schwarz inequality and~\eqref{eq:ww-t}. Repeating the same steps for the remaining term $W_{x''_k \vx_{>k}}^{g_{|x''_k}} $, and using consistency of $P$ and $A$ to conclude, proves the claim.  
\end{proof}

\begin{claim}\label{claim:item3}
 For any sub-measurement $P$, of arbitrary arity, 
$$\big|\cons(P,V)-\cons(P,T)\big| \,=\, O\big(\big|\cons(T,T)-\Tr_\rho(T)\big|^{1/2}+\delta_W^{1/2}+\eta^{1/2}\big).$$
\end{claim} 

\begin{proof}
The proof closely follows that of Claim~\ref{claim:item2}, and we omit the details. 
\end{proof}

This concludes the proof of Lemma~\ref{lem:lemma2} provided $c_2$ is chosen to be a sufficiently small constant.

\appendix

\section{Auxiliary lemmas}\label{sec:aux-lemmas}

We first recall a key lemma in the analysis of low-degree polynomials over a finite field, the Schwartz-Zippel lemma~\cite{Schwartz80JACM,Zippel79}, which we state in a form that will be useful to us. 

\begin{lemma}[Schwartz-Zippel]\label{lem:sz}
Let $\Fp$ be a finite field, $n$ an integer, and $f:\Fp^n\to\Fp$ a non-zero multilinear function. Then $f$ has at most $sn|\Fp|^{n-1}$ zeros. 
\end{lemma}

The next series of claims are all based on variants of the Cauchy-Schwarz inequality. The first follows from Eq.~(3) of Bhatia and Davis~\cite{BhaDav95LAA} (see also~\cite{Bhatia88JOT}), substituting the norm~$\triplevert{\cdot}\triplevert$ by~$\|\cdot\|_1$.

\begin{theorem} \label{thm:matrixcsnorm}
  Let~$A$ and~$B$ be arbitrary matrices such that the product $A^\dagger B $ is well-defined. 
  Then,
  \[
    \big\| A^\dagger B\big\|_1 \leq \big\|A\big\|_F\,\big\|B\big\|_F.
  \]
\end{theorem}

Winter's gentle measurement lemma~\cite[Lemma~9]{Winter99IEEEIT} (see also Aaronson's ``almost as good as new'' lemma~\cite[Lemma~2.2]{Aar05}) is a key lemma formalizing the intuitive fact that if a measurement produces a certain outcome with near-certainty when performed on a specific state, then the post-measurement state is close to the original state. The following is a variant of that lemma, and we give a proof following Ogawa and Nagaoka~\cite[Appendix~C]{OgaNag07IEEEIT}.

\begin{lemma}\label{lemma:gentle}
  Let $\rho$ be a density operator on a Hilbert space~$\calH$,
  and $X$ and~$Y$ be linear operators from~$\calH$ to a Hilbert space~$\calK$
  such that~$X^*X\preceq I$ and~$Y^*Y\preceq I$.
  Then,
  \[
    \norm{X\rho X^*-Y\rho Y^*}_1
    \le
    2\sqrt{\Tr(X-Y)\rho(X-Y)^*}.
  \]
\end{lemma}

\begin{proof}
  By the triangle inequality,
  \[
    \norm{X\rho X^*-Y\rho Y^*}_1 \le \norm{(X-Y)\rho X^*}_1 + \norm{Y\rho(X-Y)^*}_1.
  \]
  By Theorem~\ref{thm:matrixcsnorm},
  \begin{align*}
    \norm{(X-Y)\rho X^*}_1
    &\le
    \norm{(X-Y)\sqrt{\rho}}_2 \norm{\sqrt{\rho}\,X^*}_2 \\
    &=
    \sqrt{\Tr(X-Y)\rho(X-Y)^*}\sqrt{\Tr X\rho X^*} \\
    &\le
    \sqrt{\Tr(X-Y)\rho(X-Y)^*}.
  \end{align*}
  Similarly, $\norm{Y\rho(X-Y)^*}_1\le\sqrt{\Tr(X-Y)\rho(X-Y)^*}$,
  and the lemma follows.
\end{proof}

We state the following two corollaries of Lemma~\ref{lemma:gentle}. 

\begin{claim}\label{claim:gentle2} Let $\{A_i\}$ and $\{B_i\}$ be two sets of positive matrices of the same dimension, and $\rho\geq 0$. Then
$$ \Big\| \sum_i \sqrt{A_i} \,\rho \,\sqrt{A_i} - \sqrt{B_i} \,\rho \,\sqrt{B_i} \Big\|_1 \,\leq\, 2 \Bigl( \sum_i \Tr\bigl( \bigl(\sqrt{A_i} - \sqrt{B_i}\bigr)^2 \rho \bigr)\Bigr)^{1/2}.$$
\end{claim}

\begin{proof} Let $X$ be a block-column matrix with blocks the $\sqrt{A_i}$, and similarly for $Y$ and the $\sqrt{B_i}$. Then
$$  \Big\| \sum_i \sqrt{A_i} \,\rho \,\sqrt{A_i} - \sqrt{B_i} \,\rho \,\sqrt{B_i} \Big\|_1 \,\leq \, \sum_i \Big\|  \sqrt{A_i} \,\rho \,\sqrt{A_i} - \sqrt{B_i} \,\rho \,\sqrt{B_i} \Big\|_1 \,\leq\, \big\| X\rho X^\dagger - Y\rho Y^\dagger \big\|_1,$$
and
$$\Tr\bigl((X-Y)\rho(X-Y)^\dagger\bigr) \,=\,\sum_i \Tr\bigl( \bigl(\sqrt{A_i} - \sqrt{B_i}\bigr)^2 \rho \bigr),$$
so that the claim follows from Lemma~\ref{lemma:gentle}.
\end{proof}

\begin{claim}\label{claim:gentle}
Let $\sigma\geq 0$ be a (possibly un-normalized) density matrix on $3$ registers, and suppose that $\sigma$ is invariant with respect to permutation of the first two registers. Let $\{A_i\}_i$ be a POVM on either of the first two registers, and let
$$ \delta:= \sum_{i\neq j} \Tr\bigl( (A_i\otimes A_j\otimes \Id) \sigma \bigr).$$
Then
$$ \big\| \sum_i \big(\sqrt{A_i}\otimes \Id\big)\,\Tr_2(\sigma)\,\big(\sqrt{A_i} \otimes \Id \big)- \Tr_2(\sigma) \big\|_1 \,=\,O(\sqrt{\delta}), $$
where here $\sqrt{A_i}$ acts on the first register of $\sigma$, and the identity on the third. 
\end{claim}

\begin{proof} First note that, $\{A_i\}_i$ being a POVM, 
$$\Tr_2 \bigl(\sum_i \big({\Id} \otimes \sqrt{A_i} \otimes \Id\big)\, \sigma \,\big({\Id} \otimes \sqrt{A_i} \otimes \Id\big) \bigr) \,=\, \Tr_2 \bigl(\sigma\bigr).$$
Hence by monotonicity of the trace norm 
\begin{align*}
 \big\| \sum_i \sqrt{A_i}\otimes \Id &(\Tr_2(\sigma)) \sqrt{A_i} \otimes \Id - \Tr_2(\sigma) \big\|_1\\
& \leq \big\| \sum_{i,j} \sqrt{A_i} \otimes \sqrt{A_j}\otimes \Id \sigma \sqrt{A_i} \otimes \sqrt{A_j}\otimes \Id - \sum_j {\Id}\otimes \sqrt{A_j}\otimes \Id \sigma {\Id} \otimes \sqrt{A_j}\otimes \Id \big\|_1\\
&\leq \big\| \sum_{i} \sqrt{A_i} \otimes \sqrt{A_i}\otimes \Id \sigma \sqrt{A_i} \otimes \sqrt{A_i}\otimes \Id - \sum_i {\Id}\otimes \sqrt{A_i} \otimes \Id\sigma {\Id} \otimes \sqrt{A_i}\otimes \Id \big\|_1 \\
&\qquad\qquad+ \sum_{i\neq j}\Tr\bigl( A_i\otimes A_j \otimes \Id\sigma\bigr)\\
&\leq 2 \sqrt{\sum_i \Tr\bigl( (\sqrt{A_i} \otimes \sqrt{A_i}\otimes \Id - {\Id} \otimes \sqrt{A_i}\otimes \Id)^2\sigma\bigr)} + \delta\\
&\leq  2\sqrt{\delta} + \delta
\end{align*}
where the second inequality is the triangle inequality, the third is by Claim~\ref{claim:gentle2}, and for the last we expanded
\begin{align*}
\sum_i \Tr\bigl(& (\sqrt{A_i} \otimes \sqrt{A_i}\otimes \Id - {\Id} \otimes \sqrt{A_i}\otimes \Id)^2\sigma\bigr)\\
 &= \sum_i \Big(\Tr\big(A_i\otimes A_i \otimes \Id\sigma\big) + \Tr\big(\Id\otimes A_i \otimes \Id\sigma\big) -2 \Tr\big( \sqrt{A_i} \otimes A_i\otimes \Id \sigma\big)\Big)\\
&\leq  \sum_i \Big(\Tr\big(A_i\otimes A_i \otimes \Id\sigma\big) + \Tr\big(\Id\otimes A_i \otimes \Id\sigma\big) -2 \Tr\big( A_i \otimes A_i \otimes \Id\sigma\big)\Big)\\
&= \delta,
\end{align*}
where for the inequality $\sqrt{A_i} \geq A_i$ follows from $0\leq A_i \leq \Id$ for every $i$, and the last equality uses the definition of $\delta$ and $\sum_i A_i=\Id$.  
\end{proof}

The following lemma follows from the standard expansion properties of the hypercube. Recall that for $\rho\geq 0$ and any $A$, $\|A\|_\rho^2 = \Tr\big(AA^\dagger \rho)$.

\begin{claim}[Expansion lemma]\label{claim:expand} Let $\eps>0$, $S$ a finite set of size $|S|=p$, $n,d$ integers and $A:S^n \to \C^{d\times d}$ such that for every $\vx\in S^n$, $0\leq A_\vx \leq {\Id}$, and 
$$\Es{i,\vx_{\neg i},x_i,x'_i} \big\| A_{\vx}-A_{\vx'}\big\|_\rho^2 \,\leq\, \eps,$$
where the expectation is taken with respect to the uniform distribution on $[n]\times S^{n-1}\times S\times S$. 
Then
$$\Es{\vx} \Big\| A_{\vx}-\Es{\vx} A_\vx\Big\|_\rho^2 \,\leq\, 2n\eps,$$
where both expectations are taken under the uniform distribution over $S^n$. 
\end{claim}

\begin{proof} Let $M:= \sum_{\vx,i,x'_i} \ket{\vx}\bra{\vx'}$ be the adjacency matrix of the hypercube $S^n$, $L:= {np\Id} - M$ the Laplacian, and $\tilde{L} = L\otimes \rho$. Let $A = \sum_{\vx} \ket{\vx}\otimes A_x$. Then 
\beq\label{eq:lapl1}
 A^\dagger \tilde{L}\cdot A = \frac{1}{2}\sum_{\vx,i,x'_i} (A_\vx - A_{\vx'})^\dagger \rho(A_\vx - A_{\vx'}).
\eeq
The normalized Laplacian $L/(np)$ has smallest eigenvalue $0$, and second smallest $\lambda_1 \geq 1/(2n)$. Let the smallest eigenvector of $L$ be $\ket{v_0} = p^{-n/2} \sum_\vx \ket{\vx}$, and write $A = \ket{v_0}\otimes A_0 + \ket{v_1} \otimes A_1$, where $\ket{v_1}$ is orthogonal to $\ket{v_0}$, and $A_0 = p^{-n/2}\sum_\vx A_\vx$. Then 
$$ A^\dagger \tilde{L} A \,=\,  \lambda_1 A_1^\dagger \rho A_1\, \geq \,\frac{1}{2n} A_1^\dagger \rho A_1.$$
Taking the trace and using the assumption made in the claim's statement together with~\eqref{eq:lapl1}, we get $\|A_1\|_\rho^2 \leq 2n\eps p^n$, and hence by definition of $A$,
$$ \Tr\bigl( (A-\ket{v_0}\otimes A_0)^\dagger ({\Id}\otimes\rho) (A-\ket{v_0}\otimes A_0)\bigr) \,=\, \|A_1\|_\rho^2\,\leq\, 2n\eps p^n,$$
which proves the claim. 
\end{proof}


\section{Lemmas about consistency}\label{sec:consistency-lemmas}

The following useful lemma relates the consistency of a measurement when performed on two separate subsystems of a permutation-invariant state with the possibility of exchanging the sub-system on which the measurement is performed. 
 Here $\rho$ is the reduced density of a permutation-invariant state.

\begin{lemma}\label{lem:cons1} Let $k\geq \ell \geq 1$ be two integers, $T$ a family of sub-measurements of arity $k$, and $V$ a family of sub-measurements of arity $\ell$. Let $\{Z_{\vx_{\geq k}}^h\}$ be such that $\Es{\vx} \sum_h Z_{\vx_{\geq k}}^h\big(Z_{\vx_{\geq k}}^h\big)^\dagger \leq \Id$. Then it holds that 
$$
\Big|\Es{\vx}\sum_h \Trho\big(Z_{\vx_{\geq k}}^h T_{\vx_{\geq k}}^h \otimes V_{\vx_{\geq \ell}} \big) - \Es{\vx} \sum_{g,h:\,h_{|x_{\ell},\ldots,x_{k-1}}=g} \Trho\big(Z_{\vx_{\geq k}}^h T_{\vx_{\geq k}}^h \otimes V_{\vx_{\geq \ell}}^g \big)\Big| \,\leq\, \sqrt{\inc(T,V)}.
$$
\end{lemma}

\begin{proof}
The proof is a direct consequence of the Cauchy-Schwarz inequality: write 
\begin{align*}
\Big|\Es{\vx}& \sum_h \Trho\big(Z_{\vx_{\geq k}}^h T_{\vx_{\geq k}}^h \otimes V_{\vx_{\geq \ell}} \big) - \Es{\vx} \sum_{g,h:\,h_{|x_{\ell},\ldots,x_{k-1}}=g} \Trho\big(Z_{\vx_{\geq k}}^h T_{\vx_{\geq k}}^h \otimes V_{\vx_{\geq \ell}}^g \big)\Big| \\
&= \Big|\Es{\vx}\sum_{g,h:\,h_{|x_{\ell},\ldots,x_{k-1}}\neq g}  \Trho\big(Z_{\vx_{\geq k}}^h T_{\vx_{\geq k}}^h \otimes V_{\vx_{\geq \ell}}^g \big)\Big|\\
&\leq \Big(\Es{\vx}\sum_{g,h:\,h_{|x_{\ell},\ldots,x_{k-1}}\neq g}  \Trho\big(T_{\vx_{\geq k}}^h \otimes V_{\vx_{\geq \ell}}^g \big)\Big)^{1/2}\Big(\Es{\vx}\sum_{g,h}  \Trho\big(Z_{\vx_{\geq k}}^h(Z_{\vx_{\geq k}}^h)^\dagger \otimes V_{\vx_{\geq \ell}}^g \big)\Big)^{1/2}\\
&\leq \sqrt{\inc(T,V)},
\end{align*}
where the last inequality follows from the definition of $\inc(T,V)$ and our assumption on $Z_{\vx_{\geq k}}^h$.
\end{proof}

\begin{lemma}\label{lem:consmu} Let $T$ be a family of sub-measurements of arity $k$, $X$ such that $X^\dagger X\leq \Id$, and $\{Z_{\vx_{\geq k}}^h\}$ such that $\Es{\vx} \sum_h Z_{\vx_{\geq k}}^h\big(Z_{\vx_{\geq k}}^h\big)^\dagger \leq \Id$ (for instance, a family of sub-measurements of arity $\ell$, for any $\ell$). Then\footnote{A special case of interest is when the measurements are \emph{complete}, in which case the statements simplify.}
\begin{align}
\Big| \Es{\vx}\sum_{h} \Trho(Z_{\vx_{\geq k}}^h T_{\vx_{\geq k}}^h \otimes T_{\vx_{\geq k}}) - \Es{\vx} \sum_h \Trho(Z_{\vx_{\geq k}}^h T_{\vx_{\geq k}}\otimes T_{\vx_{\geq k}}^h) \Big|\,\leq\,\sqrt{\inc(T,T)} \label{eq:consmu-1}\\
\Big| \Es{\vx}\sum_{h\neq h'} \Trho(T_{\vx_{\geq k}}^h X T_{\vx_{\geq k}}^{h'}\otimes T_{\vx_{\geq k}}) \Big|\,\leq\,2\,\sqrt{\inc(T,T)}\label{eq:consmu-2}
\end{align}
\end{lemma}

\begin{proof}
We first prove~\eqref{eq:consmu-1}. We have
\begin{align*}
\Big| \Es{\vx}\sum_{h} \Trho(Z_{\vx_{\geq \ell}}^h T_{\vx_{\geq k}}^h \otimes T_{\vx_{\geq k}})&- \Es{\vx} \sum_{h} \Trho(Z_{\vx_{\ell}}^h \otimes T_{\vx_{\geq k}}^h) \Big|\\
 &= \Big| \Es{\vx}\sum_{h} \Trho\big(Z_{\vx_{\geq \ell}}^h  ( T_{\vx_{\geq k}}^h\otimes T_{\vx_{\geq k}} - T_{\vx_{\geq k}}\otimes  T_{\vx_{\geq k}}^h) \big)\Big|\\
&\leq \Big( \sum_i \Trho\big( Z_{\vx_{k}}^h \big(Z_{\vx_{k}}^h\big)^\dagger \big) \Big)^{1/2} \Big( \sum_{h\neq h'} \Trho\big( ( T_{\vx_{\geq k}}^h\otimes  T_{\vx_{\geq k}}^{h'})^2  \big)\Big)^{1/2}\\
&\leq \sqrt{\inc(T,T)}, 
\end{align*}
where the second inequality follows from Cauchy-Schwarz.
Regarding~\eqref{eq:consmu-2}, we have 
\begin{align*} 
\Big| \Es{\vx}\sum_{h\neq h'} \Trho( T_{\vx_{\geq k}}^h X T_{\vx_{\geq k}}^{h'}\otimes T_{\vx_{\geq k}})  \Big| &= \Big|  \Es{\vx} \Trho\big( T_{\vx_{\geq k}}XT_{\vx_{\geq k}}\otimes T_{\vx_{\geq k}} \big) - \Es{\vx}\sum_h \Trho\big(T_{\vx_{\geq k}}^{h}X T_{\vx_{\geq k}}^{h}\otimes T_{\vx_{\geq k}}\big) \Big|
\end{align*}
From~\eqref{eq:consmu-1} we know that 
\begin{align*}
\Big| \Es{\vx}\sum_h \Trho\big(T_{\vx_{\geq k}}^hX T_{\vx_{\geq k}}^h\big)\otimes T_{\vx_{\geq k}} - \Es{\vx}\sum_h\Trho\big(T_{\vx_{\geq k}}^h X T_{\vx_{\geq k}}\otimes T_{\vx_{\geq k}}^h\big)\Big| &\leq \sqrt{\inc(T,T)}.
\end{align*}
The second term on the left-hand side satisfies
\begin{align*}
\Big|\Es{\vx}\sum_h \Trho&\big(T_{\vx_{\geq k}}^h X T_{\vx_{\geq k}}\otimes T_{\vx_{\geq k}}^h\big) - \Es{\vx}\sum_h  \Trho\big(T_{\vx_{\geq k}}XT_{\vx_{\geq k}} \otimes T_{\vx_{\geq k}}^h\big) \Big| \\
& \leq \Big( \Es{\vx}\sum_h \Trho\big(T_{\vx_{\geq k}}X^\dagger XT_{\vx_{\geq k}}\otimes  T_{\vx_{\geq k}}^h\big) \Big)^{1/2} \Big( \Es{\vx}\sum_h \Trho\big( (T_{\vx_{\geq k}}-T_{\vx_{\geq k}}^h)^2 \otimes T_{\vx_{\geq k}}^h \big) \Big)^{1/2}\\
&\leq \sqrt{\inc(T,T)},
\end{align*}
and this concludes the proof. 
\end{proof}

\section{Proof of Corollary~\ref{corollary:and-test}}
  \label{appendix:and-test}

In this section we give the proof of Corollary~\ref{corollary:and-test}. 
A standard method to convert multiple constraints to a single constraint involving
an exponential sum is by using small-bias probability spaces.

\begin{definition}[Small-bias probability space]
  Let~$n\in\N$.
  A set~$S\subseteq\Ftwo^n$ is called an \emph{$\varepsilon$-bias probability space}
  if for every~$\vct{c}\in\Ftwo^n\setminus\{0\}$,
  it holds that
  \[
    \big|
      \Pr_{\vct{\zeta}\in S}[\vct{c}\cdot\vct{\zeta}=0]
      -
      \Pr_{\vct{\zeta}\in S}[\vct{c}\cdot\vct{\zeta}=1]
    \big| \le \varepsilon.
  \]
\end{definition}

\begin{proposition} \label{proposition:small-bias}
  Let~$n\in\N$,
  and let~$S\subset\Ftwo^n$ be an~$\varepsilon$-bias probability space.
  Let~$\Fp$ be a finite field of characteristic two.
  If~$\vct{c}\in \Fp^n\setminus\{0\}$, then
  \[
    \Pr_{\vct{\zeta}\in S} \left[ \sum_{i=1}^n \zeta_i c_i = 0 \right]
    \le
    \frac{1+\varepsilon}{2}.
  \]
\end{proposition}

\begin{proof}
  If~$\Fp=\Ftwo$, then the proposition holds because
  \begin{align*}
    \Pr_{\vct{\zeta}\in S} \left[ \sum_{i=1}^n c_i\zeta_i = 0 \right]
    &=
    \frac12+\frac12\left(
      \Pr_{\vct{\zeta}\in S} \left[ \sum_{i=1}^n c_i\zeta_i = 0 \right]
      -
      \Pr_{\vct{\zeta}\in S} \left[ \sum_{i=1}^n c_i\zeta_i = 1 \right]
    \right)
    \\
    &\le
    \frac{1+\varepsilon}{2}.
  \end{align*}

  For general~$\Fp$, regard~$\Fp$ as a vector space over~$\Ftwo$,
  and let~$\{\alpha_1,\dots,\alpha_k\}$ be a basis of~$\Fp$ over~$\Ftwo$.
  Write~$\vct{c}$ as~$\vct{c}=\alpha_1\vct{c}^{(1)}+\dots+\alpha_k\vct{c}^{(k)}$,
  where~$\vct{c}^{(1)},\dots,\vct{c}^{(k)}\in\Ftwo^n$.
  Because~$\vct{c}\ne0$, we have that~$\vct{c}^{(j^*)}\ne0$ for some~$j^*$.
  By using the case of~$\Ftwo$, it holds that
  \[
    \Pr_{\vct{\zeta}\in S} \left[ \sum_{i=1}^n c^{(j^*)}_i\zeta_i = 0 \right]
    \le
    \frac{1+\varepsilon}{2}.
  \]
  Since~$\alpha_1,\dots,\alpha_k$ are linearly independent over~$\Ftwo$,
  $\sum_{i=1}^n c_i\zeta_i = 0$ implies $\sum_{i=1}^n c^{(j)}_i\zeta_i = 0$
  for all~$j$,
  and therefore in particular $\sum_{i=1}^n c^{(j^*)}_i\zeta_i = 0$.
  Therefore,
  \[
    \Pr_{\vct{\zeta}\in S} \left[ \sum_{i=1}^n c_i\zeta_i = 0 \right]
    \le
    \Pr_{\vct{\zeta}\in S} \left[ \sum_{i=1}^n c^{(j^*)}_i\zeta_i = 0 \right]
    \le
    \frac{1+\varepsilon}{2}.
    \qedhere
  \]
\end{proof}

\begin{theorem}[Alon, Goldreich, H\aa stad, and Peralta~\cite{AloGolHasPer92RSA}] \label{theorem:small-bias}
  There exist a constant~$c>0$ and a polynomial-time algorithm~$C$ which,
  given~$K,M\in\N$, $i\in\{1,\dots,K\}$ and~$j\in\{1,\dots,M\}$,
  outputs a~$C(K,M,i,j)\in\Ftwo$
  such that the set~$\{\vct{\zeta}^{(j)}\colon 1\le j\le M\}$
  defined by~$\vct{\zeta}^{(j)}=(C(K,M,1,j),\dots,\allowbreak C(K,M,K,j))$
  is an $(K/M^c)$-bias probability space in~$\Ftwo^K$.
\end{theorem}

By arithmetizing the Boolean circuit for~$C$
by using a similar idea to the proof of Proposition~4.2 of Ref.~\cite{BabForLun91CC},
we obtain the following corollary.

\begin{corollary} \label{corollary:small-bias-arith}
  There exist a constant~$c>0$ and a polynomial-time algorithm~$A$ which,
  given~$1^k$ and~$1^m$,
  outputs~$1^t$ and an arithmetic expression~$f(\vct{i},\vct{j},\vct{l})$
  in~$k+m+t$ variables
  such that the set~$\{\vct{\zeta}^{(\vct{j})}\colon\vct{j}\in\{0,1\}^m\}$
  defined by~$\vct{\zeta}^{(\vct{j})}=(\sum_{\vct{l}\in\{0,1\}^t}f(\vct{i},\vct{j},\vct{l}))_{\vct{i}\in\{0,1\}^k}$
  is an $2^{k-cm}$-bias probability space in~$\Ftwo^{2^k}$.
\end{corollary}

\begin{proof}[Proof of Corollary~\ref{corollary:and-test}]
  The protocol works as follows.
  The verifier first computes~$m=\ceil{(k+2)/c}$,
  where~$c$ is the constant in Corollary~\ref{corollary:small-bias-arith}.
  He runs the algorithm of Corollary~\ref{corollary:small-bias-arith}
  with parameters~$k$ and~$m$
  to obtain~$t\in\N$ and an arithmetic expression~$f(\vct{i},\vct{j},\vct{l})$
  in~$k+m+t$ variables.
  Let~$d'$ be the maximum degree of~$f$ in single variables.
  He chooses~$\vct{j}\in\{0,1\}^m$ uniformly at random,
  and sends~$\vct{j}$ to the prover.
  Then he simulates the protocol in Lemma~\ref{lemma:summation-test}
  with explicit inputs~$k+t$ and~$d+d'$
  and implicit input~$h_{\vct{j}}(\vct{i},\vct{l}):=f(\vct{i},\vct{j},\vct{l})h(\vct{i})$.

  For~$\vct{i}\in \Fp^k$, $\vct{j}\in \Fp^m$, and~$\vct{l}\in \Fp^t$,
  let~$\zeta^{(\vct{j})}_{\vct{i}}=\sum_{\vct{l}}f(\vct{i},\vct{j},\vct{l})\in \Fp$
  and~$\vct{\zeta}^{(\vct{j})}=(\zeta^{(\vct{j})}_{\vct{i}})_{\vct{i}\in\{0,1\}^k}\in \Fp^{2^k}$.
  Because~$m\ge(k+2)/c$,
  Corollary~\ref{corollary:small-bias-arith} guarantees
  that~$\{\vct{\zeta}^{(\vct{j})}\colon\vct{j}\in\{0,1\}^m\}$
  is a~$1/4$-bias probability space.

  Let~$c_{\vct{i}}=h(\vct{i})$.
  Then for all~$\vct{j}\in\{0,1\}^m$, it holds that
  \begin{equation}
    \sum_{\vct{i}\in\{0,1\}^k,\vct{l}\in\{0,1\}^t}h_{\vct{j}}(\vct{i},\vct{l})
    =
    \sum_{\vct{i}\in\{0,1\}^k}\zeta^{(\vct{j})}_{\vct{i}}c_{\vct{i}}.
    \label{eq:summation-test-multiple-1}
  \end{equation}

  Completeness:
  Suppose that~$c_{\vct{i}}=0$ for all~$\vct{i}\in\{0,1\}^k$.
  Then, by Eq.~(\ref{eq:summation-test-multiple-1}), it holds that
  \[
    \sum_{\vct{i}\in\{0,1\}^k,\vct{l}\in\{0,1\}^t}h_{\vct{j}}(\vct{i},\vct{l})=0
  \]
  for all~$\vct{j}\in\{0,1\}^m$.
  Therefore, the completeness of the protocol in Lemma~\ref{lemma:summation-test}
  implies that the protocol constructed above also has perfect completeness.

  Soundness:
  Suppose that~$\vct{c}\ne0$.
  By Proposition~\ref{proposition:small-bias},
  it holds that
  \[
    \Pr_{\vct{j}\in\{0,1\}^m} \Big[ \sum_{\vct{i}\in\{0,1\}^k} \zeta^{(\vct{j})}_{\vct{i}} c_{\vct{i}} = 0 \Big]
    \le
    \frac{1+1/4}{2}=\frac58.
  \]
  Eq.~(\ref{eq:summation-test-multiple-1})
  and the soundness in Lemma~\ref{lemma:summation-test}
  imply that
  for any~$\vct{j}\in\{0,1\}^m$ such that~$\sum_{\vct{i}\in\{0,1\}^k} \zeta^{(\vct{j})}_{\vct{i}} c_{\vct{i}} \ne 0$,
  the acceptance probability conditioned on the choice of~$\vct{j}$
  is at most~$(d+d')(k+t)/\abs{\Fp}$.
  Therefore, the overall acceptance probability is at most~$5/8+(d+d')(k+t)/\abs{\Fp}$.
  The corollary follows
  because~$d'$ and~$t$ are polynomially bounded in~$k$.
\end{proof}

\bibliography{mipnexp}

\begin{thebibliography}{KKMTV11}

\bibitem[Aar05]{Aar05}
Scott Aaronson.
\newblock Limitations of quantum advice and one-way communication.
\newblock {\em Theory of Computing}, 1(1):1--28, 2005.

\bibitem[AB09]{AroBar09}
Sanjeev Arora and Boaz Barak.
\newblock {\em Computational Complexity: A Modern Approach}.
\newblock Cambridge University Press, 2009.

\bibitem[AGHP92]{AloGolHasPer92RSA}
Noga Alon, Oded Goldreich, Johan H{\aa}stad, and Ren{\'{e}} Peralta.
\newblock Simple constructions of almost {$k$}-wise independent random
  variables.
\newblock {\em Random Structures and Algorithms}, 3(3):289--304, 1992.

\bibitem[ALMSS98]{AroLunMotSudSze98JACM}
Sanjeev Arora, Carsten Lund, Rajeev Motwani, Madhu Sudan, and Mario Szegedy.
\newblock Proof verification and the hardness of approximation problems.
\newblock {\em Journal of the ACM}, 45(3):501--555, 1998.

\bibitem[AS98]{AroSaf98JACM}
Sanjeev Arora and Shmuel Safra.
\newblock Probabilistic checking of proofs: A new characterization of {$\NP$}.
\newblock {\em Journal of the ACM}, 45(1):70--122, 1998.

\bibitem[BD95]{BhaDav95LAA}
Rajendra Bhatia and Chandler Davis.
\newblock A {Cauchy--Schwarz} inequality for operators with applications.
\newblock {\em Linear Algebra and its Applications}, 223--224:119--129, 1995.

\bibitem[Bel64]{Bell:64a}
John~S. Bell.
\newblock On the {E}instein-{P}odolsky-{R}osen paradox.
\newblock {\em Physics}, 1:195--200, 1964.

\bibitem[BFK10]{BFK10}
Anne Broadbent, Joseph Fitzsimons, and Elham Kashefi.
\newblock {QMIP} = {MIP}$^*$.
\newblock Technical report, {arXiv}:1004.1130v1 [quant-ph], 2010.

\bibitem[BFL91]{BabForLun91CC}
L{\'{a}}szl{\'{o}} Babai, Lance Fortnow, and Carsten Lund.
\newblock Non-deterministic exponential time has two-prover interactive
  protocols.
\newblock {\em Computational Complexity}, 1:3--40, 1991.

\bibitem[BGKW88]{BenGolKilWig88STOC}
Michael Ben-Or, Shafi Goldwasser, Joe Kilian, and Avi Wigderson.
\newblock Multi-prover interactive proofs: How to remove intractability
  assumptions.
\newblock In {\em Proceedings of the 20th Annual ACM Symposium on Theory of
  Computing (STOC)}, pages 113--131, 1988.

\bibitem[Bha88]{Bhatia88JOT}
Rajendra Bhatia.
\newblock Perturbation inequalities for the absolute value map in norm ideals
  of operators.
\newblock {\em Journal of Operator Theory}, 19(1):129--136, 1988.

\bibitem[BHP08]{BenHP08}
Michael Ben-Or, Avinatan Hassidim, and Haran Pilpel.
\newblock Quantum multi prover interactive proofs with communicating provers.
\newblock In {\em 49th Annual Symposium on Foundations of Computer Science
  (FOCS)}, pages 467--476, 2008.

\bibitem[BLR93]{BLR93}
Manuel Blum, Michael Luby, and Ronitt Rubinfeld.
\newblock Self-testing/correcting with applications to numerical problems.
\newblock {\em Journal of Computer and System Sciences}, 47(3):549--595, 1993.

\bibitem[CGJ09]{CleGavJai09QIC}
Richard Cleve, Dmitry Gavinsky, and Rahul Jain.
\newblock Entanglement-resistant two-prover interactive proof systems and
  non-adaptive {PIR}.
\newblock {\em Quantum Information and Computation}, 2009.

\bibitem[CHTW04]{CHTW04}
Richard Cleve, Peter H{\o}yer, Benjamin Toner, and John Watrous.
\newblock Consequences and limits of nonlocal strategies.
\newblock In {\em 19th Annual IEEE Conference on Computational Complexity
  (CCC)}, pages 236--249, 2004.

\bibitem[DLTW08]{DLTW08}
Andrew~C. Doherty, Yeong-Cherng Liang, Benjamin Toner, and Stephanie Wehner.
\newblock The quantum moment problem and bounds on entangled multi-prover
  games.
\newblock In {\em 23rd Annual IEEE Conference on Computational Complexity
  (CCC)}, pages 199--210, 2008.

\bibitem[FL92]{FeiLov92STOC}
Uriel Feige and L{\'{a}}szl{\'{o}} Lov{\'{a}}sz.
\newblock Two-prover one-round proof systems: Their power and their problems.
\newblock In {\em Proceedings of the 24th Annual ACM Symposium on Theory of
  Computing (STOC)}, pages 733--744, 1992.

\bibitem[Gol08]{Goldreich08}
Oded Goldreich.
\newblock {\em Computational Complexity: A Conceptual Perspective}.
\newblock Cambridge University Press, 2008.

\bibitem[H{\aa}s01]{Hastad01}
Johan H{\aa}stad.
\newblock Some optimal inapproximability results.
\newblock {\em Journal of the ACM}, 48:798--859, 2001.

\bibitem[IKM09]{ItoKM09}
Tsuyoshi Ito, Hirotada Kobayashi, and Keiji Matsumoto.
\newblock Oracularization and two-prover one-round interactive proofs against
  nonlocal strategies.
\newblock In {\em 24th Annual IEEE Conference on Computational Complexity
  (CCC)}, pages 217--228, 2009.

\bibitem[IKPSY08]{yao:tsirelson}
Tsuyoshi Ito, Hirotada Kobayashi, Daniel Preda, Xiaoming Sun, and Andrew C.-C.
  Yao.
\newblock Generalized {Tsirelson} inequalities, commuting-operator provers, and
  multi-prover interactive proof systems.
\newblock In {\em 23rd Annual IEEE Conference on Computational Complexity
  (CCC)}, pages 187--198, 2008.

\bibitem[Ito10]{Ito10}
Tsuyoshi Ito.
\newblock Polynomial-space approximation of no-signaling provers.
\newblock In {\em 37th international colloquium conference on Automata,
  languages and programming (ICALP)}, pages 140--151. Springer-Verlag, 2010.

\bibitem[Ito11]{Ito-parallel}
Tsuyoshi Ito.
\newblock Parallelization of entanglement-resistant multi-prover interactive
  proofs, 2011.
\newblock Submitted.

\bibitem[JJUW11]{JJUW11}
Rahul Jain, Zhengfeng Ji, Sarvaghya Upadhyay, and John Watrous.
\newblock {QIP} = {PSPACE}.
\newblock {\em Journal of the ACM}, 58(6):30:1--30:27, 2011.

\bibitem[JPPVW10]{JungePPVW10}
M.~Junge, C.~Palazuelos, D.~P{\'{e}}rez-Garc{\'{i}}a, I.~Villanueva, and M.~M.
  Wolf.
\newblock Operator space theory: A natural framework for {Bell} inequalities.
\newblock {\em Physical Review Letters}, 104:170405, 2010.

\bibitem[JUW09]{JUW09}
Rahul Jain, Sarvagya Upadhyay, and John Watrous.
\newblock Two-message quantum interactive proofs are in {PSPACE}.
\newblock In {\em 50th Annual Symposium on Foundations of Computer Science
  (FOCS)}, pages 534--543, 2009.

\bibitem[KKMTV11]{KKMTV11}
Julia Kempe, Hirotada Kobayashi, Keiji Matsumoto, Ben Toner, and Thomas Vidick.
\newblock Entangled games are hard to approximate.
\newblock {\em SIAM Journal on Computing}, 40(3):848--877, 2011.

\bibitem[KM03]{KobMat03JCSS}
Hirotada Kobayashi and Keiji Matsumoto.
\newblock Quantum multi-prover interactive proof systems with limited prior
  entanglement.
\newblock {\em Journal of Computer and System Sciences}, 66(3):429--450, 2003.

\bibitem[KRT10]{KRT10}
Julia Kempe, Oded Regev, and Ben Toner.
\newblock Unique games with entangled provers are easy.
\newblock {\em SIAM Journal on Computing}, 39(7):3207--3229, 2010.

\bibitem[KSV02]{KitSheVya02}
Alexei~Yu. Kitaev, Alexander~H. Shen, and Mikhail~N. Vyalyi.
\newblock {\em Classical and Quantum Computation}, volume~47 of {\em Graduate
  Studies in Mathematics}.
\newblock American Mathematical Society, 2002.

\bibitem[KV11]{KV11parallel}
Julia Kempe and Thomas Vidick.
\newblock Parallel repetition of entangled games.
\newblock In {\em Proceedings of the 43rd Annual ACM Symposium on the Theory of
  Computing (STOC)}, pages 353--362, 2011.

\bibitem[LFKN92]{LunForKarNis92JACM}
Carsten Lund, Lance Fortnow, Howard Karloff, and Noam Nisan.
\newblock Algebraic methods for interactive proof systems.
\newblock {\em Journal of the ACM}, 39:859--868, 1992.

\bibitem[Mer90]{MerminMS}
N.~David Mermin.
\newblock Simple unified form for the major no-hidden-variables theorems.
\newblock {\em Physical Review Letters}, 65:3373--3376, 1990.

\bibitem[NC01]{NieChu01}
Michael~A. Nielsen and Isaac~L. Chuang.
\newblock {\em Quantum Computation and Quantum Information}.
\newblock Cambridge University Press, 2001.

\bibitem[NPA08]{NPA08NJP}
Miguel Navascu{\'{e}}s, Stefano Pironio, and Antonio Ac{\'{\i}}n.
\newblock A convergent hierarchy of semidefinite programs characterizing the
  set of quantum correlations.
\newblock {\em New Journal of Physics}, 10(073013), 2008.

\bibitem[ON07]{OgaNag07IEEEIT}
Tomohiro Ogawa and Hiroshi Nagaoka.
\newblock Making good codes for classical-quantum channel coding via quantum
  hypothesis testing.
\newblock {\em IEEE Transactions on Information Theory}, 53(6):2261--2266,
  2007.

\bibitem[Per90]{PeresMS}
Asher Peres.
\newblock Incompatible results of quantum measurements.
\newblock {\em Physics Letters A}, 151(3-4):107--108, 1990.

\bibitem[PY86]{PapYan86IC}
Christos~H. Papadimitriou and Mihalis Yannakakis.
\newblock A note on succinct representations of graphs.
\newblock {\em Information and Control}, 71:181--185, 1986.

\bibitem[RT07]{RapTS07}
Alex Rapaport and Amnon Ta-Shma.
\newblock On the power of quantum, one round, two prover interactive proof
  systems.
\newblock {\em Quantum Information Processing}, 6:445--459, 2007.

\bibitem[RUV12]{ReiUngVaz-1209.0448}
Ben~W. Reichardt, Falk Unger, and Umesh Vazirani.
\newblock A classical leash for a quantum system: Command of quantum systems
  via rigidity of {CHSH} games.
\newblock Technical report, {arXiv}:1209.0448v1 [math-ph], 2012.

\bibitem[Sch80]{Schwartz80JACM}
Jacob~T. Schwartz.
\newblock Fast probabilistic algorithms for verification of polynomial
  identities.
\newblock {\em Journal of the ACM}, 27(4):707--717, 1980.

\bibitem[Sha92]{Sha92}
Adi Shamir.
\newblock {IP} = {PSPACE}.
\newblock {\em Journal of the ACM}, 39(4):869--877, 1992.

\bibitem[Sho90]{Shoup90MCOM}
Victor Shoup.
\newblock New algorithms for finding irreducible polynomials over finite
  fields.
\newblock {\em Mathematics of Computation}, 54(189):435--447, 1990.

\bibitem[SW08]{SW08}
Volkher~B. Scholz and Reinhard~F. Werner.
\newblock Tsirelson's problem.
\newblock Technical report, {arXiv}:0812.4305v1 [math-ph], 2008.

\bibitem[Tsi80]{Tsirelson80LMP}
Boris~S. Tsirelson.
\newblock Quantum generalizations of {Bell's} inequality.
\newblock {\em Letters in Mathematical Physics}, 4(2):93--100, 1980.

\bibitem[Wat09]{Watrous09-survey}
John Watrous.
\newblock Quantum computational complexity.
\newblock In Robert~A. Meyers, editor, {\em Encyclopedia of Complexity and
  System Science}. Springer, 2009.

\bibitem[Weh06]{Weh06STACS}
Stephanie Wehner.
\newblock Entanglement in interactive proof systems with binary answers.
\newblock In {\em 23rd Annual Symposium on Theoretical Aspects of Computer
  Science (STACS), Proceedings}, volume 3884 of {\em Lecture Notes in Computer
  Science}, pages 162--171, 2006.

\bibitem[Win99]{Winter99IEEEIT}
Andreas Winter.
\newblock Coding theorem and strong converse for quantum channels.
\newblock {\em IEEE Transactions on Information Theory}, 45(7):2481--2485,
  1999.

\bibitem[Zip79]{Zippel79}
Richard Zippel.
\newblock Probabilistic algorithms for sparse polynomials.
\newblock In {\em Symbolic and Algebraic Computation: An International
  Symposiumon on Symbolic and Algebraic Manipulation (EUROSM)}, volume~72 of
  {\em Lecture Notes in Computer Science}, pages 216--226, 1979.

\end{thebibliography}

\end{document}